\documentclass[final,10pt]{elsarticle}
 \usepackage{etoolbox}
\makeatletter
\patchcmd{\ps@pprintTitle}{\footnotesize\itshape
       Preprint submitted to \ifx\@journal\@empty Elsevier
      \else\@journal\fi\hfill\today}{\relax}{}{}
\makeatother
\makeatletter
\def\ps@pprintTitle{%
 \let\@oddhead\@empty
 \let\@evenhead\@empty
 \def\@oddfoot{\centerline{\thepage}}%
 \let\@evenfoot\@oddfoot}
\makeatother

\usepackage{amssymb}

\usepackage{amssymb}
\newcommand\Tau{\mathrm{T}}
\usepackage{amsmath}
\usepackage{color}

 \usepackage{amsthm}
 \newtheorem{thm}{Theorem}[section]
 
 \newtheorem{lem}{Lemma}[section]
 \newtheorem{cor}{Corollary}[section]
 \newtheorem{propp}{Property}[section]
 \newtheorem{prop}{Proposition}[section]
 \newdefinition{rmk}{Remark}[section]
  \newdefinition{defe}{Definition}[section]
 \newproof{pf}{Proof}
 \newproof{pot}{Proof of Theorem \ref{thm2}}

 \usepackage{caption}
\usepackage{subcaption}

\journal{Elsevier}

\begin{document}

\begin{frontmatter}



\title{Analysis vs Synthesis with Structure -- An Investigation of Union of Subspace Models on Graphs}


\author[]{M. S. Kotzagiannidis\corref{cor1}}
\ead{madeleine.kotzagiannidis@ed.ac.uk }
\cortext[cor1]{Corresponding author}

\author[]{M. E. Davies}
\ead{mike.davies@ed.ac.uk }

\address{Institute for Digital Communications, The University of Edinburgh, King's Buildings, Thomas Bayes Road, Edinburgh EH9 3FG, UK}

\fntext[]{This work was supported by the ERC project C-SENSE (ERC-ADG-2015-694888). MED is also supported by a Royal Society Wolfson Research Merit Award.}

\begin{abstract}
We consider the problem of characterizing  the `duality gap' between sparse synthesis- and cosparse analysis-driven signal models
through the lens of spectral graph theory, in an effort to comprehend their precise equivalencies and discrepancies. By detecting and exploiting the inherent connectivity structure, and hence, distinct set of properties, of rank-deficient graph difference matrices such as the graph Laplacian, we are able to substantiate discrepancies between the cosparse analysis and sparse synthesis models, according to which the former constitutes a constrained and translated instance of the latter. In view of a general union of subspaces model, we conduct a study of the associated subspaces and their composition, which further facilitates the refinement of specialized uniqueness and recovery guarantees, and discover an underlying structured sparsity model based on the graph incidence matrix. Furthermore, for circulant graphs, we provide an exact characterization of underlying subspaces by deriving closed-form expressions as well as demonstrating transitional properties between equivalence and non-equivalence for a parametric generalization of the graph Laplacian.
\end{abstract}


\end{frontmatter}



\section{Introduction}
\label{intro}
Data models, representing a set of imposed mathematical constraints, are predominantly employed for the regularization of ill-posed inverse problems and, accordingly, their distinct choice and a thorough understanding of their properties is fundamental for their successful embedding in signal and image processing tasks. The study of inverse problems has recently inspired the investigation of a `duality gap' between the (\textit{sparse}) \textit{synthesis} and (\textit{cosparse}) \textit{analysis} signal models, \cite{elad}, \cite{cos}, signifying two prominent instances of the more comprehensive Union of Subspaces (UoS) signal model \cite{uos}, and, as such, it has become of increasing interest to comprehend when the two models cease to be equivalent, and, specifically, how they differ. While the models are known to be equivalent when the generating operator is nonsingular, their relation lacks a precise characterization when the operator at hand is rank-deficient square (or rectangular). In particular, it has become evident that it is necessary to go beyond the description of single quantities such as the spark of a synthesis operator ${\bf D}$, and, more recently, with the advent of \textit{cosparsity} \cite{cos}, the maximum subspace dimension $\kappa_{{\bf \Omega}}(l)$ of an analysis operator ${\bf \Omega}$ inducing cosparsity $l$, in order to quantify and describe non-trivial linear dependencies and understand model discrepancies precisely. \\
\\
We propose to remedy this by concretely looking at the underlying structure and conducting an analytic characterization of the defining subspaces.
In this work, we investigate UoS signal models in the structured domain of graphs, primarily based on the graph Laplacian matrix ${\bf L}$, as a fundamental graph (difference) operator, which is square rank-deficient, with extensions pertaining to the rectangular rank-deficient oriented incidence matrix ${\bf S}$ and higher-order generalizations. Here, we interpret ${\bf \Omega}={\bf L}$ as an analysis operator and consider its synthesis counterpart through the Moore-Penrose Pseudoinverse (MPP) ${\bf D}={\bf L}^{\dagger}$. \\
By focusing on a subset of highly structured graph difference operators, with known and novel annihilation properties, we aim to provide refined insights into the discrepancy between the cosparse analysis and sparse synthesis models. In particular, we wish to elucidate the transition between the two when the generating operator is square singular in the discrete structured domain of undirected connected (circulant) graphs, in an effort to comprehend their concise differences and eventually motivate an Ansatz for more general scenarios.  At its core, this study uncovers that the underlying linear dependencies of rank-deficient graph analysis operators pose constraints on their defining subspaces, effectively reducing the order of the underlying functions, as a result of the \textit{Fredholm Alternative} (F.A.) \cite{fred}, in contrast to their unconstrained synthesis counterparts.\\
\\
In the course of this analysis, we conduct a complete characterization of the subspaces of the two models by defining their underlying structure through discrete Green's functions (in closed-form expressions, where applicable), as well as by quantifying their number and dimension, which, as will become evident, cease to be uniform beyond the non-singular case. Due to their rich structure, circulant graph matrices particularly lend themselves for a more concrete analysis and the development of unique closed-form expressions. As a result of the connection between Green's functions and (pseudo)inverses, we observe the occurrence of boundary value phenomena which impact the shape and order of functions, thereby distinctly demonstrating a difference between synthesis and analysis-based solution subspaces via the F.A. constraint.\\
Specifically, we discover how the singularity of ${\bf L}$ creates a discrepancy in the type and localization of the underlying signal subspaces of each model as well as dictates their associated discontinuities, i.e.\ (\textit{structured}) \textit{sparsity} pattern. For circulant graphs, we concretely demonstrate that while the sparse synthesis model induces up to piecewise quadratic polynomial signals, the cosparse analysis model is degree-reduced and only encompasses up to piecewise linear polynomials, both of which are subject to graph-dependent perturbations. By considering a parametric extension of ${\bf L}$ on circulant graphs, in the form of the previously developed operator ${\bf L}_{\alpha}$, which annihilates complex exponential signals with exponent $\alpha\in\mathbb{C}$, \cite{splinesw}, we further show that when ${\bf L}_{\alpha}$ is singular, the associated sparse synthesis model generates up to (perturbed) linear complex exponential polynomials, while the analysis model only generates (perturbed) complex exponentials. In contrast, when ${\bf L}_{\alpha}$ is nonsingular, for certain choices of $\alpha$, both models become equivalent and generate (perturbed) complex exponential signals. Accordingly, the separate study of a range of graph difference operators with closed-form expressions is intended to uncover transitional properties between equivalence and non-equivalence of the two models, thereby exemplifying their fundamental difference.
Ultimately, we leverage developed insights to initiate a model-based UoS framework on graphs. Here, we directly quantify subspace measures for the refinement of existing uniqueness guarantees as well as for UoS-based sampling theorems. We further discover that at the heart of the introduced graph Laplacian-based UoS model lies a \textit{structured sparsity} model which inspires the creation of tailored UoS models with desirable properties, such as a reduced number of total subspaces of low dimension. 
\\
\\
{\bf Why Graphs?} \\
Graph matrices provide insightful tools for the characterization of subspaces, and UoS signal models in particular, due to a number of convenient properties which can help capture the inherent signal geometry: Graph Laplacians are positive semi-definite (PSD) matrices of Gramian structure ${\bf L} ={\bf S}^T{\bf S}$, with well-defined, sparse structured incidence matrix ${\bf S}$. The graph connectivity manifests itself in the irreducibility\footnote{A matrix is irreducible if it cannot be transformed into a block-upper triangular matrix via permutations \cite{horn}.} of ${\bf L}$, which ceases to apply for graphs with more than one connected component. The nullspace and range of ${\bf L}$ (and by association, of ${\bf S}$ and ${\bf S}^T$ respectively) are known, and the linear dependencies of the matrix are essentially the result of the zero-sum constraint (capturing wavelet-like `vanishing moments' on graphs) inherent in the columns of ${\bf S}^T$. 
These linear constraints form an essential part of the analysis model of ${\bf L}$, as imposed through the F. A., and ultimately quantify a convenient rank deficiency of the matrix, which in turn give rises to a structured sparsity model. 
In addition, owing to structural properties of ${\bf L}$ and ${\bf S}$, and their MPPs by association, including structured sparsity, and, for more specialized graph structures such as circulants, extending to polynomial functions, they offer a broad representation range for signals.
The growing field of \textit{Graph Signal Processing (GSP)}, \cite{shu}, seeks to leverage the inherent ability of graphs to capture the geometric complexity of irregularly structured, complex data for arising signal representation and processing tasks. Since this requires the extension of classical signal processing theory to the graph domain, the development of a rigorous theoretical foundation is paramount.\\
Circulant graphs have been noted for providing a link between the classical (Euclidean) and graph domain of signal processing, which previously motivated the development of sparse graph wavelet analysis and sampling theory \cite{splinesw}, \cite{acha2}. As we will discover in Sect.\ $4$, most of the analysis pertaining to circulant graphs can be derived on the basis of the simple cycle, which is central to classical signal processing considering that a signal defined on its vertices is equivalent to a time-periodic signal, while its graph Laplacian, which encapsulates two vanishing moments, forms the high-pass portion of a fundamental discrete wavelet matrix. Further, circulant matrices are uniformly diagonalizable by the DFT-matrix and, as such, its underlying subspaces are essentially (complex exponential) polynomials, which will be substantiated in the derivation of MPPs, and provides a tangible link between graph theory and harmonic analysis.\\
\\
The presented study leverages links between several fields, facilitated through the structural properties of graphs in general, and circulant graphs in particular, in order to tackle the analysis vs. synthesis problem.
In the first instance, graph and matrix theory benefit from a rich interplay in that graphs can be represented as structured matrices with unique linear algebraic properties, while their analysis uncovers interesting phenomena which inspire the study of more general matrices.\footnote{An example product of such an interaction is the Perron-Frobenius Theorem \cite{horn}.} Further, the Fredholm Alternative, marking the decisive constraint which separates the analysis from the synthesis framework, and which is further found in the formulation of the MPP, originates from the theory of PDEs, as does the concept of Green's functions. The derivation of closed-form expressions for MPPs on circulant graphs further leverages polynomial equations and recurrence relations, ordinarily employed for the solution of differential equations.
\\
It is precisely the occurrence of structure which facilitates these links and renders the close investigation of the problem feasible. This work seeks to elucidate these connections and thereby provide a rich and comprehensive  description of the defining subspaces in data models, with relevant implications for graph theory and extending up to signal processing.
\\
\\
{\bf Contributions.} We summarize the main contributions as follows:
\begin{enumerate}
\item Analytic description of the discrepancy between analysis and synthesis models for the graph Laplacian ${\bf L}$ through subspace analysis (Thms.\ \ref{thm331}, \ref{thmmain}, \ref{main2})
\item Development of closed-form expressions of the (pseudo)inverses (Green's functions) of ${\bf L}$, ${\bf S}$ and ${\bf L}_{\alpha}$ on circulant graphs, providing broad representation range for signals (Lemmata \ref{circlins}, \ref{lemcircl}, \ref{lemlincirc}, \ref{inva}, \ref{mppa}, Cor.\ \ref{mppinva})
\item Analysis of the special case of the parametric circulant graph Laplacian ${\bf L}_{\alpha}$, showing a transition from inverse to MPP, and hence, the incremental formation of model discrepancies (Thm.\ \ref{main2}, Rem.\ \ref{contr1})
\end{enumerate}
{\bf Related Work.}
The study of identifying equivalencies and discrepancies between analysis-and synthesis-driven signal models was initiated in \cite{elad},\cite{cos}, whereby \cite{cos} introduced the novel concept of \textit{cosparsity}, as a distinct and potentially more powerful avenue than that of sparsity. While \cite{cos} establishes for matrices in general position that the analysis model, as an instance of a UoS model, is a special case of the synthesis model, and conducts a case study for the incidence matrix of the grid graph, it does not take into consideration the specific composition of the subspaces and their dependencies, nor does its analysis apply to rank-deficient square operators. 
Unser et al.\ \cite{unser} derive representer theorems, which provide the general, regularization-dependent, solution structure of both synthesis- and analysis-driven approaches in comparison; while in infinite dimensions, this requires the derivation of a stable constrained right-inverse operator for the analysis case, the finite rank-deficient case is not treated and a precise comparison of subspaces not conducted. 
Further, in \cite{unser}, a boundary condition is employed for the construction of the right inverse operator in infinite dimensions; as pointed out by Flinth et al. \cite{flinth}, the former focus on a specific class of operators with finite kernel, termed Fredholm operators. As will become evident, the Fredholm Alternative is crucial in the characterization of the analysis-synthesis discrepancy of finite rank-deficient operators.\\
The field of GSP has featured analysis-driven approaches, in the form of generalized graph operator design, (multiresolution) graph wavelet analysis and filterbank construction (e.g. \cite{Coifman}, \cite{spectral}, \cite{ortega3}), as well as synthesis-driven approaches, including the learning and/or design of graph-based dictionaries \cite{dict1}, with instances of one inducing the other, \cite{kov}, \cite{tools}. Nevertheless, a comparative theoretical study of the two models has not been realized in this context.\\
A previous body of work \cite{splinesw}, \cite{acha2}, which developed a framework for sparse graph wavelet analysis and sampling on circulant graphs and beyond, initiated the study of signal sparsity in the light of the connectivity of graphs; however, it offered only implicit characterization of the underlying signal model when the graph at hand is circulant. In Sect.\ $4$, we specifically leverage prior results in order to motivate and expand the comparative study of the two models.\\
In \cite{smola}, the topic of analysis-driven graph trend filtering, which employs the $l_1$-minimization of a difference term based on the graph Laplacian and its higher-order generalizations, is explored; nevertheless, the derivation of the analysis solution subspaces is flawed, ignoring crucial constraints. In \cite{Pesenson}, variational splines on graphs are defined as the Green's functions of a regularized graph Laplacian operator ${\bf L}+\epsilon{\bf I}_N$, with small parameter $\epsilon>0$, which is effectively invertible. However, due to its invertibility, the operator loses the distinctive features, aka the associated linear dependencies, of ${\bf L}$, rendering the analysis-synthesis discrepancy non-existent, and its underlying signal model constitutes only an approximation on the graph. In contrast, in this work, we focus on the properties of the graph Laplacian MPP ${\bf L}^{\dagger}$, as a Green's function.\\
Part of this work appears in a conference paper \cite{global}.
\\
\\
This paper is organized as follows: we state the notation and relevant prerequisite theory from graph theory and graph signal processing in Sect. $2$. In Sect. $3$, we formulate the problem statement and introduce the (co)sparse signal models on general undirected graphs, where we firstly conduct a separate study of the two models, before establishing their distinct differences as instances of UoS models with regard to the composition, dimension and number of unique combinations of their associated subspaces. In Sect. $4$, we lay the focus on circulant graphs which facilitate a concretization of previously discovered model discrepancies. In addition, we study a  generalized, parametric version of the graph Laplacian on circulant graphs, and its (pseudo)inverse, in order to elucidate how the rank-deficiency of a structured difference operator creates a discrepancy between the subspaces of analysis and synthesis-driven models. At last, in Sect. $5$, we leverage derived results in order to refine uniqueness and recovery guarantees for signals belonging to constrained UoS graph models, and position them within the field of model-based compressed sensing. Further, we establish that the developed (co)sparse UoS graph models give rise to a structured sparsity model and discuss its properties, in the light of the question of what constitutes a desirable UoS model. In Sect. $6$, we make concluding remarks and give all proofs not included in the main text in the appendix.

\section{Preliminaries}
\subsection{Notation}
\noindent We denote vectors with boldfaced lower case letters ${\bf x}$ and matrices with boldfaced uppercase letters ${\bf A}$. Let ${\bf 1}_N$ and ${\bf 0}_N$ define the constant column vectors of length $N$ with entries of 1's and 0's respectively, while ${\bf t}=\lbrack 0\ 1\ ...\ N-1\rbrack^T$ denotes the vector of sequential numbers from $0$ to $N-1$, as we adopt zero-based numbering unless stated otherwise. Further, let ${\bf 1}_C$, for some index set $C$, denote the vector with $1$'s at positions in $C$ and zeros otherwise. The vector and matrix norms of relevance are the $l_0$-pseudo-norm, denoted with $||{\bf x}||_0=\#\{i:x_i\neq0\}$ and the $l_2$-norm, given by $||{\bf x}||_2=\left(\sum_{i=0}^{N-1} |x_i|^2\right)^{1/2}$. The canonical basis vectors ${\bf e}_i$ satisfy $e_i(i)=1$ and $e_i(j)=0,\enskip j\neq i$, and ${\bf J}_N$ is the all-ones matrix of size $N$.
Given a matrix ${\bf L}$ and (the sets of) indices $A$ and $B$, the notation ${\bf L}(A,B)$ or ${\bf L}_{A,B}$ indicates that the corresponding rows and columns in ${\bf L}$ are chosen. In addition the matrix ${\bf \Psi}_{\Lambda}$ is defined as the sampling matrix with 
\[  \Psi_{\Lambda}( i,j)=\left\{
  \begin{array}{@{}ll@{}}1, & j=\lambda_i\in \Lambda\\    
0,& \text{otherwise}
 \end{array}\right.
\] 
for the ordered index set $\Lambda\subset \lbrack 0\enskip... \enskip N-1\rbrack$, where $\lambda_i$ denotes the $i$-th element in $\Lambda$. Further, let ${\bf \Psi}_{\Lambda}{\bf x}={\bf x}_{ \Lambda}$. We maintain for simplicity the notational convention according to which ${\bf D}_{\Lambda}$ denotes the ${\Lambda}$-indexed columns of ${\bf D}$ in the synthesis model, while for the analysis model, ${\bf \Omega}_{\Lambda}$ represents the rows ${\Lambda}$ of ${\bf \Omega}$. If no model is specified, the former convention is applied.

\subsection{Graph Theory}
A graph constitutes a connectivity structure, characterized by a set $V=\{0,1,...,N-1\}$ of vertices and set $E$ of edges, which connect pairs of vertices, and is formally denoted by $G=(V,E)$ with cardinality $|V|=N$. The adjacency matrix ${\bf A}\in\mathbb{R}^{N\times N}$ captures the graph connectivity by assigning non-zero weights to existing edges between any pair of vertices $\{i,j\}$ at entries $A_{i,j}>0$ and  $A_{i,j}=0$ otherwise, while the diagonal degree matrix ${\bf D}$ contains the sum of the weights (denoted as degree) at each vertex with $D_{i,i}=\sum_j A_{i,j}$. The (non-normalized) graph Laplacian ${\bf L}={\bf D}-{\bf A}$ is a prominent graph matrix whose particular set of properties has been widely investigated within spectral graph theory and made use of beyond, and constitutes a fundamental graph difference matrix. In particular, for an undirected graph, ${\bf L}$ is a symmetric positive semi-definite (PSD) matrix with a complete set of orthogonal eigenvectors $\{{\bf u}_l\}_{l=0}^{N-1}$ and a non-negative spectrum $0=\lambda_0\leq \lambda_1\leq ..\leq \lambda_{N-1}$; when the graph at hand is additionally connected we have $\lambda_1>0$. Further, the operator ${\bf L}^k, \enskip k\in\mathbb{N}$, is strictly $k$-hop localized in the vertex domain, with $({\bf L}^k)_{i,j}=0$ when the shortest-path distance (i.e. the number of hops) between $i,j$ is greater than $k$.
For the remainder of this work, we a priori assume that the graph at hand is undirected. 

Another graph matrix of interest is the oriented edge-vertex incidence matrix ${\bf S}\in\mathbb{R}^{|E|\times |V|}$ of $G$, which assigns an arbitrary but fixed direction to each edge, conventionally with $S_{k,i}=\sqrt{A_{i,j}}$ and $S_{k,j}=-\sqrt{A_{i,j}}$ if the $k$-th edge $\{i,j\}$ is directed from $i$ to $j$, resulting in the operation $({\bf S}{\bf x})_{\{i,j\}}=\sqrt{A_{i,j}}(x(i)-x(j))$ at edge $\{i,j\}$. In an analogy to discrete differential geometry operators and as outlined in prior work \cite{splinesw}, the graph Laplacian ${\bf L}$ constitutes a graph-realization of a second-order differential operator, while the incidence matrix ${\bf S}$ can be interpreted as a first-order differential operator. Moreover, the graph incidence and graph Laplacian matrices are linked through the Gram operation ${\bf L}={\bf S}^T{\bf S}$, which, as we will discover, marks a convenient property when dealing with pseudoinverse operations, while both ${\bf S}$ and ${\bf L}$ have rank $N-t$ in general, where $t$ is the number of connected components in $G$. For a connected graph with $t=1$, we have $N({\bf L})=N({\bf S})=z{\bf 1}_N, \ z\in\mathbb{R}$, where $N(\cdot)$ denotes the nullspace, while for a disconnected graph with $t$ connected components and corresponding vertex sets $\{C_k\}_{k=1}^t$, this becomes $N({\bf L})=N({\bf S})=span\{{\bf 1}_{C_1},...,{\bf 1}_{C_t}\}$.
\\
\\
We define a signal on the vertices of a graph $G$ as a complex-valued scalar function of dimension $N$, which assigns a sample value $x(i)$ to node $i$ and can be represented as a vector ${\bf x}\in\mathbb{C}^N$. This notion is central to Graph Signal Processing (GSP), which i.a. studies the design of linear graph-based operators and their application on the signal associated with the graph at hand \cite{shu}. In GSP, signal smoothness has been more commonly associated with sparsity in the graph frequency domain, a concept also known as bandlimitedness \cite{shu}, while a complete characterization in the vertex domain remains opaque. In this work, we interpret the class of (\textit{piecewise}) \textit{smooth} graph signals to signify sparsity with respect to the underlying graph connectivity, as induced via the operation with a designated graph operator, as follows:
\begin{defe}
The graph signal ${\bf x}\in\mathbb{R}^N$ is piecewise-smooth on $G$ with respect to a graph operator ${\bf L}$ if its representation ${\bf L}{\bf x}$ is sparse, i.e. $||{\bf L}{\bf x}||_0\ll N$. 
\end{defe}
As will be established, this notion is guided by the underlying Green's functions (and by association, the MPP) of the operator. Furthermore, a relevant class of graph signals on circulant graphs has been revealed to be that of classical piecewise polynomials \cite{splinesw}, which, for clarity, we define as follows:
\begin{defe} (\cite{splinesw})
A graph signal ${\bf p}\in \mathbb{R}^N$ defined on the vertices of a circulant graph $G$ is (piecewise) polynomial if its labelled sequence of sample values, with value $p(i)$ at node $i$, is the discrete, vectorized version of a standard (piecewise) polynomial. In particular, we have ${\bf p}=\sum_{j=1}^K {\bf p}_j \circ {\bf 1}_{\lbrack t_j,t_{j+1})}$, for Hadamard product $\circ$, where $t_1=0$ and $t_{K+1}=N$, with pieces
$p_j(t)=\sum_{d=0}^D a_{d,j} t^d,\enskip j=1,...,K$, for $t\in\mathbb{Z}^{\geq 0}$, coefficients $a_{d,j} \in\mathbb{R}$, and maximum degree $D=deg(p_j(t))$.
\end{defe}
\subsection{Circulant Graphs and Matrices}
Circulant graphs reveal a distinct set of properties, which have previously facilitated the development of graph wavelet analysis and sampling \cite{splinesw}, \cite{acha2}. 
A circulant graph $G_S$ is defined via a generating set $S = \{s_1, ..., s_M \}$, with $0<s_k \leq N/2$, whose elements indicate the existence of an edge between node pairs $(i, (i \pm s_k)_N )$, $\forall s_k \in S$, where $()_N$ is the mod $N$ operation. In general, a graph is circulant if its associated graph Laplacian is a circulant matrix under a particular node labelling (see Fig.\ \ref{fig:aa3} for examples). In order to fully leverage the properties of circulant matrices, we henceforth assume that the circulant graph at hand is labelled such that its associated matrices are circulant. Further, a circulant graph is connected if the greatest common divisor of the elements in its generating set $S$ and the graph dimension $N$ is $1$ \cite{circcon}; for simplicity and, as will become evident in subsequent derivations, in the interest of mathematical convenience, we always assume $s=1\in S$ to ensure connectivity. Circulant matrices are characterized by a representer polynomial $l(z)=\sum_{k=0}^{N-1} l_k z^k$ whose entries are taken from its first row $\lbrack l_0 \ l_1 \ ...\ l_{N-1}\rbrack$. Symmetric circulant matrices with first row $\lbrack l_0\enskip l_1\enskip l_2\ ...\ l_2\enskip l_1\rbrack$ and bandwidth $M$ are of the form
\[{\bf L}=\begin{bmatrix}
 l_0 & l_{1} &  \cdots     & l_{2} &  l_{1} \\
 l_{1} &  l_0 & \ddots & \ddots &  l_{2} \\
\vdots &  \vdots & \ddots & \ddots &  \vdots  \\
 l_{2} &  l_{3} & \ddots & \ddots & l_{1} \\
l_{1} &  l_{2} &  \cdots     & \cdots     &  l_0
\end{bmatrix}
\]
where $l(z)$ can be converted to a Laurent polynomial with $l_i=l_{N-i}$ for $i>0$ such that $l(z) = l_0 +  \sum^M_{i=1} l_i(z^i + z^{-i})$. It is further noteworthy that circulant matrices are diagonalizable by the DFT-matrix. In particular, $l(z)$ gives rise to the eigenvalues of ${\bf L}$, as ordered per diagonalization by the DFT-matrix, at frequency locations $\frac{2\pi i k}{N}$ with $l(e^{\frac{2\pi i k}{N}}) = \lambda_k$, $k = 0,...,N-1$ \cite{circul}.
Moreover, for ${\bf L}={\bf D}-{\bf A}$ with degree matrix ${\bf D}=d{\bf I}_N$ and symmetric circulant adjacency matrix ${\bf A}$, the individual entries are given by $l_0=d=\sum_{i=1}^M 2 d_i$ and $l_i=-d_i$, where $d_i=A_{j,(i+j)_N}$ denote the symmetric edge weights.
 \begin{figure}
 \centering
  \begin{subfigure}[htbp]{0.32\textwidth}
  \centering
{\includegraphics[width=1.2in]{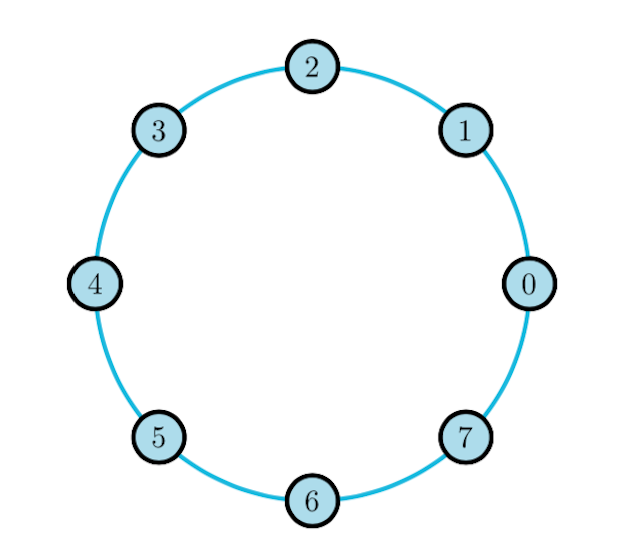}} 
\caption{\small{$S=\{1\}$ (simple cycle)}}
\end{subfigure}
\begin{subfigure}[htbp]{0.32\textwidth}
\centering
 { \includegraphics[width=1.3in]{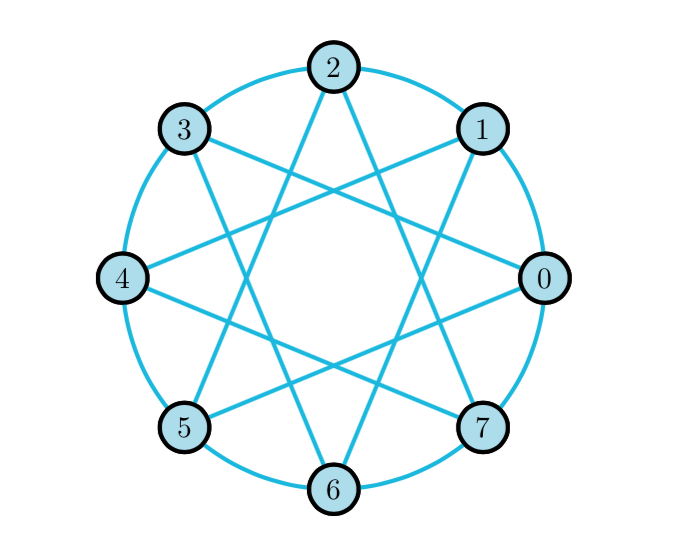}}
 \caption{\small{$S=\{1,3\}$}}
 \end{subfigure}
  \begin{subfigure}[htbp]{0.34\textwidth}
  \centering
 { \includegraphics[width=1.22in]{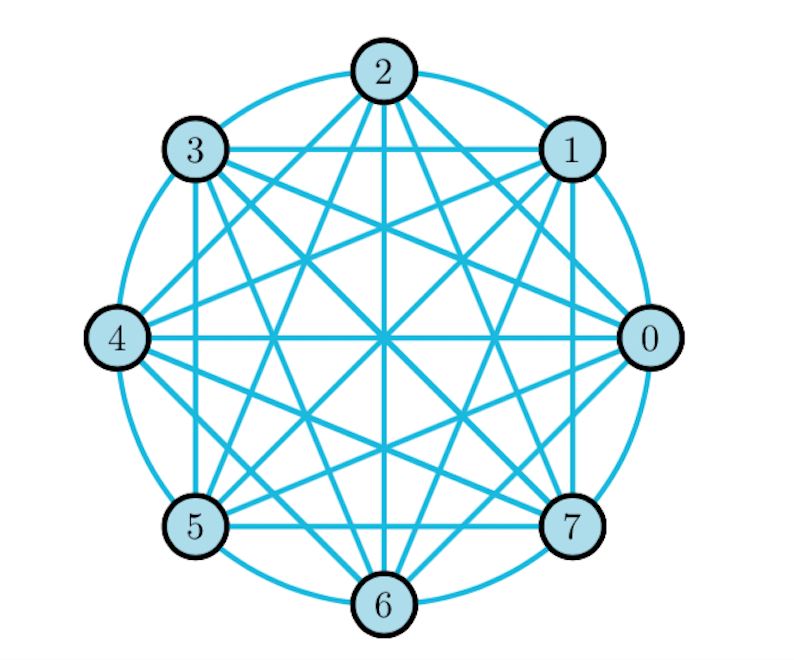}}
 \caption{\small{$S=\{1,2,3,4\}$ (complete) }}
 \end{subfigure}
	\caption{Circulant Graphs $G_S$ with Generating Set $S$}
\label{fig:aa3}\end{figure}
We focus on cyclically banded circulant matrices (graphs) of bandwidth $M$ with $2M\ll N$. It has been shown in prior work (Lemma $3.1$, \cite{splinesw}) that the representer polynomial $l(z)$ of the banded circulant graph Laplacian  ${\bf L}$ has $2$ vanishing moments, i.e. it annihilates up to linear polynomials, and by circular convolution (matrix multiplication), ${\bf L}^k$ annihilates polynomials of up to order $2k-1$, subject to a border effect dependent on the bandwidth $M$ of the graph. Here, the `border effect' refers to the amplified spread around the discontinuities of a signal, incurred following the application of a finite banded circulant graph operator, i.e. each signal discontinuity is amplified by the bandwidth of the graph matrix. 
\\
\\
Further, a parametric generalization of the graph Laplacian on circulant graphs was designed to be of the form ${\bf L}_{\alpha}=d_{\alpha}{\bf I}_N-{\bf A}$, with parameterized degree $d_{\alpha}=\sum_{j=1}^M 2d_j \cos(\alpha j)$, $\alpha\in\mathbb{C}$, so as to annihilate complex exponential polynomial signals ${\bf x}={\bf p} e^{\pm i \alpha {\bf t}}$, with discrete polynomial ${\bf p}$ of degree $0$ and ${\bf t}=\lbrack 0\ 1\ ...\ N-1\rbrack^T$. Unless $\alpha=\frac{2\pi k}{N}$ for $k\in \lbrack 0\enskip N-1\rbrack$, this is subject to a border effect dependent on the bandwidth $M$ of ${\bf L}_{\alpha}$ and ${\bf L}_{\alpha}$ is invertible (follows from Cor.\ $3.3$ \cite{splinesw}, see also Property \ref{prop4}, Sect. $4$). By extension, ${\bf L}_{\alpha}^k$ annihilates complex exponential polynomials with the same exponent of degree $k-1$. Here, the choice of $d_{\alpha}$ is related to the characteristic eigenvalue structure of the circulant ${\bf A}$, and for $\alpha=\frac{2\pi k}{N}$, we have $d_{\alpha}=\lambda_k({\bf A})$, $k\in \lbrack 0 \enskip N-1\rbrack$. The operator ${\bf L}_{\alpha}$ forms part of a more generalized class of graph Laplacians for arbitrary undirected graphs $G$, given by $\tilde{{\bf L}}={\bf L}+{\bf P}$ \cite{laplaceeigen}, where ${\bf P}$ is an arbitrary diagonal matrix, whose entries may be chosen to produce e.g. an annihilation effect for a given class of signals.

\section{The (Co)sparse Signal Model on Graphs}
\subsection{Problem Statement and Context}
For the task of modelling a given signal ${\bf x}\in\mathbb{R}^N$, we distinguish between a \textit{sparse (generative) synthesis} approach and a \textit{cosparse (descriptive) analysis} approach. In the former case, the signal is `synthesized' as the linear combination of few columns from a given dictionary ${\bf D}\in\mathbb{R}^{N\times M}$ with $M\geq N$, giving rise to the representation ${\bf x}={\bf D}{\bf c}$, where the coefficient vector ${\bf c}\in\mathbb{R}^{M}$ is sparse with $||{\bf c}||_0\ll M$. In the latter case, the signal is implicitly described through the application of an analysis operator ${\bf \Omega}\in\mathbb{R}^{M\times N}$ such that ${\bf \Omega}{\bf x}$ is sparse with $||{\bf \Omega}{\bf x}||_0\ll M$; in particular, this model has inspired the advent of the concept of \textit{cosparsity}, which places emphasis on the zeros as opposed to the non-zeros, as measured by the quantity $l:=M-||{\bf \Omega}{\bf x}||_0$.\\
\\
Let ${ \Lambda}$ and ${\Lambda}^{\complement}$ denote the locations of the sparse and cosparse entries of ${\bf c}$ and ${\bf \Omega}{\bf x}$ respectively such that ${\bf x}={\bf D}_{\Lambda^{\complement}}{\bf c}_{\Lambda^{\complement}}$ and ${\bf \Omega}_{\Lambda}{\bf x}={\bf 0}_{\Lambda}$. Accordingly, the solution subspaces induced by these signal models are respectively given by $V_{\Lambda^{\complement}}:=span({\bf D}_j, j\in\Lambda^{\complement})$ and $W_{\Lambda}:=N({\bf \Omega}_{\Lambda})$.\\
When ${\bf D}={\bf \Omega}^{-1}$, the analysis and synthesis model become trivially equivalent as their underlying subspaces are identical, however, when the operators are rank-deficient and/or rectangular, discrepancies arise whose study is still in its infancy and an exact characterization is desirable.\\
\\
Given the analysis operator ${\bf \Omega}$, one possibility to evaluate its synthesis counterpart operator is through the Moore-Penrose Pseudoinverse (MPP) ${\bf D}={\bf \Omega}^{\dagger}$. For the general inverse problem ${\bf \Omega}{\bf x}={\bf y}$, the MPP, which is uniquely defined, is known to provide the set of solutions ${\bf x}'$ which minimize  $||{\bf \Omega}{\bf x}'-{\bf y}||_2$ \cite{mpp}, given by 
\begin{equation}
 {\bf x}'={\bf \Omega}^{\dagger}{\bf y}+({\bf I}_N-{\bf \Omega}^{\dagger}{\bf \Omega}){\bf z},\enskip {\bf z}\in{\bf R}^N.
 \end{equation}
Here, ${\bf P}_{N({\bf \Omega})}=({\bf I}_N-{\bf \Omega}^{\dagger}{\bf \Omega})$ is the orthogonal projector onto the nullspace of ${\bf \Omega}$ and solutions to the above problem exist if and only if ${\bf \Omega}{\bf \Omega}^{\dagger}{\bf y}={\bf y}$. Specifically, according to the \textit{Fredholm Alternative} (F.A.) \cite{fred}, a solution ${\bf x}'$ exists if and only if $\langle {\bf y},{\bf n}_j \rangle =0$ $\forall {\bf n}_j\in N({\bf \Omega}^T)$, where ${\bf n}_j\in\mathbb{R}^N$ denote the basis atoms of the nullspace, and equivalently if and only if ${\bf y}\in N({\bf \Omega}^T)^{\perp}=R({\bf \Omega})$, for range $R(\cdot)$. This condition forms the basis for the creation of the MPP ${\bf \Omega}^{\dagger}$ \cite{mpp}, and, as such, is essential for the characterization of the solution set.
\\
When ${\bf \Omega}$, with $M=N$, is invertible, the solution to the above problem is unique and the nullspace translation-term vanishes; in other words, the analysis operation ${\bf \Omega}{\bf x}$ is equivalent to the synthesis operation ${\bf \Omega}^{-1}{\bf \Omega}{\bf x}={\bf \Omega}^{-1}{\bf y}$ with no loss of information. When ${\bf \Omega}$, with $M> N$, is of full column rank, the solution remains unique only if the Fredholm Alternative is satisfied for a particular choice of ${\bf y}$, however, as has been pointed out in \cite{elad}, the two operations are no longer equivalent owing to more subtle differences in their defining subspaces.\\
\\
In \cite{cos}, the measure  $\kappa_{{\bf \Omega}}(l):=\max_{|\Lambda|\geq l} dim (W_{\Lambda})$ was defined to characterize the interdependency between rows of the analysis operator, as a counterpart to the \textit{spark} in the synthesis model, which defines the minimum number of linearly dependent columns in ${\bf D}$. Derived solutions can be leveraged for the quantification of necessary and sufficient (co)sparsity levels for uniqueness of representation results. Here, we aim to take the comparison of the two models a step further by not only quantifying these dependency measures (as conducted in Sect. $5$) but also directly characterizing the underlying solution subspaces which facilitates more concrete claims on linear dependencies and how these models are fundamentally interrelated.\\ 
\\
We  tackle this problem by considering analysis and synthesis operators which represent structured graph (difference) matrices and their MPPs, specifically focusing on the rank-deficient square graph Laplacian ${\bf L}$, with extensions to its high-order generalizations. In the following, we separately adopt the synthesis and analysis perspective with focus on general undirected graphs in the first instance, and connected (banded) circulant graphs in particular, the latter of which are of interest due to the characteristic closed-form expressions and concise insights one can obtain as a result of their structure. Subsequently, we compare and discuss these derivations in light of a generalized union of subspaces model and aim to position them within an analysis vs. synthesis discrepancy spectrum as well as attempt to motivate generalizations beyond the structured matrices of graphs.
We begin by considering the graph Laplacian analysis operator ${\bf \Omega}={\bf L}$ and its synthesis counterpart ${\bf D}={\bf L}^{\dagger}$ in the vertex domain of an undirected graph $G=(V,E)$.
\subsection{The Cosparse Analysis Model}
In order to characterize the subspaces of piecewise-smooth graph signals ${\bf x}$ on $G$, which can be partially annihilated by ${\bf L}$, we need to derive $N({\bf \Psi}_{\Lambda}{\bf L})$, where $\Lambda\subset V$ denotes the \textit{cosupport} associated with the zero entries of ${\bf L}{\bf x}$. Hence, we state the following:
\begin{prop}\label{prop1} The nullspace of ${\bf \Psi}_{\Lambda}{\bf L}$, where ${\bf L}$ is the graph Laplacian of an undirected connected graph $G=(V,E)$, has rank $|\Lambda^{\complement}|$ for $|\Lambda|<N$, and is described by $N({\bf \Psi}_{\Lambda}{\bf L})=z{\bf 1}_N +{\bf L}^{\dagger}{\bf \Psi}_{\Lambda^{\complement}}^T{\bf W}{\bf c}$, for arbitrary $z\in\mathbb{R}$, ${\bf c}\in\mathbb{R}^{|\Lambda^{\complement}|-1}$, and ${\bf W}\in\mathbb{R}^{|\Lambda^{\complement}|\times(|\Lambda^{\complement}|-1)}$ given by
 \begin{equation}{\bf W}:=\scalebox{0.8}{$\begin{pmatrix} |\Lambda^{\complement}|-1&0&&\dots& &0\\
-1&|\Lambda^{\complement}|-2& 0&\dots& &0\\
&-1&|\Lambda^{\complement}|-3&&&\vdots\\
\vdots&&&&&\\
&&&&&0\\&&&&&1\\-1&-1&\dots&&&-1\end{pmatrix}$}.\end{equation}
\end{prop}
\begin{proof}
We have $N({\bf L})=z{\bf 1}_N\subset N({\bf \Psi}_{\Lambda}{\bf L})$ for $z\in\mathbb{R}$. \\
In order to obtain the complete subspace $N({\bf \Psi}_{\Lambda}{\bf L})$, we consider the problem ${\bf \Psi}_{\Lambda}{\bf L}{\bf u}={\bf 0}_{\Lambda}$, which is equivalent to solving ${\bf b}={\bf L}{\bf u}=N({\bf \Psi}_{\Lambda})$. For a solution ${\bf u}$ to exist, we need to satisfy the Fredholm Alternative, i.e. $N({\bf \Psi}_{\Lambda})\perp N({\bf L})$. Since this is not true in general, we need to introduce linear constraints; in particular, without loss of generality, let $N({\bf \Psi}_{\Lambda})={\bf \Psi}_{\Lambda^{\complement}}^T$ and consider the constraint matrix ${\bf W}\in\mathbb{R}^{|\Lambda^{\complement}|\times k}$ so that ${\bf L}{\bf u}={\bf \Psi}_{\Lambda^{\complement}}^T{\bf W}{\bf c}$, for unknown subspace dimension $k$ and arbitrary coefficient vector ${\bf c}\in\mathbb{R}^k$. Then, for ${\bf w}:={\bf W}{{\bf c}}$, a solution ${\bf u}$ exists with ${\bf u}={\bf L}^{\dagger}{\bf \Psi}_{\Lambda^{\complement}}^T{\bf w}$ which we can synthesize to obtain:\\
\begin{equation}\label{eq:nullspace}{\bf \Psi}_{\Lambda}{\bf L}{\bf u}={\bf \Psi}_{\Lambda}{\bf L}{\bf L}^{\dagger}{\bf \Psi}_{\Lambda^{\complement}}^T{\bf w}={\bf \Psi}_{\Lambda}\left({\bf I}_N-\frac{1}{N}{\bf J}_N\right) {\bf \Psi}^T_{\Lambda^{\complement}}{\bf w}=-\frac{1}{N}{\bf J}_{{\Lambda},{\Lambda}^{\complement}}{\bf w}={\bf 0}_{\Lambda}.\end{equation}
Here, we have used the equality ${\bf L}{\bf L}^{\dagger}=({\bf I}_N-\frac{1}{N}{\bf J}_N)$ for general graph Laplacians of connected graphs. Overall, we deduce \begin{equation}\label{eq:const}{\bf W}:=N({\bf J}_{|{\Lambda}|,|{\Lambda}^{\complement}|})=\scalebox{0.8}{$\begin{pmatrix} |\Lambda^{\complement}|-1&0&&...& &0\\
-1&|\Lambda^{\complement}|-2& 0&..& &0\\..&-1&|\Lambda^{\complement}|-3&&&0\\&&&&&\\..&&&&&0\\&&&&&1\\-1&-1&...&&&-1\end{pmatrix}$}\end{equation}which reveals a zero-sum column structure. In other words, the basis ${\bf W}\in\mathbb{R}^{|\Lambda^{\complement}|\times(|\Lambda^{\complement}|-1)}$ is spanned by atoms whose entries sum to 0. Hence, we have
\[N({\bf \Psi}_{\Lambda}{\bf L})=z{\bf 1}_N+{\bf L}^{\dagger}{\bf \Psi}^T_{\Lambda^{\complement}}{\bf W}{{\bf c}},\] for arbitrary coefficient vector ${\bf c}\in\mathbb{R}^{|\Lambda^{\complement}|-1}$ and $z\in\mathbb{R}$. 
\end{proof}
\noindent It becomes evident that Prop.\ \ref{prop1} facilitates a synthesis (generative) representation of an analysis subspace which establishes a first relation between the two models. 
\begin{rmk} Provided $|\Lambda|<N$, we discover that the matrix ${\bf \Psi}_{\Lambda}{\bf L}$ has full row-rank $|\Lambda|$, and accordingly the basis for $N({\bf \Psi}_{\Lambda}{\bf L})$ has column rank $N-|\Lambda|=|\Lambda^{\complement}|$ and is spanned by $N({\bf L})$ of rank 1 and, for $|\Lambda|<N-1$, ${\bf L}^{\dagger}{\bf \Psi}^T_{\Lambda^{\complement}}{\bf W}$ of rank $|\Lambda^{\complement}|-1$.\end{rmk}
A closer examination of the F.A. constraint on the solution subspaces in Prop.\ \ref{prop1} reveals a fundamental connection to the transposed incidence matrix ${\bf S}^T$, whose Gram matrix is given by ${\bf L}$: 
\begin{rmk}\label{rems} The constraint imposed by the F.A. is encapsulated in the zero-sum column structure of ${\bf \Psi}^T_{\Lambda^{\complement}}{\bf W}$. The same constraint is mirrored by the columns of ${\bf S}^T\in\mathbb{R}^{|V|\times|E|}$, up to a weight factor $\sqrt{A_{i,j}}$ per column. In particular, the F.A. constraint can be alternatively (and more sparsely) expressed as a basis in $({\bf e}_i-{\bf e}_j)$ for any $i,j\in\Lambda^{\complement}\subset V$ which are connected via a path; since the graph is connected, this is stringently satisfied for any vertex pair in $V$. 
In fact, since there always exists a path between any pair of vertices in a connected graph, we can interchangeably express the constraint both with respect to ${\bf S}^T$ and basis ${\bf S}_P^T\in\mathbb{R}^{|V|\times |V|-1}$, which is the incidence matrix of the simple path graph. Hence, we have ${\bf \Psi}^T_{\Lambda^{\complement}}{\bf W}{\bf c}={\bf S}^T{\bf t}={\bf S}^T_P\tilde{{\bf t}}$ for suitable ${\bf c}\in\mathbb{R}^{|\Lambda^{\complement}|-1}$, ${\bf t}\in\mathbb{R}^{|E|}$, and $\tilde{{\bf t}}\in\mathbb{R}^{|V|-1}$, signifying a broad representation range.
\end{rmk}
\noindent For higher order operators ${\bf L}^k$, the generalization $N({\bf \Psi}_{\Lambda}{\bf L}^k)=z{\bf 1}_N +{\bf L}^{\dagger k}{\bf \Psi}_{\Lambda^{\complement}}^T{\bf W}{\bf c}$ trivially follows. Following the proof of Prop.\ \ref{prop1} and specifically employing the F.A. constraint, we can further provide a qualitative decomposition of the known graph theoretical result $N({\bf S}_{\Lambda})=span\{{\bf 1}_{C_1},...,{\bf 1}_{C_t}\}$ for the incidence matrix: 
\begin{cor}\label{incnull}
Consider the disconnected graph $G_{\Lambda}=(V, \Lambda)$ with $t$ connected components, which arises from deleting all edges in $\Lambda^{\complement}\subset E$ with $E=\Lambda \cup \Lambda^{\complement}$, and let $\{C_k\}_{k=1}^t$ denote the corresponding vertex sets. Then we can express $N({\bf \Psi}_{\Lambda}{\bf S})=z{\bf 1}_N+span(\sum _{k\in C_i}{\bf S}^{\dagger}{\bf S}_k,\ i\in \lbrack 1\enskip t\rbrack)$, where the second term is a subspace of dimension $t-1$, whose spanning set may be taken over any $t-1$ components $C_i$.
\end{cor}
\begin{proof}
We have $N({\bf S})=z{\bf 1}_N\subset N({\bf \Psi}_{\Lambda}{\bf S})$ for $z\in\mathbb{R}$, as before. Further, we need to solve ${\bf \Psi}_{\Lambda}{\bf S} {\bf u}={\bf 0}_{\Lambda}$, or ${\bf S}{\bf u}=N({\bf \Psi}_{\Lambda})$, which has a solution ${\bf u}$ iff the F.A. constraint ${\bf \Psi}_{\Lambda^{\complement}}^T{\bf w} \perp {\bf n}_j, \forall {\bf n}_j\in N({\bf S}^H)$ is satisfied for some ${\bf w}\in\mathbb{R}^{|\Lambda^{\complement}|}$. \\
Since for columns ${\bf S}_k$, we have ${\bf S}_k\perp {\bf n}_j$, we require ${\bf \Psi}_{\Lambda^{\complement}}^T{\bf w}=\sum _{k\in \tilde{V}}{\bf S}_k$ for a suitable subset of vertices $\tilde{V}\subset V$ such that the support of $\sum _{k\in \tilde{V}}{\bf S}_k$ is in $\Lambda^{\complement}$. \\
In particular, we have that any unweighted sum of columns ${\bf S}_k$, respectively associated with indexed vertex $k\in \tilde{V}$, has zeros at edge positions which connect two vertices in the set, and non-zeros for every edge that only has one vertex in $\tilde{V}$. In order for this support to match the locations of the removed edge set $\Lambda^{\complement}$, and assuming the connected components $C_k$ of the graph $G_{\Lambda}$ are known, we need to consider the sums $\sum _{k\in C_i}{\bf S}_k,\enskip C_i\subset V$.
We thus have ${\bf u}\in span(\sum _{k\in C_i}{\bf S}^{\dagger}{\bf S}_k,\ i\in \lbrack 1\enskip t\rbrack)$, which is orthogonal to $N({\bf S})$; in order for the vectors to form a basis we require only $t-1$ components since ${\bf S}{\bf 1}_N={\bf 0}_{|E|}$.\\
It follows that the solution set is given by  $N({\bf S}_{\Lambda})=z{\bf 1}_N +span(\sum _{k\in C_i}{\bf S}^{\dagger}{\bf S}_k,\ i\in \lbrack 1\enskip t\rbrack)$, the second term spanning any $t-1$ components, from which it becomes evident that $\sum _{k\in C_i}{\bf S}^{\dagger}{\bf S}_k=\sum _{k\in C_i}({\bf I}_N-\frac{1}{N}{\bf J}_N)_k$ is piecewise constant, i.e. constant over each connected component in $C_i$. Combined with $N({\bf S})$, this can be transformed into the known basis $N({\bf S}_{\Lambda})=span\{{\bf 1}_{C_1},...,{\bf 1}_{C_t}\}$.
\end{proof}
\noindent It trivially follows from the above that $N({\bf \Psi}_{\Lambda}{\bf S}{\bf L}^k)={\bf 1}_N+{\bf L}^{\dagger k} span(\sum _{k\in C_i}{\bf S}^{\dagger}{\bf S}_k,\ i\in \lbrack 1\enskip t\rbrack)$ for the span over any $t-1$ connected components $C_i$, must hold for higher-order operators ${\bf S}{\bf L}^k$.\\
\\
{\bf Disconnected Graphs.}
Consider a graph $G=(V,E)$ which is not connected, instead consisting of $t$ connected components (subgraphs) $G_k=(C_k,E_k)$ with vertex sets $C_k$ such that $V=\bigcup_{k=1}^t C_k$. Hence, the analysis subspace $N({\bf \Psi}_{\Lambda}{\bf L})$ is modified by the underlying rank-deficiency and linear dependency structure of its graph Laplacian ${\bf L}$ as follows:
\begin{cor}\label{disc1}
The nullspace $N({\bf \Psi}_{\Lambda}{\bf L})$ for ${\bf L}$ of a disconnected graph $G=(V,E)$ with $t$ connected components and corresponding vertex sets $\{C_k\}_{k=1}^t$ is spanned by $N({\bf L})=\{{\bf 1}_{C_1},...,{\bf 1}_{C_t}\}$ of rank $t$ and ${\bf L}^{\dagger}{\bf \Psi}_{\Lambda^{\complement}}^T{\bf W}$, with ${\bf W}$ given by 
\begin{equation}\label{eq:newnull}{\bf W}=\begin{bmatrix} {\bf W}_1 & 0 &\dots&\\0 & {\bf W}_2 &0 &\dots\\ \dots&&&\\0 & \dots & &{\bf W}_t\\
\end{bmatrix},\end{equation}
of rank at least $|\Lambda^{\complement}|-t$, with ${\bf W}_k\in\mathbb{R}^{|\Lambda_k^{\complement}|\times|\Lambda_k^{\complement}|-1}$ as in Eq.\ (\ref{eq:const}) and $C_k=\Lambda_k^{\complement}\cup\Lambda_k$. If for each subset $\Lambda_k^{\complement}\subset \Lambda^{\complement}$ we assume $|\Lambda^{\complement}_k|\geq 1$, the rank becomes exact. 
\end{cor}
\noindent \textit{Proof.} See \ref{appa}
\begin{rmk}
If there exists some set $\Lambda_k^{\complement}=\emptyset$, ${\bf W}_k$ is empty and the dimension of ${\bf W}\in\mathbb{R}^{|\Lambda^{\complement}|\times |\Lambda^{\complement}|-t+1}$ increases accordingly by one, so that ${\bf L}^{\dagger}{\bf \Psi}_{\Lambda^{\complement}}^T{\bf W}$ has rank $|\Lambda^{\complement}|-t+1$. 
For ${\bf W}_k$ to be non-empty, we require $|\Lambda_k^{\complement}|\geq 2$.
\end{rmk}
\noindent Further to Rem.\ \ref{rems}, for disconnected graphs the constrained solution space ${\bf L}^{\dagger}{\bf \Psi}_{\Lambda^{\complement}}^T{\bf W}$ can be alternatively expressed as 
\[{\bf L}^{\dagger}{\bf \Psi}_{\Lambda^{\complement}}^T{\bf W}{\bf c}=\begin{bmatrix} {\bf L}_1^{\dagger}{\bf S}_1^T & 0 &\dots&\\0 & {\bf L}_2^{\dagger}{\bf S}_2^T &0 &\dots\\ \dots&&&\\0 & \dots & & {\bf L}_t^{\dagger}{\bf S}_t^T\\
\end{bmatrix}\begin{bmatrix}{\bf t}_1\\{\bf t}_2\\ \dots \\{\bf t}_t\\\end{bmatrix}\]
with corresponding ${\bf c}\in\mathbb{R}^{|\Lambda^{\complement}|-t}$, where ${\bf S}_i$ denote the incidence matrices of ${\bf L}_i={\bf S}_i^T{\bf S}_i$ per connected component, and the block-wise coefficient vectors ${\bf t}_i\in\mathbb{R}^{|E_i|}$ are chosen such that ${\bf S}^T_i{\bf t}_i=\tilde{{\bf \Psi}}_{\Lambda_i^{\complement}}^T{\bf W}_i{\bf c}_i$ for suitable ${\bf c}_i\in\mathbb{R}^{|\Lambda_i^{\complement}|-1}$ and sampling matrix $\tilde{{\bf \Psi}}_{\Lambda_i}\in\mathbb{R}^{|\Lambda_i|\times |C_i|}$. Further, we can more sparsely express each subgraph constraint $\tilde{{\bf \Psi}}_{\Lambda_i^{\complement}}^T{\bf W}_i$ as a basis in ${\bf e}_k-{\bf e}_j$ for any $k,j\in\Lambda_i^{\complement}\subset V$. Overall, the sparse vector of coefficients ${\bf \Psi}_{\Lambda^{\complement}}^T{\bf W}{\bf c}$ becomes block-wise sparse where each block (connected subgraph) with vertex set $C_i=\Lambda_i\cup \Lambda^{\complement}_i$ is associated with sparsity $|{\Lambda_i^{\complement}}|$ and its coefficients sum to zero, as a result of the inherent constraint of each set of basis functions ${\bf W}_i$. This type of graph-structured sparsity will be directly leveraged in the construction of desirable UoS models in Sects. $3.4$ and $5$.

\subsection{The Sparse Synthesis Model}
We proceed with the analysis of synthesis representation ${\bf x}={\bf L}^{\dagger}{\bf c}={\bf L}_{\Lambda^{\complement}}^{\dagger}{\bf c}_{\Lambda^{\complement}}$, where ${\bf c}\in\mathbb{R}^N$ is sparse with vertex support $\Lambda^{\complement}\in V$ on a connected graph $G=(V,E)$ and ${\bf D}={\bf L}^{\dagger}$ is a dictionary of graph signals, and explore its possible solution subspaces for different choices  of ${\bf c}$ as well as higher-order generalizations of ${\bf L}^{\dagger}$.\\
\\
In an effort to provide a more intuitive characterization of the synthesis subspace with span ${\bf L}_{\Lambda^{\complement}}^{\dagger}$, we consider the specific notion of \textit{discrete Green's functions}  and its relation to the MPP. According to \cite{green1}, the \textit{Green's function} (or \textit{matrix}) of an invertible operator constitutes its inverse; in case of a square singular operator ${\bf L}$, the Green's functions need to satisfy the conditions ${\bf L}{\bf L}^{\dagger}={\bf L}^{\dagger}{\bf L}={\bf I}_N-{\bf P}^H{\bf P}$, and ${\bf L}^{\dagger}{\bf P}={\bf 0}$, where ${\bf P}$ is the matrix of eigenvectors associated with the zero eigenvalue (or $N({\bf L})$), in order to be uniquely defined. This establishes a relation to the MPP, which can be interpreted as constrained Green's functions \cite{plonka2}.\\ 
In light of this interpretation, the synthesis representation ${\bf x}$ constitutes a linear combination of Green's functions\footnote{For simplicity, we will refer to both the matrix as well as its individual columns as Green's functions.}. Furthermore, when ${\bf x}$ is constrained to be sparse with respect to ${\bf L}$, the previous MPP constraints, in conjuction with the Gramian structure of ${\bf L}={\bf S}^T{\bf S}$, reveal an inherent constrained, or rather, \textit{structured sparsity} pattern for coefficient vector ${\bf c}$, whose non-zeros correspond to the locations of discontinuities of the underlying Green's functions. \\
In particular, we have ${\bf L}({\bf L}^{\dagger}{\bf S}_i^T)={\bf L}({\bf S}_i^{\dagger})={\bf S}_i^T$, which entails that sparse linear combinations of the elementary Green's functions (columns) ${\bf S}_i^{\dagger}$ are sparse with respect to analysis operator ${\bf L}$ with structured sparsity in the range of ${\bf S}^T$; more generally, any number of linear combinations of $\{{\bf S}_i^{\dagger}\}_{i\in E_S}$, whose index locations (sequence of edges $E_S$) form connected paths on $G$, can be combined to yield a sparse output. Since the columns of ${\bf S}^T$ are $2$-sparse this further implies that any piecewise smooth signal is at least $2$-sparse with respect to ${\bf L}$. It becomes evident that the analysis operation ${\bf L}{\bf x}={\bf c}$ directly characterizes the constrained synthesis representation ${\bf x}={\bf L}^{\dagger}\sum_{j\in E_S}{\bf S}^T_j=\sum_{j\in E_S}{\bf S}^{\dagger}_j$ with sparse pattern ${\bf c}=\sum_{j\in E_S}{\bf S}^T_j$. 
This insight, along with generalizations to higher order operators, is summarized in the following Lemma:
\begin{lem}\label{summ}
The signal ${\bf x}$ on connected graph $G=(V,E)$, which is piecewise-smooth with respect to a designated graph difference operator, can be characterized through weighted combinations of elementary Green's functions of the operator, for suitable weights $w_j\in\mathbb{R}$, and resulting  structured sparse vectors, as follows:\\
$(i)$ The signal ${\bf x}=\sum_{j\in E_S}w_j{\bf L}^{\dagger k}{\bf S}_j^T=\sum_{j\in E_S} w_j{\bf L}^{\dagger k-1}{\bf S}_j^{\dagger}$, for suitable edge subset $E_S\subset E$, is sparse with respect to ${\bf L}^k$ with sparse vector $\sum_{j\in E_S}w_j{\bf S}_j^T\in\mathbb{R}^{|V|}$.\\
$(ii)$ The signal ${\bf x}=\sum_{j\in V_S}w_j({\bf L}^{\dagger k})_j$, for suitable vertex subset $V_S\subset V$, is sparse with respect to ${\bf S}{\bf L}^k$ with sparse vector $\sum_{j\in V_S}w_j{\bf S}_j\in\mathbb{R}^{|E|}$.\\ 
Further, if $G$ is sufficiently sparse with $||{\bf c}||_0\ll N$ for ${\bf c}=\sum_{j\in V_S}w_j{\bf L}_j$:\\
$(iii)$ The signal ${\bf x}=\sum_{j\in V_S}w_j({\bf L}^{\dagger k-1})_j=\sum_{j\in V_S}w_j{\bf L}^{\dagger k}{\bf L}_j$, for suitable vertex subset $V_S\subset V$, is sparse with respect to ${\bf L}^{k}$ with sparse vector $\sum_{j\in V_S}w_j{\bf L}_j\in\mathbb{R}^{|V|}$. %
\end{lem}
\noindent \textit{Proof}. See Appendix.\\
\\
The piecewise-smooth graph signals in Lemma \ref{summ} can also be interpreted as the splines of their defining operator \cite{splinesn}, with knots given by the support of the induced sparse pattern. Since we are exclusively operating in the discrete domain and the operators are rank deficient, the connection with splines originating from continuous Green's functions is not strict.\footnote{In contrast to classical Green's functions and splines, we note that while the discrete Green's functions of ${\bf L}$ are given by ${\bf L}^{\dagger}$, its splines are in the range of ${\bf S}^{\dagger}={\bf L}^{\dagger}{\bf S}^T$, marking a difference in order.} 
For simplicity, we will use the term \textit{knots} here more generally to refer to the discontinuities of (any linear combination of) discrete Green's functions of rank-deficient operators.  As an aside, the nullspace of the operators at hand also forms part of the space of piecewise-smooth signals but is not directly associated to knots.\\
\noindent It thus becomes evident that the non-zero entries of the columns of ${\bf S}^T$, and, where applicable, of ${\bf S}$ and ${\bf L}$, assume the role of \textit{knots} of the Green's functions and thereby define an underlying structured sparsity model for graphs. Thus choosing ${\bf c}$ in the range of ${\bf S}^T$ (and if applicable, ${\bf L}$) provides a constrained synthesis representation of the form ${\bf x}={\bf L}^{\dagger k}{\bf c}$ which is analysis-sparse, and therefore also contained in the corresponding analysis model. 
\begin{rmk} \label{typerm}The relation ${\bf L}^{\dagger}{\bf L}={\bf I}_N-\frac{1}{N}{\bf J}_N$ constitutes the projection onto $N({\bf L})^{\perp}$. Accordingly, we deduce that the synthesis representations ${\bf L}^{\dagger}{\bf L}_j$ and ${\bf L}^{\dagger}{\bf S}^T_j={\bf S}^{\dagger}_j$, which can be respectively annihilated by ${\bf L}$ with resulting patterns ${\bf L}_j$ (if $G$ is sufficiently sparse) and ${\bf S}^T_j$, encapsulate different orders of smoothness. Here, ${\bf L}^{\dagger}{\bf L}_j$, which is also sparse with respect to the lower order operator ${\bf S}$ with pattern ${\bf S}_j$, is one degree less smooth than ${\bf S}^{\dagger}_j$ and two degrees less smooth than ${\bf L}^{\dagger}_j$. This interpretation can be generalized to higher order operators. \end{rmk}
The synthesis representation ${\bf x}$ is further sparse with respect to the analysis operator ${\bf S}{\bf L}^k$ with sparsity pattern in the range of ${\bf S}$; nevertheless, its basis atoms, consisting of edge signals ${\bf S}_j\in\mathbb{R}^{|E|}$, do not lie in the same domain, providing an intermediate analysis stage. In particular, ${\bf S}{\bf L}^{k-1}$ and ${\bf L}^k$ signify incremental orders of annihilation, which will be more clearly substantiated for circulant graphs in Sect. $4$. For arbitrary ${\bf c}$, provided the graph at hand sufficiently sparse, the analysis of ${\bf x}$ with ${\bf L}^{k+m}$ and ${\bf S}{\bf L}^{k+m}$ for $m,k\in\mathbb{Z}^{\geq 0}$, creates the constrained sparse vectors ${\bf L}^{m}{\bf c}$ and ${\bf S}{\bf L}^{m}{\bf c}$ respectively.

\subsection{The graph Laplacian UoS model}
In light of previous derivations, we summarize the structural relation between the graph Laplacian-based signal models for a connected graph:
\begin{thm}\label{thm331} For a connected graph, the cosparse analysis model, described by $N({\bf \Psi}_{\Lambda}{\bf L})$, constitutes a constrained case of the sparse synthesis model, generated by ${\bf L}^{\dagger}{\bf \Psi}_{\Lambda^{\complement}}^T$, with respect to matrix ${\bf W}\in\mathbb{R}^{|{\Lambda^{\complement}}|\times {\Lambda^{\complement}}-1}$ (see Eq.\ (\ref{eq:const})) which satisfies $\langle{\bf \Psi}_{\Lambda^{\complement}}^T{\bf W}{\bf e}_j {,} N({\bf L})\rangle=0,\ \ j=0, ..., |\Lambda^{\complement}|-2$, up to a translation by $N({\bf L})=z{\bf 1}_N,\ z\in\mathbb{R}$. 
\end{thm}
\noindent \textit{Proof.}  Follows from previous discussion.\\
\\
Notably, the atoms of the synthesis subspaces (and of its constrained counterpart in the analysis model), are orthogonal to the translation vector $N({\bf L})$, confining its solution space.\\
\\
More formally, the signals ${\bf x}$ which satisfy the model constraint $||{\bf L}{\bf x}||_0=N-l$ (or ${\bf L}_{\Lambda}{\bf x}={\bf 0}_{\Lambda}$) for cosupport $\Lambda$ are elements of the \textit{analysis} union of subspaces $\bigcup_{\Lambda}\textit{W}_{\Lambda}$ over all subspaces $W_{\Lambda}$, defined as $W_{\Lambda}:=N({\bf L}_{\Lambda})$, of cardinality $|\Lambda|=l$, i.e. the subspaces of all signals which have cosparsity $l$ with respect to ${\bf L}$. Given a dictionary ${\bf L}^{\dagger}$, the vectors ${\bf x}$ that satisfy ${\bf x}={\bf L}^{\dagger}{\bf c}$ with $||{\bf c}||_0=k$ (or ${\bf x}={\bf L}^{\dagger}_{\Lambda^{\complement}}{\bf c}_{\Lambda^{\complement}}$), are elements of the \textit{synthesis} union of subspaces $\bigcup_{\Lambda^{\complement}}V_{\Lambda^{\complement}}$, where $V_{\Lambda^{\complement}}$ are $k$-dimensional subspaces spanned by $k$ columns from ${\bf L}_j^{\dagger}$, i.e. $V_{\Lambda^{\complement}}=span({\bf L}^{\dagger}_j ,j\in\Lambda^{\complement})$. \\
In direct comparison, consider the subspace $V_{\Lambda^{\complement}}$ of dimension $|\Lambda^{\complement}|=k$; for arbitrary support $\Lambda^{\complement}$ of fixed cardinality $k$, the synthesis model defines the union $\bigcup_{|\Lambda^{\complement}|=k} \textit{V}_{\Lambda^{\complement}}$ comprising ${N}\choose{k}$ subspaces. In contrast, the analysis model, excluding the nullspace $N({\bf L})$ and provided $|\Lambda|<N-1$, induces the subspace $W_{{\Lambda}}$ spanned by the columns ${\bf L}^{\dagger}({\bf e}_i-{\bf e}_j),\ i,j\in\Lambda^{\complement}$ (see Rem. \ref{rems}), whose union $\bigcup_{|\Lambda|=N-k} \textit{W}_{\Lambda}$, for arbitrary cosupport $\Lambda$ of fixed cardinality $l=N-k$, comprises the same number of subspaces ${N}\choose{k}$ of reduced dimension $k-1$. It becomes evident that, despite generating the same number of unique subspaces\footnote{Here, we assume for simplicity, the absence of degenerate cases, resulting from further linear dependencies.}, the analysis UoS model describes only a subset of elements of the synthesis UoS model as a result of the rank deficiency constraint and associated dimensionality reduction.   
In this case, we conclude $\bigcup_{|\Lambda|=N-k} \textit{W}_{\Lambda}\subseteq \bigcup_{|\Lambda^{\complement}|=k} \textit{V}_{\Lambda^{\complement}}$, which is in line with the result for operators in general position \cite{cos}.\\
Nevertheless, when $N({\bf L})$ of dimension $1$ is taken into account, the analysis model is translated and no longer contained in the synthesis model, while their respective subspace dimensions become identical. For the previous nesting property to hold, one therefore needs to consider the projection of the analysis subspaces onto the range of ${\bf L}$ via ${\bf L}{\bf L}^{\dagger}$. Table \ref{tab1} summarizes the differences in subspace distribution of the two models. Here, it is noteworthy that the subspace dimensions and number for both models are equivalent, just as in the case of an invertible square operator, except when $k=N-l=1$. 
\begin{table}[t]
\centering
\scalebox{0.85}{
\begin{tabular}{||c| c c c|| c c c||} 
\hline
&&Synthesis &&&Analysis&\\
 \hline
Sparsity & Dim. & Subsp. & No. & Dim. & Subsp. & No.\  \\
 \hline\hline
 $1$ &  $1$ & ${\bf L}^{\dagger}_j$ & $N$ & $1$ & ${\bf 1}_N$ & $1$  \\
 &&&&&&\\
$2$ & $2$ & $span({\bf L}^{\dagger}_j,j\in\Lambda^{\complement})$ & ${N}\choose{2}$ & $2$ & $span({\bf 1}_N;{\bf L}^{\dagger}({\bf e}_i-{\bf e}_j),i,j\in\Lambda^{\complement})$ & ${N}\choose{2}$ \\
&&&&&&\\
$k\ll N$& $k$& $span({\bf L}^{\dagger}_j,j\in\Lambda^{\complement})$  &${N}\choose{k}$ &$k$ & $span({\bf 1}_N;{\bf L}^{\dagger}({\bf e}_i-{\bf e}_j),i,j\in\Lambda^{\complement})$ &${N}\choose{k}$\\ 
\hline
\end{tabular}}
\caption{Subspace Characterization of ${\bf L}$ for a Connected Graph}
\label{tab1}
\end{table}
\\
{\bf Disconnected Graphs}.
When the graph is disconnected with $t$ components, the number of subspaces and their dimension is no longer globally consistent and depends on the distribution of the vertex sets $C_i=\Lambda_i\bigcup \Lambda_i^{\complement}$ per connected component of cardinality $N_i=|C_i|$, with $N=\sum_{i=1}^t N_i$ and $\Lambda=\bigcup_{i=1}^t \Lambda_i$.\\
For simplicity, let $t=2$ and assuming $k<N_i$ for the sparsity level $k$, the synthesis UoS model consists of ${N\choose{k}}$ possible unique subspaces $V_{\Lambda^{\complement}}$, with $|{\Lambda^{\complement}}|=k$, of dimension $k$.\\
Conversely, the analysis UoS model, excluding $N({\bf L})$, of subspaces ${W}_{\Lambda}$ with cosupport cardinality $|{\Lambda}|=l=N-k$, is split into ${{N_1}\choose{k-1}}+{{N_2}\choose{k-1}}+\sum_{k_1=2}^{k-2}{{N_1}\choose{k_1}}{{N_2}\choose{k-k_1}}$ subspaces of dimension $k-2$ and ${{N_1}\choose{k}}+{{N_2}\choose{k}}$ subspaces of dimension $k-1$, whose total sum is less than ${N\choose{k}}$, marking a reduction in comparison to the synthesis model. Further, the ${{N_1}\choose{k-1}}+{{N_2}\choose{k-1}}$ subspaces of dimension $k-1$ in fact generate signals of increased cosparsity $l+1=N-k+1$, despite the constraint $|\Lambda|=l$. These phenomena occur as a result of the bases in ${\bf L}_t^{\dagger}({\bf e}_i-{\bf e}_j),\ i,j \in \Lambda_t$ being empty for $|{\Lambda_t^{\complement}}|\leq 1$, and entails that subspace unions for a fixed $|\Lambda|$ share subspaces with others of higher/lower cardinality, but with the same cosparsity.
Here, we have employed the Chu-Vandermonde identity \cite{chu} \[{{N}\choose{k}}=\sum_{k_1=0}^k{{N_1}\choose{k_1}}{{N_2}\choose{k-k_1}}\enskip \text{with}\enskip N=N_1+N_2\enskip\text{and}\enskip k=k_1+k_2,\] to establish a reduction in the number of unique subspaces of fixed dimension, i.e. ${{N}\choose{k}}>{{N_1}\choose{k_1}}{{N_2}\choose{k_2}}$, when the sparsity level $k_i$ per subgraph is fixed. One may exclude the special cases with $k_i=0,1,\enskip i=1,2,$ by imposing $k_i\geq 2$ per connected component, resulting in $\sum_{k_1=2}^{k-2}{{N_1}\choose{k_1}}{{N_2}\choose{k-k_1}}$ subspaces of fixed, uniform dimension $k-2$ and cosparsity $l$. \\
Overall, it becomes evident that while the distribution of subspace dimension and number is location dependent in the analysis model, the total number of subspaces changes as well. When taking into account $N({\bf L})$, the total dimension of the analysis subspaces is equivalent to that of the synthesis subspaces, except for the cases with $|\Lambda_i^{\complement}|=0$, where the total dimension increases to $k+1$ in the analysis model. The comparison for $t=2$ is summarized in Table \ref{tab2}; here, the analysis model includes $N({\bf L})$.\\

\begin{table}[t]
\centering
\scalebox{0.8}{
\begin{tabular}{||c| c  c|| c||} 
\hline
&&Synthesis &Analysis\\
 \hline
$|\Lambda^{\complement}|$ & Dim. & No. & No. $\&$ Dim.  \\
 \hline\hline
 $1$ &  $1$&$N$ &$\# 1$ $N({\bf L})$ of dim. $2$ \\
$2$  & $2$ & ${{N}\choose{2}}$ & $\# 1$ $N({\bf L})$ of dim. $2$, ${{N_1}\choose{2}} +{{N_2}\choose{2}}$ of dim. $3$\\
$3$&$3$&${{N}\choose{3}}$& ${{N_1}\choose{2}} +{{N_2}\choose{2}}$ of dim. $3$, ${{N_1}\choose{3}} +{{N_2}\choose{3}}$ of dim. $4$\\
&&&\\
$k$& $k$&${{N}\choose{k}}$ & ${{N_1}\choose{k-1}}+{{N_2}\choose{k-1}}+\sum_{k_1=2}^{k-2}{{N_1}\choose{k_1}}{{N_2}\choose{k-k_1}}$ of dim. $k$, ${{N_1}\choose{k}}+{{N_2}\choose{k}}$ of dim. $k+1$\\
\hline
\end{tabular}}
\caption{Subspace Characterization of ${\bf L}$ for Graph with $t=2$ Connected Components}
\label{tab2}
\end{table}%
\noindent This analysis can be further generalized to graphs with $t$ connected components, as the formula admits the extension 
\[{{N}\choose{k}}=\sum_{k_1+k_2+...+k_t=k}{{N_1}\choose{k_1}}{{N_2}\choose{k_2}}...{{N_t}\choose{k_t}}.\]
In particular, provided $t<k<N_i$, the UoS of cardinality $|\Lambda|=N-k$, for a graph with $t$ connected components, comprises subspaces ranging from total dimension $k$ to the maximum $k+t-1$, and whose total number is necessarily smaller than ${{N}\choose{k}}$. This can be standardized by applying support constraints and we summarize this scenario in Table \ref{tab3}.\\
\begin{table}[t]
\centering
\scalebox{0.8}{
\begin{tabular}{||c| c c || c c||} 
\hline
&&Synthesis &&Analysis\\
 \hline
$|\Lambda^{\complement}|$ & Dim. & No. & Constr.& No. $\&$ Dim.  \\
 \hline\hline
 $1,...,t$&  $t$&${{N}\choose{t}}$ &$|\Lambda_i^{\complement}|\leq 1$ &$\#1$ $N({\bf L})=\{{\bf 1}_{C_1},...,{\bf 1}_{C_t}\}$ of dim. $t$ \\
 &&&&\\
$t<k<N_i$& $k$&${{N}\choose{k}}$ &  $|\Lambda_i^{\complement}|\geq 2$ & $\#1$ $N({\bf L})$ of dim. $t$ and \\
&&&&$\sum_{k_1+...+k_t=k,\ 2\leq k_i\leq k-2}{{N_1}\choose{k_1}}...{{N_t}\choose{k_t}}$ of dim. $k-t$\\
\hline
\end{tabular}}
\caption{Subspace Characterization of ${\bf L}$ for Graph with $t$ Connected Components (Constrained)}
\label{tab3}
\end{table}%
Further to Tables \ref{tab1} and \ref{tab2}, we note that the structure of the analysis subspaces constitutes a `reduced' version, and in the special case of circulant graphs, as will be shown in Sect.\ $4$, a `degree-reduced' version, of the synthesis subspaces by a step-size of $1$ for $t=1$ (see Rem. \ref{typerm}). For $t>1$, the analysis subspaces become piecewise-'reduced', consisting of $t$ pieces. Specifically, due to the block-diagonal structure of ${\bf L}^{\dagger}$, the generated subspaces in both models define piecewise distributed graph signals, which are piecewise-smooth with respect to the individual subgraphs; the discrepancy between the two models is maintained through the imposed block-wise constraints in the analysis model, which induce a (rank-)reduction. In general, it would appear that the higher the rank-deficiency of ${\bf L}$ (associated with decreased connectivity), the more structurally constrained its subspaces. \\
 \\
The linear constraints imposed in form of a rank-deficient analysis operator define a UoS signal model of piecewise smooth (Green's) functions with structured discontinuities which separate functions of (possibly) different type, following the application of constraints of variable order (see Rem.\ \ref{typerm}). 
Further, the basis in $({\bf e}_i-{\bf e}_j), \  i,j\in{\Lambda^{\complement}}$ of the analysis model, signifies a structured sparsity model whose linear dependencies (constraints) affect both the location and value of its coefficients; this will be further enlarged upon in Sect. $5$. 

\section{The (Co)sparse Signal Model on Circulant Graphs}
\subsection{Motivation: From Analysis to Synthesis}
Previous work on wavelets and sparsity on circulant graphs \cite{splinesw}, has focused on the derivation of annihilation properties of circulant graph difference operators, including the graph Laplacian ${\bf L}$ and the novel, generalized graph Laplacian design ${\bf L}_{\alpha}$. Serving as the pillars of the cosparse analysis model on circulant graphs, we proceed to show how these insights, in combination with derived closed-form expressions for the Green's functions of the (generalized) graph Laplacian in form of MPP ${\bf L}^{\dagger}$ (${\bf L}_{\alpha}^{\dagger}$), and the theory of Sect.\ $3$, can be extended to form a comprehensive and uniquely characterizable union of subspaces model. In an effort to further comprehend transitional properties, we further demonstrate how the arising discrepancies between the analysis and synthesis models can be concretely exemplified on the basis of the parametric generalized graph Laplacian ${\bf L}_{\alpha}$, which is nonsingular for certain $\alpha$ and singular otherwise. We initiate the discussion by intuitively demonstrating the relation between an analysis and synthesis view of the circulant graph Laplacian:\\
\\
Let ${\bf x}\in\mathbb{R}^N$ denote a linear polynomial signal, as per Def.\ $2.2$, with slope $\epsilon$ and a discontinuity between $x(0)$ and $x(N-1)$ on an undirected banded circulant graph of dimension $N$ with associated graph Laplacian ${\bf L}$, and consider the graph measurement vector ${\bf y}={\bf L}{\bf x}$. In \cite{splinesw} it was established that the representer polynomial $l(z)$ of ${\bf L}$ has two vanishing moments, i.e. it annihilates linear polynomials subject to a border effect dependent on the bandwidth of the graph. With $N({\bf L})={\bf 1}_N$, for the system to yield a solution, as per the Fredholm Alternative, we require $\langle {\bf y},{\bf 1}_N \rangle =0$. We directly evaluate ${\bf y}$ via simple algebra to be of the form
\[{\bf y}^T=\begin{bmatrix}
y_1 & y_2&\dots & y_M &0 & \dots 0 & -y_M &\dots&-y_2& -y_1\end{bmatrix}^T
\]
where $y_i=-N(\sum_{j=i}^M d_j)\epsilon$, and discover that the F.A. trivially holds. \\
\\
Let ${\bf L}^{\dagger}$ denote the MPP of ${\bf L}$, which is also symmetric and circulant, with columns ${\bf L}_j^{\dagger}$; accordingly, the solution ${\bf x}'$ to the graph Laplacian system can be more generally written in a \textit{sparse, constrained synthesis} formulation as:
\begin{equation}
{\bf x}'={\bf L}^{\dagger}{\bf y}+N({\bf L})=\sum_{j=1}^M y_j ({\bf L}_{j-1}^{\dagger}-{\bf L}_{N-j}^{\dagger})+z{\bf 1}_N,\enskip z\in\mathbb{R}.
\end{equation}
The particular structure of ${\bf y}$ and ${\bf L}^{\dagger}$ elucidate an interesting phenomenon, specific to circulant matrices: the column pairs $\{{\bf L}_j^{\dagger},{\bf L}_{N-j-1}^{\dagger}\}_j$ of ${\bf L}^{\dagger}$, whose pairwise difference is considered, constitute flipped versions of each other. In particular, when ${\bf L}_j^{\dagger}$ is a discrete polynomial with even order exponent $2k$, its vertically flipped version ${\bf L}_{N-j-1}^{\dagger}$ represents its reflection across a vertical line so that their difference removes all terms of that order; this is not satisfied in the case of  odd exponents. \\
In order to deduce the exact nature of the basis functions ${\bf L}_j^{\dagger}$, we proceed with the analysis of the simple cycle, as the base case for general circulant graphs, which will lay the foundation of all subsequent analysis.

\subsubsection{The simple cycle}
Consider the simple cycle graph $G_C=(V,E_C)$ with edge weight set to $d_1=1$, without loss of generality, and Laplacian representer polynomial $l_C(z)=(2-(z+z^{-1}))$, which corresponds to two vanishing moments. We deduce from the annihilation property of $l_C(z)$ and, hence, expression ${\bf x}=y_1 ({\bf L}_{C, 0}^{\dagger}-{\bf L}_{C, N-1}^{\dagger})+z{\bf 1}_N,\ z\in\mathbb{R}$ from Sect.\ $4.1$, that the difference between consecutive entries $\tilde{l}_i=L_C^{\dagger}(j,(i+j)_{N})$ of ${\bf L}_C^{\dagger}$, such as in
\[{\bf L}_{C, 0}^{\dagger}-{\bf L}_{C, N-1}^{\dagger}=\begin{bmatrix}
\tilde{l}_0-\tilde{l}_1\\
\tilde{l}_1-\tilde{l}_2\\
... \\
-(\tilde{l}_1-\tilde{l}_2)\\
-(\tilde{l}_0-\tilde{l}_1)\\
\end{bmatrix}
\]
must be linearly decreasing, which implies that ${\bf L}_{C, 0}^{\dagger}$ and ${\bf L}_{C, N-1}^{\dagger}$ must be of a quadratic polynomial form. In fact, an explicit expression for the MPP pseudoinverse ${\bf L}_C^{\dagger}$ of the simple cycle graph Laplacian has been derived as the sum of quadratic and piecewise linear terms:
\begin{propp}\label{propp1} (see Thm.\ $1$, \cite{green2}). The pseudoinverse ${\bf L}_C^{\dagger}$ of the (unnormalized) graph Laplacian of the simple cycle has entries
\newline $L_C^{\dagger}(i,j)=\frac{(N-1)(N+1)}{12N}-\frac{1}{2}|j-i|+\frac{(j-i)^2}{2N},\enskip\text{for}\enskip 0\leq i,j\leq N-1.$
\end{propp}
\noindent We discover that on the basis of this expression, one may re-derive the annihilation property of ${\bf L}_C$ and, by extension, ${\bf L}$. In particular, consider the columns $L_C^{\dagger}(i,j_1)$ and $L_C^{\dagger}(i,j_2)$, with variable $i$, $0\leq i \leq N-1$, and fixed but arbitrary $j_1$, $j_2$ $\in\lbrack 0\enskip N-1\rbrack$, such that
\[L_C^{\dagger}(i,j_1)-L_C^{\dagger}(i,j_2)=-\frac{1}{2}|j_1-i|+\frac{1}{2}|j_2-i|+\frac{(j_1^2-j_2^2)-2i(j_1-j_2)}{2N}\]
which is (piecewise) linear in $i$. Specifically, for $j_2<j_1$ (wlog), we have 
\begin{equation}\label{eq:diff}L_C^{\dagger}(i,j_1)-L_C^{\dagger}(i,j_2)=\frac{(j_1^2-j_2^2)-2i(j_1-j_2)}{2N}+\scalebox{0.8}{$\left \{
  \begin{aligned}
    &\frac{1}{2}(j_2-j_1), && 0\leq i\leq j_2\\
&i-\frac{1}{2}(j_2+j_1), &&\ j_2< i\leq j_1\\
&-\frac{1}{2}(j_2-j_1), &&j_1< i\leq N-1\\
\end{aligned} \right.$}
\end{equation}
with discontinuities whose location depends on the specific choice of $j_1,j_2$. For increasing $j_2$ and decreasing $j_1$, this produces a $2$-piece linear polynomial with respective pieces on $j_2< i\leq j_1$, and, following circularity, on the joint interval of $j_1< i\leq N-1$ and $0\leq i\leq j_2$. When $j_1=N-1$ and $j_2=0$, we obtain for ${\bf L}_{C, N-1}^{\dagger}-{\bf L}_{C, 0} ^{\dagger}$ \begin{equation}\label{eq:sym1}L_C^{\dagger}(i,N-1)-L_C^{\dagger}(i,0)=-\frac{N-1}{2 N}+\frac{i}{N},\ 0\leq i \leq N-1,\end{equation} which returns the $1$-piece linear polynomial with a discontinuity at its endpoints, as the minimum possible number of discontinuities. 
This (anti-)symmetric pattern continues for all pairs chosen $j_1=N-1-j_2$, \\
\[L_C^{\dagger}(i,j_1)-L_C^{\dagger}(i,j_2)=\frac{(2i -N+1)(2 j_2-N+1)}{2N}+ \scalebox{0.8}{$\left \{
  \begin{aligned}
    & {j_2-(N-1)}/{2}, &&\ 0\leq i\leq j_2 \\
    &i-({N-1})/{2}, && j_2< i\leq N-1-j_2\\
    &-(j_2-{(N-1)}/{2}), && N-1-j_2< i\leq N-1\\
  \end{aligned} \right.$}
,\]
which gives rise to a $2$-piece linear polynomial which is anti-symmetric with respect to a vertical line. \\
\\
It becomes evident that simple differences of the Green's functions always annihilate the quadratic term; this further reaffirms our previous observation that the difference of pairs ${\bf L}^{\dagger}_j$ and ${\bf L}^{\dagger}_{N-j-1}$ from a symmetric circulant matrix ${\bf L}^{\dagger}$ induces a degree reduction when they are of polynomial form with even exponents.\\
\\
Moreover, the general signal structure $a L_C^{\dagger}(i,j_1)+b L_C^{\dagger}(i,j_2)$, for $a,b\in \mathbb{R}$ and $a\neq -b$, gives rise to piecewise quadratic polynomials. At last, when coefficient tuples of $3$ are chosen to be of the form of columns of ${\bf L}_C$, we have ${\bf L}_C^{\dagger}{\bf L}_C={\bf I}_N-\frac{1}{N}{\bf J}_N$, constituting a twofold degree reduction to a piecewise constant solution.\\
\\
In line with previous derivations, we state the complete subspace of the cosparse analysis model for the simple cycle graph $G_C=(V,E_C)$ as
\[N({\bf \Psi}_{\Lambda}{\bf L}_C)=z {\bf 1}_N+{\bf L}_C^{\dagger}{\bf \Psi}_{\Lambda^{\complement}}^T{\bf W}{\bf c}=z {\bf 1}_N+\sum_{j\in E_S}t_j({\bf S}^{\dagger}_C)_{j},\] 
where ${\bf S}^{\dagger}_C$ is the MPP of the incidence matrix ${\bf S}_C$ of $G_C$, for suitable edge sequence $E_S\subset E$, weights $t_j\in\mathbb{R}$, $z\in\mathbb{R}$ and coefficient vector ${\bf c}\in\mathbb{R}^{|{\Lambda^{\complement}}|-1}$.\\
\\
{\bf The incidence matrix}. We conclude by deriving a closed-form expression for ${\bf S}^{\dagger}_C$:
\begin{lem}\label{circlins}
On the simple cycle graph, the entries of the MPP ${\bf S}_C^{\dagger}$ of the circulant incidence matrix ${\bf S}_C$ with first row $\lbrack 1\ -1\ 0\ ...\ 0\rbrack$ are given by
\[S_C^{\dagger}(i,j)=\frac{N-1}{2N}-\frac{j-i}{N}\enskip\text{for}\enskip i\leq j,\enskip\enskip 0\leq i,j\leq N-1\]
and give rise to piecewise linear polynomial rows and columns.
\end{lem}
\begin{proof}
From ${\bf S}_C^{\dagger}{\bf S}_C={\bf I}_N-\frac{1}{N}{\bf J}_N$ and letting $S_C^{\dagger}(n):=S_C^{\dagger}(i,j)$ with $n=j-i$ represent each entry as a function of the distance $n$ due to circularity, we obtain the recurrence relation $S_C^{\dagger}(n)-S_C^{\dagger}(n-1)=-\frac{1}{N}$ with boundary condition $S_C^{\dagger}(0)-S_C^{\dagger}(N-1)=1-\frac{1}{N}$ and, from ${\bf S}^{\dagger}{\bf 1}_N={\bf 0}_N$, the constraint $\sum_{n=0}^{N-1} S_C^{\dagger}(n)=0$. The solution is given by $S_C^{\dagger}(n)=-\frac{n}{N}+\frac{N-1}{2N}$. Since ${\bf S}_C$ is circulant (but not symmetric), we conduct the change of variable $n=j-i$
and obtain the desired result. 
While this describes that the rows and columns of the upper-triangular part of ${\bf S}_C^{\dagger}$ are piecewise linear, it follows from circularity that this also holds for the lower-triangular part.
\end{proof}
\begin{rmk} \label{syu} One can alternatively directly infer an expression for ${\bf S}^{\dagger}_C$ via ${\bf L}^{\dagger}_C{\bf S}^T_C={\bf S}^{\dagger}_C$ from 
\[{S}_C^{\dagger}(i,j)={L}_C^{\dagger}{S}_C^{T}(i,j)={L}_C^{\dagger}(i,j)-{L}_C^{\dagger}(i,j+1)=\scalebox{0.85}{$
\left \{
  \begin{aligned}
&\frac{N-1}{2N}-\frac{j-i}{N},&& \ 0\leq i\leq j\\
&     \frac{1-N}{2N},&& \ i=j+1\\
&\frac{-N-1}{2N}-\frac{j-i}{N},&& \ j+1< i\leq N-1\\
  \end{aligned} \right.$}
,\]
where  $0\leq j\leq N-2$. When $j=N-1$, we have ${S}_C^{\dagger}(i,N-1)=-\frac{N-1}{2N}+\frac{i}{N}$ on $0\leq i \leq N-1$ as in Eq. (\ref{eq:sym1}), due to circularity.
\end{rmk}
In contrast to taking differences between arbitrary columns of ${\bf L}_C^{\dagger}$, as seen in Eq.\ (\ref{eq:diff}), which generally results in piecewise linear polynomial solutions with at most two pieces, the special case of Rem.\ \ref{syu} where differences are consecutive, resulting in exclusively $1$-piece linear polynomials, presents an alternative solution space with respect to the incidence matrix.\\
Furthermore, we discover that taking arbitrary differences of the rows of ${\bf S}_C^{\dagger}$ produces piecewise-constant solutions:
\begin{lem}\label{diffind}
The differences between arbitrary rows $i_1$ and $i_2$ of ${\bf S}^{\dagger}_C$, with $i_1<i_2$ wlog, and variable $0\leq j\leq N-1$, produce piecewise-constant row-functions of the form
\[{S}_C^{\dagger}(i_1,j)-{S}_C^{\dagger}(i_2,j)=
\left \{
  \begin{aligned}
&\frac{i_1-i_2}{N},&& \ 0\leq j\leq i_1-1\\
&1+\frac{i_1-i_2}{N},&& \ i_1\leq j\leq i_2-1\\
&\frac{i_1-i_2}{N},&& \ i_2\leq j\leq N-1\\
  \end{aligned} \right.
\]
with zero row-sum.
\end{lem}
\noindent \textit{Proof.} See \ref{appa}.\\
\\
\noindent Similarly, one can show that taking differences between arbitrary columns of ${\bf S}^{\dagger}_C$ produces piecewise constant solutions. In view of Cor.\ \ref{incnull}, it is interesting to observe that the known solution space of $N({\bf \Psi }_{\Lambda}{\bf S}_C)$ with terms of the form ${\bf S}^{\dagger}_C ({\bf S}_C)_j$ is given through the degree-reduction of piecewise linear polynomials. \\
More generally, this entails that the solution ${\bf L}_C^{\dagger}{\bf c}$ subject to the double constraint ${\bf c}=(({\bf e}_i-{\bf e}_j)\ast ({\bf e}_k-{\bf e}_l))_{modN}$, which is the circular convolution modulo $N$ of two zero-sum constraint vectors, and provided at least one pair of $i, j \in V$, $k,l \in V$, is consecutive, gives piecewise constant solutions. The particular choice ${\bf c}=({\bf L}_{C})_j$ of support $3$, which has $2$ vanishing moments in the classical sense, produces the solution with least discontinuities ${\bf L}_C^{\dagger}({\bf L}_{C})_j=({\bf e}_j-\frac{1}{N}{\bf 1}_N)$.
\\
\noindent In summary, by reinterpreting the analysis description of the simple cycle graph Laplacian as a constrained synthesis formulation generated via the MPP and leveraging its closed-form expression as a discrete polynomial, we have arrived at a complete analytic characterization of the underlying functions defining the subspaces of the \textit{sparse synthesis model}. Here, arbitrary linear combinations of the columns of ${\bf L}_C^{\dagger}$ give rise to up to piecewise-quadratic polynomials, whose discontinuities depend on the particular choice of columns, or in other words, the locations of the sparse entries in ${\bf y}\in\mathbb{R}^N$ for synthesis representation ${\bf x}={\bf L}_C^{\dagger}{\bf y}$. When, furthermore, the amplitudes of the sparse entries $y_i$ are constrained to sum to $0$, which will be further enlarged upon in the next section, the representation ${\bf x}$ assumes a purely piecewise linear polynomial (or piecewise constant) form, which can also be described via the incidence matrix MPP ${\bf S}^{\dagger}_C$, subject to a translation by $N({\bf L}_C)$, within the \textit{cosparse analysis model} of ${\bf L}_C$. \\
\\
Notably, the developed relations establish a synthesis interpretation of (classical) vanishing moment constraints in form of analytical MPP expressions, revealing the underlying solution subspaces which govern the implicit analysis operation of ${\bf L}_C$, which is central as a high-pass operator to classical discrete wavelet analysis. Here, vanishing moments further represent the linear dependencies which determine the rank-deficiency of ${\bf L}_C$.

\subsection{General Circulant Graphs}
Equipped with the comprehensive theory for the simple cycle graph, we proceed to investigate signal models induced by difference operators on general undirected connected circulant  graphs, and begin by deriving a decomposition result for the graph Laplacian, which extends (Lemma $3.1$, \cite{splinesw}): 
\begin{lem}\label{lemdecomp}The graph Laplacian ${\bf L}$ of a circulant graph $G_S=(V,E)$, with generating set $S$ such that $s=1\in S$ and bandwidth $M<N/2$, can be decomposed as ${\bf L}={\bf P}_{G_S}{\bf L}_C$, where ${\bf P}_{G_S}$ is a circulant positive definite matrix of bandwidth $M-1$, which depends on $G_S$, with representer polynomial 
\[P_{G_S}(z)= \left(\sum_{i=1}^{M} i d_i \right)+\sum_{i=1}^{M-1} \left(\sum_{k=i+1}^{M}(k-i) d_k \right) (z^i+z^{-i})\]
and weights $d_i=A_{j,(i+j)_{N}}$, and ${\bf L}_C$ denotes the graph Laplacian of the unweighted simple cycle.
\end{lem}
\noindent \textit{Proof.} See \ref{appa}.\\
\\
\noindent Here, $s=1\in S$ is a sufficient condition which ensures that $G_S$ is connected. In the case of the simple cycle graph, we trivially have ${\bf P}_{G_S}={\bf I}_N$ up to a constant weight factor, and, in general, ${\bf P}_{G_S}$ encapsulates the additional connectivity information of ${G_S}$, creating the border effect in the annihilation of polynomial signals; from Lemma \ref{lemdecomp} it becomes evident that its coefficients directly mirror the structure of ${\bf y}$ in Sect.\ $4.1$ as a special case. \\
\\
Leveraging the decomposition of ${\bf L}$ on circulant graphs, we establish:
\begin{lem}\label{lemcircl}The MPP of ${\bf L}$ in Lemma \ref{lemdecomp} can be decomposed as ${\bf L}^{\dagger}={\bf P}_{G_S}^{-1}{\bf L}^{\dagger}_C$ and its rows/columns constitute piecewise quadratic polynomials which are `perturbed' by ${\bf P}_{G_S}^{-1}$ and orthogonal to $N({\bf L})=z{\bf 1}_N,\ z\in\mathbb{R}$.
\end{lem}
\noindent \textit{Proof}. See \ref{appa}.
\begin{rmk}
According to \cite{volkov}, the inverse of a cyclically banded positive definite matrix, such as ${\bf P}_{G_S}$ when $G_S$ is banded, contains entries that decay exponentially (in absolute value) away from the diagonal and corners of the matrix, inducing a `perturbation' on ${\bf L}^{\dagger}_C$. Here, the specific shape of ${\bf P}_{G_S}^{-1}$ is governed by the edge weights and bandwidth $M$, and assumes that of an approximately banded matrix when the weights are (close to being) uniform and $M$ sufficiently small. 
\end{rmk}
\noindent We further infer for the elementary piecewise-smooth functions $({\bf S}^{\dagger})_j$ on $G_S$:
\begin{lem}\label{lemlincirc}The columns of ${\bf S}^{\dagger}={\bf P}_{G_S}^{-1}{\bf L}_C^{\dagger}{\bf S}^T$ for a circulant graph (as above) are piecewise linear polynomials, subject to a perturbation by ${\bf P}_{G_S}^{-1}$.
\end{lem}
\noindent The preceding result establishes a fundamental relation between polynomial functions and circulant graphs. In particular, while Lemma \ref{circlins} states that the rows and columns (and their linear combinations) of ${\bf S}^{\dagger}_C$ are piecewise linear polynomials, Lemma \ref{lemlincirc} implies that the columns of ${\bf P}_{G_S}{\bf S}^{\dagger}$ and their linear combinations describe piecewise linear polynomials. Here, both ${\bf P}_{G_S}$ and ${\bf S}^{\dagger}$ encapsulate information, specific to the connectivity of the circulant graph at hand, which is cancelled out through their product. In other words, any piecewise linear polynomial function, which is in the subspace orthogonal to ${\bf 1}_N$, 
can be represented via ${\bf P}_{G_S}{\bf S}^{\dagger}$ on an arbitrary circulant graph.\\
\\
In light of Lemmata \ref{lemcircl} and \ref{lemlincirc}, we proceed to quantify the discrepancy between the analysis and synthesis models on circulant graphs:
\begin{thm}\label{thmmain} On circulant graphs (as per Lemma \ref{lemdecomp}), the signal subspaces associated with the cosparse analysis model of ${\bf L}$, given by
\[N({\bf \Psi}_{\Lambda}{\bf L})=z{\bf 1}_N+{\bf P}_{G_S}^{-1}{\bf L}_C^{\dagger}{\bf \Psi}^T_{\Lambda^{\complement}}{\bf W}{\bf c},\] for $z\in\mathbb{R}$ and ${\bf c}\in\mathbb{R}^{|\Lambda^{\complement}|-1}$, with alternate representation ${\bf L}_C^{\dagger}{\bf \Psi}^T_{\Lambda^{\complement}}{\bf W}{\bf c}={\bf S}_C^{\dagger}{\bf t}$ for suitable ${\bf t}\in\mathbb{R}^N$,
consist of up to piecewise linear polynomials, subject to a translation by $N({\bf L})$, while those of the sparse synthesis model, generated by ${\bf L}^{\dagger}$, describe up to piecewise quadratic polynomials, respectively subject to graph-dependent perturbations in form of factor ${\bf P}^{-1}_{G_S}$. \\
\end{thm}
\noindent \textit{Proof.} Follows from above discussion.
\begin{rmk} Owing to unique structural properties of ${\bf L}$ and ${\bf S}$, including structured sparsity and polynomial functions, Thm.\ \ref{thmmain} exemplifies their rich representation range within both models: these properties facilitate alternative representations of signals through both ${\bf L}^{\dagger}$ and ${\bf S}^{\dagger}$, subject to constraints, while simultaneously establishing closed-form expressions for crucial graph operators.\end{rmk}
The synthesis model defines graph signals ${\bf x}={\bf L}^{\dagger}{\bf c}$ as linear combinations of (perturbed) piecewise quadratic polynomials, encompassing piecewise linear and piecewise constant signals for suitable choices of ${\bf c}$, which are orthogonal to ${\bf 1}_N$ and whose discontinuities depend on the location of the sparse entries (knots) of ${\bf c}$ at $\Lambda^{\complement}$. This establishes a comprehensive model, containing different orders of smoothness and hop-localization of the subspaces\footnote{In particular, ${\bf S}^{\dagger}_j$ is piecewise smooth with respect to ${\bf L}$, which induces sparsity localized with respect to a $1$-hop neighborhood of the graph, and for sufficiently sparse graphs, ${\bf L}^{\dagger}_j$ is piecewise smooth with respect to ${\bf L}^2$, which induces $2$-hop localized sparsity.}, while the analysis model constitutes a structured instance thereof, comprising signals that are sparse with respect to ${\bf L}$, up to a shift by ${\bf 1}_N$.\\
The constraint matrix ${\bf W}$ in the analysis subspace ${\bf P}_{G_S}^{-1}{\bf L}_C^{\dagger}({\bf \Psi}^T_{\Lambda^{\complement}}){\bf W}$, which takes zero-sum differences between atoms of ${\bf L}_C^{\dagger}$, thereby giving rise to (piecewise) linear polynomials subject to a perturbation by ${\bf P}_{G_S}^{-1}$, marks a significant discrepancy in the structure of basis functions between the analysis and synthesis models. By applying the double constraint ${\bf c}=(({\bf e}_i-{\bf e}_j)\ast ({\bf e}_k-{\bf e}_l))_{modN}$, the solution ${\bf L}^{\dagger}{\bf c}$ further reduces to perturbed piecewise constant form, while in the special case ${\bf L}^{\dagger}{\bf L}$, we obtain piecewise constant solutions with one discontinuity. 
\\
\\
{\bf Higher-Order Operators.} Theorem \ref{thmmain} can be trivially generalized to higher-order operators ${\bf L}^k$, yielding the partial subspace ${\bf L}^{\dagger k}{\bf \Psi}^T_{\Lambda^{\complement}}{\bf W}{\bf c}={\bf L}^{\dagger k-1}{\bf S}^{\dagger}{\bf t}$, for suitable ${\bf t}\in\mathbb{R}^{|E|}$ and ${\bf c}\in\mathbb{R}^{|\Lambda^{\complement}|-1}$. The known analysis-property that ${\bf L}^k$ annihilates up to degree $2k-1$ polynomials, subject to a graph-dependent border effect, as per (Lemma $3.1$, \cite{splinesw}), implicitly defines the solution space of $N({\bf \Psi}_{\Lambda}{\bf L}^k)$, which, combined with the constrained synthesis-based expression of Lemma \ref{summ} $(i)$, reveals that ${\bf L}^{\dagger k-1}{\bf S}^{\dagger}$ necessarily gives rise to degree $2k-1$ (piecewise) polynomials, which are orthogonal to ${\bf 1}_N$, and (possibly) subject to graph-dependent perturbations through ${\bf P}_{G_S}^{-1}$.\\
In view of Lemma \ref{summ} $(iii)$ for $k=2$, the MPP relation ${\bf L}^2{\bf L}^{\dagger}={\bf L}$ further reveals that for sufficiently banded $G_S$, ${\bf L}^{2}$, whose rows have $4$ vanishing moments, annihilates the `perturbed' quadratic polynomial $({\bf L}^{\dagger})_j$ with structured sparse output ${\bf L}_j$. In other words, the non-zeros of the columns ${\bf L}_j$ are the locations of the knots of the Green's functions ${\bf L}^{\dagger 2}{\bf L}={\bf L}^{\dagger}$. Here, ${\bf L}^{\dagger}$ is sparse with respect to a $2$-hop localized neighborhood on the graph.\\ 
While closed-form expressions for the higher-order MPPs need to be derived on a case by case basis, it has been shown that the columns of ${\bf L}^{\dagger 2}_C$ describe 4th order polynomials \cite{plonka}, from which it then naturally follows that the columns of ${\bf L}^{\dagger 2}_C{\bf S}^T_C={\bf L}^{\dagger}_C{\bf S}^{\dagger}_C$ are cubic polynomials which are sparse with respect to the operator ${\bf L}_C^2$; accordingly, ${\bf L}_C^{\dagger 2}{\bf L}_C={\bf L}_C^{\dagger}$ describe quadratic polynomials, following a twofold degree-reduction. In addition ${\bf L}^{\dagger 2}_C$ is sparse with respect to ${\bf S}_C{\bf L}_C^2$, which possesses an extra vanishing moment through ${\bf S}_C$. Statements carry over, subject to perturbations, for sufficiently banded circulant graphs. Further building on this discussion, we state a generalization of the annihilation property to operators of the form ${\bf S}{\bf L}^k$:
\begin{lem}\label{gan} For any banded circulant graph, ${\bf S }{\bf L}^k$, $k\in\mathbb{N}$ annihilates polynomial signals of up to order $2k$, subject to graph-dependent border effects.\end{lem}
\noindent It becomes evident that ${\bf S}{\bf L}^k$ functions as the intermediate graph difference operator between ${\bf L}^{k}$ and ${\bf L}^{k+1}$, whose number of vanishing moments\footnote{In the ideal scenario of a simple cycle, ${\bf S}_C{\bf L}_C^k$ has an odd number of vanishing moments $2k+1$, compared the even number of vanishing moments $2k$ of ${\bf L}_C^k$.} increases in degree steps of $2$. Further, the above result, in conjunction with Lemma \ref{summ} $(ii)$, provides a concrete identification of the constrained synthesis (or analysis) signal space for circulant graphs, whereby a signal of the form ${\bf x}={\bf L}^{\dagger k}{\bf w}$ with suitable ${\bf w}\in\mathbb{R}^N$ that is sparse with respect to ${\bf S}{\bf L}^k$, in the range of ${\bf S}$, must belong to the space of degree $2k$ (piecewise) polynomials, which are orthogonal to ${\bf 1}_N$, possibly subject to perturbations through ${\bf P}_{G_S}^{-1}$. Thus, in the specific instance when a given higher-order polynomial can be expressed as ${\bf L}^{\dagger k}{\bf w}$, its sparse pattern with respect to ${\bf S}{\bf L}^k$ is given by ${\bf S}{\bf w}$.
\\
\\
{\bf A note on discontinuities.} The (perturbed) piecewise linear functions of $N({\bf \Psi}_{\Lambda}{\bf L})$ exhibit discontinuities whose locations depend on the choice of subset $\Lambda$ and the graph connectivity/bandwidth. The specific structure of ${\bf P}_{{G_S}}^{-1}{\bf L}_C^{\dagger}{\bf t}$, with structured sparse vector ${\bf t}={\bf \Psi}^T_{\Lambda^{\complement}}{\bf W}{\bf c}$, whose non-zeros are at locations ${\Lambda}^{\complement}$ and sum to $0$, in fact, reveals piecewise linear polynomials, whose discontinuities are perturbed (and hence amplified in spread) by the graph-dependent, (if applicable) approximately banded matrix ${\bf P}_{G_S}^{-1}$. Here, the choice and density of pattern ${\bf t}$ determines the number, location and/or spread (i.e. amplification via ${\bf P}_{{G_S}}^{-1}$) of the discontinuities. The analysis of the simple cycle previously revealed that the number of discontinuities, and hence, number of distinct polynomial pieces, is minimized to one, when considering differences between consecutive atoms of ${\bf L}_C^{\dagger}$. Accordingly, the spread of the discontinuity is amplified by ${\bf P}_{{G_S}}^{-1}$ for general circulant graphs, its size increasing with the graph connectivity. One may further enforce an anti-symmetric function shape, by e.g. taking differences ${\bf L}^{\dagger}_j-{\bf L}^{\dagger}_i$ for $i,j\in V$, confined to pairwise flipped versions $j=N-i-1$, which facilitates anti-symmetry with respect to a vertical line, as observed on the simple cycle. Furthermore, one may specifically tailor the relative signal smoothness through ${\bf t}$, as the following remark exemplifies:
\begin{rmk}\label{smoth}
The representation ${\bf x}={\bf P}_{G_S}^{-1}{\bf L}_C^{\dagger}{\bf \Psi}_{\Lambda^{\complement}}^T{\bf W}{\bf c}$, assumes different orders of smoothness according to the choice of basis ${\bf W}$. Consider basis ${\bf \Psi}_{\Lambda^{\complement}}^T\tilde{{\bf W}}$, which absorbs column $({\bf P}_{G_S})_j$ while satisfying the constraints of ${\bf W}$ via circular convolution, such that $\Lambda^{\complement}$ covers the support of ${\bf p}:=(({\bf P}_{G_S})_j\ast({\bf e}_k-{\bf e}_l))_{mod N}$, for suitable $j,k,l\in V$. Consequently, ${\bf x}={\bf L}^{\dagger}{\bf p}$ becomes sparse with respect to ${\bf L}_C$, in addition to ${\bf L}$, with sparse output respectively given by $({\bf e}_j\ast({\bf e}_k-{\bf e}_l))_{mod N}$ and ${\bf p}$, and can be considered piecewise-smooth in the `classical' sense as well as with respect to $G_S$, removing the perturbation by ${\bf P}^{-1}_{G_S}$.
\end{rmk}
In view of the above, a signal ${\bf L}^{\dagger}{\bf t}$ on a general circulant graph can achieve the same smoothness and minimum number of discontinuities as a signal on the simple cycle for ${\bf t}=(({\bf P}_{G_S})_j\ast({\bf e}_k-{\bf e}_l))_{mod N}$ with suitable $j,k,l\in V$. One needs to distinguish here between the discontinuities relative to the graph connectivity, in form of a more dense ${\bf t}$, as well with respect to the simple cycle, which facilitates a notion of classical smoothness that is detached from the graph domain.
\\
\\
{\bf Special Case (Complete Graph)}. In the special case of an unweighted complete (circulant) graph, we have ${\bf S}^{\dagger}=\frac{1}{{N}}{\bf S}^{T}$ and ${\bf L}^{\dagger}=\frac{1}{{N^2}}{\bf L}$, which result from ${\bf L}^{\dagger}{\bf L}=\frac{1}{N}{\bf L}$ and the MPP relations ${\bf L}^{\dagger}{\bf S}^T={\bf S}^{\dagger}$ and ${\bf L}{\bf S}^{\dagger}={\bf S}^T$. This implies that $({\bf S}^{\dagger})_j$ is trivially smooth with respect to ${\bf L}$, as it simultaneously represents its sparsity pattern. Nevertheless, since the graph is dense, $({\bf L}^{\dagger})_j$ is not sparse with respect to ${\bf L}^2$. As a result of the maximum graph connectivity, the piecewise-smooth functions $({\bf S}^{\dagger})_j$ thus take the form of two opposite-sign impulses, as per $({\bf S}^{T})_j$, while $({\bf L}^{\dagger})_j$ is piecewise constant, just as $({\bf L})_j$.\\
\\
We conclude by illustrating the specific model subspace discrepancies through examples. In Fig.\ \ref{fig:newfig}, $(a)$-$(b)$, we consider two distinct quadratic polynomial atoms $({\bf L}^{\dagger})_j$ with a discontinuity respectively at position $t=i,j$ along with their difference, which gives rise to the linear signal with two discontinuities at positions $i,j$ and, hence, $2$-sparse analysis representation ${\bf e}_i-{\bf e}_j$ with respect to ${\bf L}$. Here, subfigures $(a)$ and $(b)$ respectively illustrate the functions resulting from the domain of two different unweighted circulant graphs of dimension $N$. The perturbation effect, resulting from approximately banded matrix ${\bf P}_{G_S}^{-1}$ of decaying support, is clearly visible around the discontinuities for the non-trivial circulant case in $(b)$.
Further, in Fig. \ref{fig:newfig}$(c)$, we consider the MPP $({\bf S}^{\dagger})_0={\bf L}^{\dagger}({\bf S}^{T})_0$ for unweighted circulant graphs of increasing connectivity, subject to rescaling for visibility, where $({\bf S}^{T})_0$ uniformly corresponds to $({\bf e}_1-{\bf e}_0)$. For the simple cycle, $({\bf S}^{\dagger})_0$ is a $1$-piece linear polynomial with a single discontinuity at its endpoints, which is subject to perturbations with increasing graph connectivity. In the limiting case of the complete graph, we observe two consecutive, opposite sign impulses at the discontinuity location and zeros otherwise.
\begin{figure}
 \centering
  \begin{subfigure}[htbp]{0.49\textwidth}
{\includegraphics[width=2.5in]{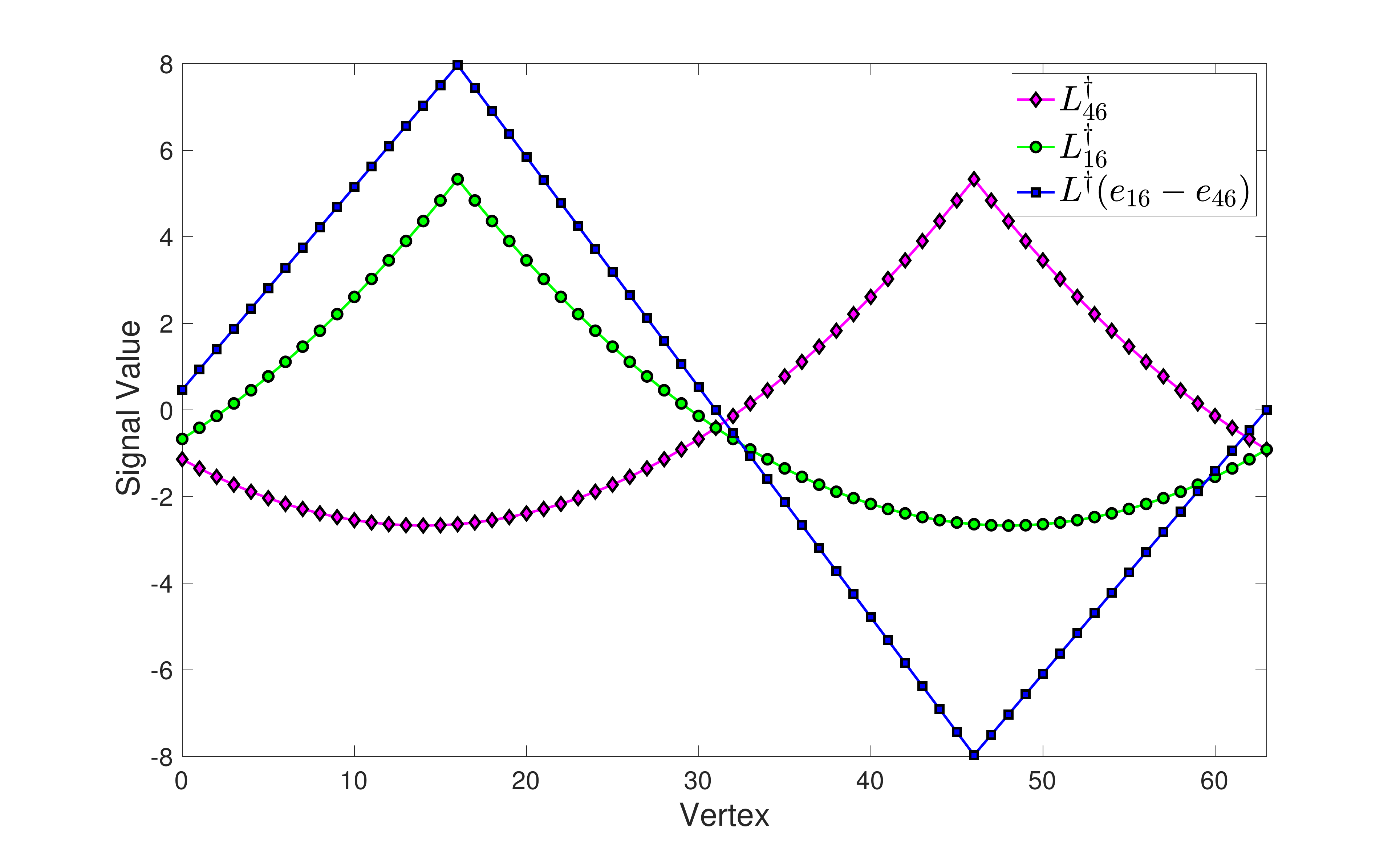}} 
\caption{\small{Functions on $G_S$ with $S=\{1\}$}}
\end{subfigure}
\begin{subfigure}[htbp]{0.49\textwidth}
 { \includegraphics[width=2.5in]{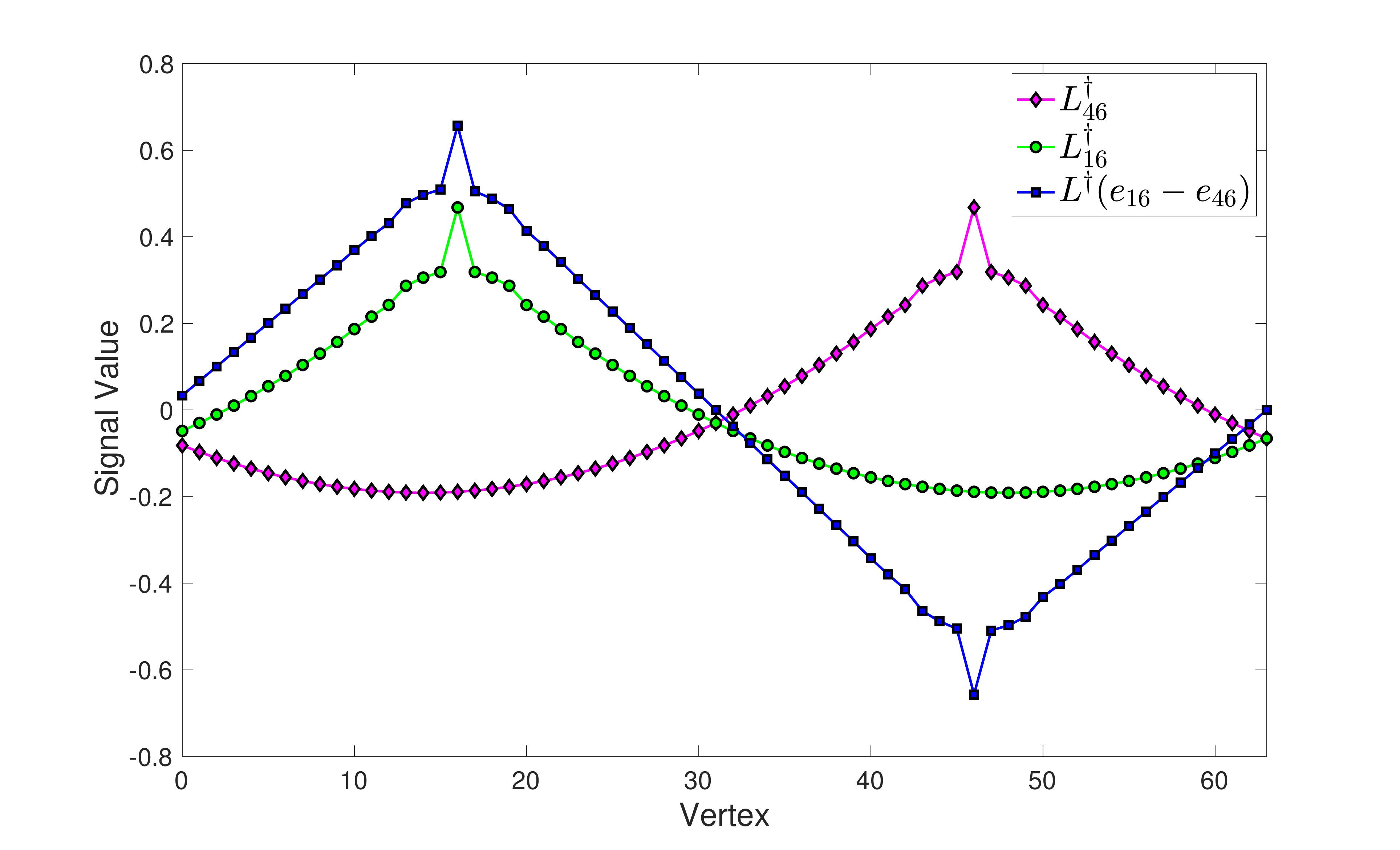}}
 \caption{\small{Functions on $G_S$ with $S=\{1,2,3\}$}}
 \end{subfigure}
 \begin{subfigure}[htbp]{0.49\textwidth}
 { \includegraphics[width=2.5in]{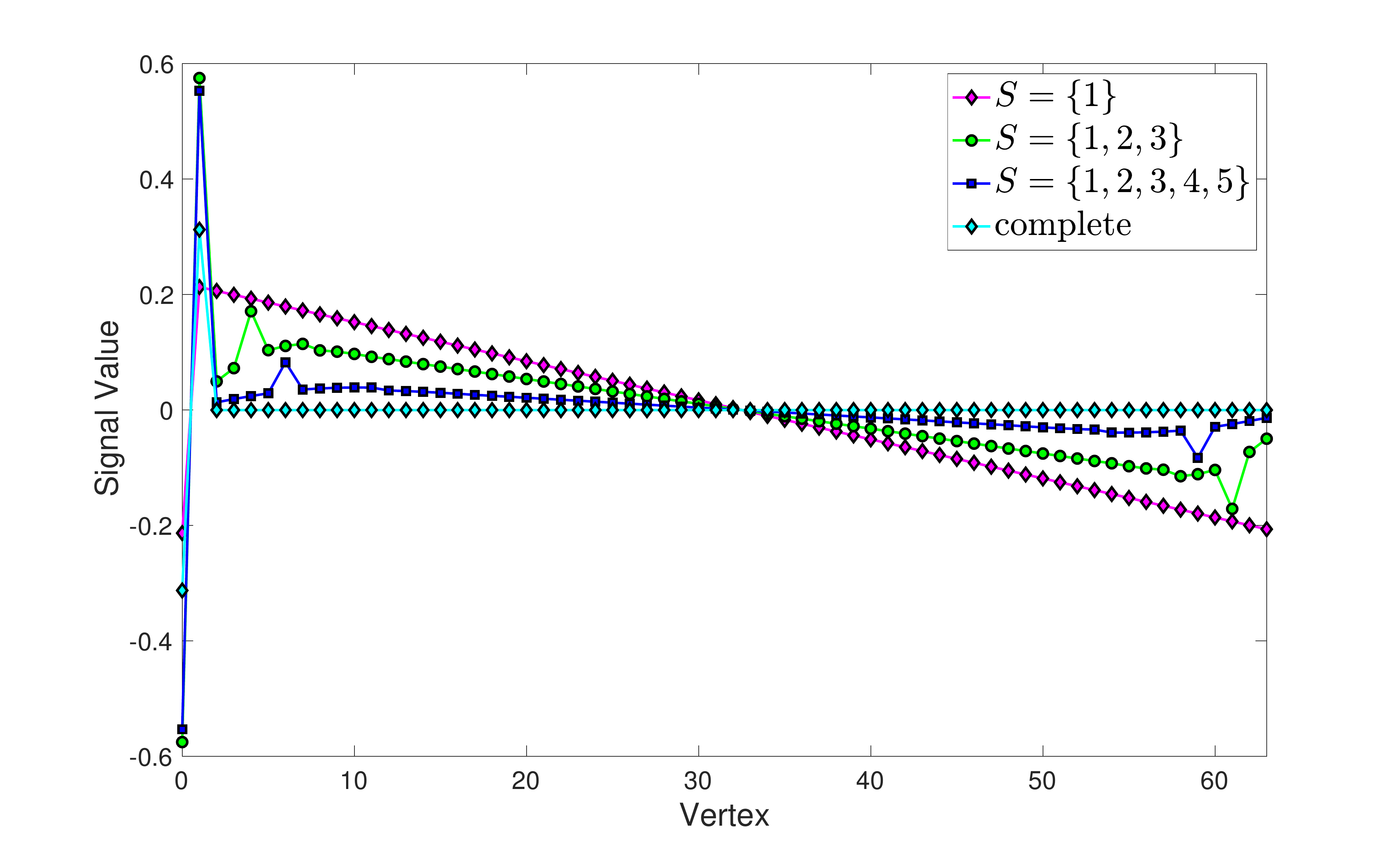}}
 \caption{\small{$({\bf S}^{\dagger})_0$ on different $G_S$}}
 \end{subfigure}
	\caption{Comparison of signal models on circulant graphs}
\label{fig:newfig}\end{figure}

\subsection{The generalized graph Laplacian}
In an effort to explore signal models for broader classes of signals on graphs, we study the generalized, parametric circulant graph Laplacian operator ${\bf L}_{\alpha}=d_{\alpha}{\bf I}_N-{\bf A}$, with node degree $d_{\alpha}=\sum_{j=1}^M 2 d_j \cos(\alpha j)$ parameterized by $\alpha\in\mathbb{C}$. Designed, in previous work (see Lemma\ $3.2$, \cite{splinesw}), to possess the property to annihilate classes of complex exponential signals of the form ${\bf x}=e^{\pm i\alpha{\bf t}}$ with exponent $\alpha\in\mathbb{C}$, and by extension, for its powers ${\bf L}_{\alpha}^k$ to annihilate complex exponential polynomials of degree $k-1$ of the form ${\bf x}=e^{\pm i\alpha {\bf t}}{\bf p}$, this operator further possesses the following emergent characteristic (see \ref{app2}):
\begin{propp}\label{prop4}The generalized graph Laplacian ${\bf L}_{\alpha}$ on a circulant graph is invertible for $\alpha\neq2\pi k/N, \ k\in\mathbb{N}$, and rank-deficient otherwise. \end{propp}
\noindent Consequently, we achieve exact annihilation of complex exponentials with $N({\bf L}_{\alpha})=z_1e^{i\alpha {\bf t}}+z_2e^{-i\alpha {\bf t}},\ z_1,z_2\in\mathbb{C}$ for $\alpha=2\pi k/N, \ k\in\mathbb{N}$, while, otherwise, this is subject to a border effect dependent on the graph bandwidth $M$. Further, ${\bf L}_{\alpha}$ reduces to the classical graph Laplacian for $\alpha=0$. Accordingly, it serves as a relevant case to explore transitional properties and relations between UoS models for square invertible and rank-deficient operators respectively. \\
We proceed to derive closed-form expressions for both cases on the simple cycle graph with ${\bf L}_{C,\alpha}=2\cos(\alpha){\bf I}_N-{\bf A}_C$ and representer polynomial $l_{C, \alpha}=2\cos(\alpha)-z-z^{-1}=-z^2(1-e^{i\alpha}z)(1-e^{-i\alpha}z)$, which has two exponential vanishing moments:
\begin{lem}\label{inva}
The inverse ${\bf L}_{C,\alpha}^{-1}$ of ${\bf L}_{C, \alpha}$ on the simple cycle, for $\alpha\neq 2\pi k/N,\enskip k\in\mathbb{N}$, has entries 
\[{L}_{C, \alpha}^{-1}(m,n)=\frac{1}{(-e^{-i\alpha}+e^{i\alpha})(-1+e^{i\alpha N})} e^{i\alpha |n-m|}
+\frac{1}{(e^{-i\alpha}-e^{i\alpha})(-1+e^{-i\alpha N})} e^{-i\alpha |n-m|},\]\[0\leq m,n\leq N-1.\]
\end{lem}
\begin{lem}\label{mppa}
The pseudoinverse ${\bf L}_{C, \alpha}^{\dagger}$ for $\alpha=2\pi k/N,\enskip k\in\mathbb{N}$ and $\alpha\neq 0, k \pi$, of ${\bf L}_{C, \alpha}$ on the simple cycle has entries
\[{L}_{C, \alpha}^{\dagger}(m,n)=\frac{e^{i\alpha}}{2N}\left(\frac{2|n-m|(-1+e^{2 i\alpha})+(N-1)-e^{2 i\alpha}(N+1)}{(-1+e^{2 i \alpha})^2}\right)e^{i\alpha |n-m|}\]
\[+\frac{e^{-i\alpha}}{2N}\left(\frac{2|n-m|(-1+e^{-2 i\alpha})+(N-1)-e^{-2 i\alpha}(N+1)}{(-1+e^{-2 i \alpha})^2}\right)e^{-i\alpha |n-m|},\enskip 0\leq m,n\leq N-1.\]
For $\alpha=\pi$ ($k=N/2$ at even $N$) we have
\[{L}_{C, \pi}^{\dagger}(m,n)=(-1)^{|n-m|+1}\left(\frac{(n-m)^2}{2N}-\frac{1}{2}|n-m|+\frac{N^2-1}{12 N}\right)\]\[=(-1)^{|n-m|+1}{L}_{C}^{\dagger}(m,n),\enskip 0\leq m,n\leq N-1.\]
\end{lem}
\noindent It becomes evident that, according to $\alpha$, the rows and columns of the parametric inverse ${\bf L}_{C, \alpha}^{-1}$ define complex exponentials, while those of the rank-deficient parametric pseudoinverse ${\bf L}_{C, \alpha}^{\dagger}$ define linear complex exponential polynomials, i.e. their function order exceeds the former by two degrees, as a result of the operator's rank-deficiency. In addition, the atoms $({\bf L}_{C, \alpha}^{-1})_j$ and $({\bf L}_{C, \alpha}^{\dagger})_j$ both exhibit a discontinuity between positions $j-1$ and $j$.
\begin{rmk}\label{contr1}
The derived formulae both cease to apply at $\alpha=0$, the classical graph Laplacian case, since they have a pole at that value; nevertheless, the expression for ${\bf L}_C^{\dagger}$ is known from previous work (see Property \ref{propp1}) to be a piecewise quadratic polynomial. To further reinforce the link between these derivations, we note that the recurrence relation resulting from the graph Laplacian MPP system ${\bf L}_C{\bf L}^{\dagger}_C={\bf I}_N-\frac{1}{N}{\bf J}_N$ with unknown ${L}^{\dagger}_C(m,n)$ and ${L}^{\dagger}_C(r):={L}^{\dagger}_C(m,n)$ for $r=|n-m|$, is $-{L}^{\dagger}_C(r+1)+2 {L}^{\dagger}_C(r)-{L}^{\dagger}_C(r-1)=-1/N$, with boundary condition $2 {L}^{\dagger}_C(0)-2 {L}^{\dagger}_C(1)=1-1/N$ and constraint $\sum_{r=0}^{N-1} {L}^{\dagger}_C(r)=0$. The homogeneous solution to the system is $({L}^{\dagger}_C)_H(r)=\frac{N^2-1}{12N} -\frac{1}{2} r$, which is notably linear, while the particular solution becomes $({L}^{\dagger}_C)_P(r)=\frac{r^2}{2N}$ as a result of repeated roots at $1$. In contrast, the recurrence relations for ${\bf L}_{C, \alpha}^{-1}$ and ${\bf L}_{C, \alpha}^{\dagger}$ give rise to single and double roots at $(e^{\pm i\alpha})^r$ respectively, which leads to their respective solution sets of complex exponentials $e^{\pm i\alpha {\bf t}}$ and linear complex exponential polynomials ${\bf p}e^{\pm i\alpha {\bf t}}+e^{\pm i\alpha {\bf t}}$. Here, distinct circular boundary conditions and constraints lead to multiple repeated roots, and hence, an extended solution space through the particular solution, for the MPP ${\bf L}_{C, \alpha}^{\dagger}$. Eventually, in the special case of $\alpha=0$, the system solution for ${\bf L}_{C, \alpha=0}^{\dagger}={\bf L}_C^{\dagger}$ gains triple roots with a quadratic polynomial solution set, i.e. the solution set for ${\bf L}_{C, \alpha}^{\dagger}$ gains another repeated root at $1^r$ for $\alpha=0$. For $\alpha=\pi$, we similarly observe the occurrence of triple roots in the recurrence relations for ${\bf L}_{C, \pi}^{\dagger}$, with the distinction  that they are given by $(-1)^{r}$; in fact, as a result of similar boundary conditions, the formulae for ${\bf L}_{C, \pi}^{\dagger}$ and ${\bf L}_{C}^{\dagger}$ coincide up to a factor $(-1)^{r+1}$. We refer to \ref{app2} for the detailed proofs.
\end{rmk}
In order to extend these properties to general circulant graphs, we state the following decomposition, as a generalization of Lemma \ref{lemdecomp}: 
\begin{lem}\label{newdecomp}
The generalized graph Laplacian ${\bf L}_{\alpha}$ on a connected circulant graph, with generating set $S$ and bandwidth $M<N/2$, can be decomposed as ${\bf L}_{\alpha}={\bf L}_{C, \alpha}{\bf P}_{\alpha}$, where ${\bf L}_{C, \alpha}$ is the generalized simple cycle graph Laplacian and ${\bf P}_{\alpha}$ is a circulant matrix of bandwidth $M-1$ with representer polynomial
\[{P}_{\alpha}(z)=\sum_{j=1}^Md_j \left(r_{j-1}+\sum_{t=1}^{j-1} r_{j-1-t} (z^t+z^{-t}))\right)\]
where 
\[r_t=\sum_{k=0}^{(t+1)/2-1}2\cos(\alpha(2k-t))=2\cos(\alpha t)+2\cos(\alpha(2-t))+...+2\cos(\alpha 3)+2\cos(\alpha)\]
when t is odd, and otherwise,
\[r_t=1+\sum_{k=0}^{t/2-1}2\cos(\alpha(2k-t))=2\cos(\alpha t)+...+2\cos(\alpha 2)+1\]
which depends on the graph at hand. 
\end{lem}
\noindent Unlike Lemma \ref{lemdecomp} for $\alpha=0$, the positive definiteness (or invertibility) of the banded matrix ${\bf P}_{\alpha}$ depends on parameter $\alpha$ and the graph connectivity; due to the complex connectivity of graphs and plethora of special cases, we refrain from deriving exact conditions on when this is satisfied here and assume this to be the case a priori. We note a relation to the occurrence of repeated eigenvalues in ${\bf L}_{\alpha}$. In general, we observe that with an increase in connectivity (i.e. bandwith $M$) and in the value of $\alpha$, ${\bf P}_{\alpha}^{-1}$ gains larger magnitude entries, and eventually loses its positive definite (non-singular) property. We proceed to state:
\begin{cor}\label{mppinva}
$(i)$ The inverse of ${\bf L}_{\alpha}$ for $\alpha\neq 2\pi k/N,\enskip k\in\mathbb{N}$, can be decomposed as ${\bf L}_{\alpha}^{-1}={\bf L}_{C, \alpha}^{-1}{\bf P}_{\alpha}^{-1}$, and its rows and columns constitute complex exponentials, as per Lemma \ref{inva}, which are perturbed by ${\bf P}_{\alpha}^{-1}$.\\
$(ii)$ For suitable $\alpha$, satisfying $\alpha=2\pi k/N$ and $\alpha\neq 0, k \pi,\enskip k\in\mathbb{N}$, and graph connectivity $G_S$, such that ${\bf P}_{\alpha}$ is positive definite, the MPP of ${\bf L}_{\alpha}$ can be decomposed as ${\bf L}_{\alpha}^{\dagger}={\bf L}_{C, \alpha}^{\dagger}{\bf P}_{\alpha}^{-1}$. The rows and columns of ${\bf L}_{\alpha}^{\dagger}$ constitute complex exponential linear polynomials, as per Lemma \ref{mppa}, which are perturbed by ${\bf P}_{\alpha}^{-1}$.
\end{cor}
\noindent Consequently, we derive the cosparse analysis model for ${\bf L}_{\alpha}$, generated by $N({\bf \Psi}_{\Lambda}{\bf L}_{\alpha})$:
\begin{prop}\label{newnullsp1}
The nullspace $N({\bf \Psi}_{\Lambda}{\bf L}_{\alpha})$ of the generalized graph Laplacian operator ${\bf L}_{\alpha}$ on a circulant graph is given by\\
$(i)$ $N({\bf \Psi}_{\Lambda}{\bf L}_{\alpha})={\bf L}_{\alpha}^{-1}{\bf \Psi}^T_{\Lambda^{\complement}}\tilde{{\bf c}}$ for $\alpha\neq 2\pi k/N,\ k\in\mathbb{N}$, with arbitrary $\tilde{{\bf c}}\in\mathbb{R}^{|\Lambda^{\complement}|}$, and\\
$(ii)$ $N({\bf \Psi}_{\Lambda}{\bf L}_{\alpha})={\bf L}_{\alpha}^{\dagger}{\bf \Psi}_{\Lambda^{\complement}}^T{\bf W}_{\alpha}{\bf c}+z_1e^{i\alpha{\bf t}}+z_2e^{-i\alpha{\bf t}}$ for $\alpha=2\pi k/N,\ k\in\mathbb{N}$, and $\alpha\neq 0, k \pi$, with ${\bf W}_{\alpha}:=N({\bf \Psi}_{\Lambda}{\bf E}_{\alpha}{\bf \Psi}^T_{\Lambda^{\complement}})\in\mathbb{C}^{|\Lambda^{\complement}|\times |\Lambda^{\complement}|-2}$, ${\bf E}_{\alpha}=\sum_{j=1}^2{\bf u}_j {\bf u}^H_j$ for ${\bf u}_1=e^{ i\alpha {\bf t}}$, ${\bf u}_2=e^{-i\alpha {\bf t}}$, and arbitrary ${\bf c}\in\mathbb{C}^{|\Lambda^{\complement}|-2}$ and $z_1,z_2\in\mathbb{C}$. \\If $|\Lambda^{\complement}|< 3$, $N({\bf \Psi}_{\Lambda}{\bf L}_{\alpha})=N({\bf L}_{\alpha})$, while in the special case where $\Lambda^{\complement}=\{m,(m+\frac{N}{2})\},\ m\in V$, $N({\bf \Psi}_{\Lambda}{\bf L}_{\alpha})$ has rank $|\Lambda^{\complement}|+1$.  
\end{prop}
\noindent \textit{Proof.} See Appendix.\\
\\
\noindent In the rank-deficient case of Prop.\ \ref{newnullsp1}, the analysis constraint ${\bf \Psi}_{\Lambda^{\complement}}^T{\bf W}_{\alpha}$ encapsulates orthogonality to complex exponentials such that $({\bf E}_{\alpha})_j\perp{\bf \Psi}_{\Lambda^{\complement}}^T{\bf W}_{\alpha}{\bf c},\enskip j\in\Lambda$, is satisfied and the complex exponential part of ${\bf L}_{\alpha}^{\dagger}$ is annihilated, signifying a reduction in its order. 
Specifically, we can express ${\bf \Psi}_{\Lambda^{\complement}}^T{\bf W}_{\alpha}{\bf c}$ as $({\bf w}_{\alpha}\ast{\bf w}_{-\alpha})_{mod N}$, the circular convolution of $2$-sparse constraint vectors ${\bf w}_{\pm \alpha}\in\mathbb{C}^N$, with ${\bf w}_{\pm \alpha}=(e^{\pm i \alpha (n-m)}  {\bf e}_m-{\bf e}_n),\ n>m,$ of support in $\Lambda^{\complement}$, which respectively capture orthogonality to $e^{i \alpha {\bf t}}$ and $e^{-i \alpha {\bf t}}$. When ${\Lambda^{\complement}}$ is entirely arbitrary, the specific shape of the subspace ${\bf L}_{\alpha}^{\dagger}{\bf \Psi}_{\Lambda^{\complement}}^T{\bf W}_{\alpha}$ is more opaque 
due to the variation in ${\bf \Psi}_{\Lambda^{\complement}}^T{\bf W}_{\alpha}$; if additional conditions are imposed however, we can state the following: 
\begin{rmk}\label{constr}
In the case when $\Lambda^{\complement}\subset V$ has at least three consecutive elements, the constraint ${\bf \Psi}_{\Lambda^{\complement}}^T{\bf W}_{\alpha}$ can be conveniently formulated through a basis whose vectors possess consecutive non-zero entries $\lbrack -1\enskip 2\cos(\alpha)\enskip -1 \rbrack$ and zeros otherwise. In particular, these represent the (wavelet) basis elements of shortest length $3$, which are orthogonal to ${e^{\pm i\alpha {\bf t}}}$, and further represent the rows/columns of ${\bf L}_{C, \alpha}$. Due to ${\bf L}_{\alpha}^{\dagger}({\bf L}_{C, \alpha})_j={\bf P}_{\alpha}^{-1}{\bf L}_{C,\alpha}^{\dagger}({\bf L}_{C,\alpha})_j={\bf P}_{\alpha}^{-1}({\bf I}_N-\frac{1}{N}{\bf E}_{\alpha})_j$ such that ${\bf \Psi}_{\Lambda^{\complement}}^T{\bf W}_{\alpha}{\bf c}=({\bf L}_{C,\alpha})_j$ for suitable $j$ and ${\bf c}$, the analysis subspace consequently assumes the form of perturbed complex exponentials with discontinuities at positions $j$.
When ${\Lambda^{\complement}}$ further contains the graph support such that ${\bf L}_{\alpha}^{\dagger}{\bf \Psi}_{\Lambda^{\complement}}^T{\bf W}_{\alpha}{\bf c}={\bf L}_{\alpha}^{\dagger}({\bf L}_{\alpha})_j=({\bf I}_N-\frac{1}{N}{\bf E}_{\alpha})_j$ for some ${\bf W}_{\alpha}{\bf c}$ and index $j$, the solution subspace generates complex exponentials with knots at positions $j$. Nevertheless, unlike in the classical graph Laplacian case where one may consider the constrained ${\bf L}^{\dagger}{\bf L}_j$ and ${\bf L}^{\dagger}{\bf S}_j^T$, this representation does not facilitate any further reduction in degree since the generalized ${\bf L}_{\alpha}$ does not feature any additional vanishing moments compared to ${\bf L}_{C, \alpha}$.\footnote{Since ${\bf L}$ is PSD, the possible analysis constraints, given by ${\bf S}^T_j$, which is necessary and sufficient by the F.A., and the higher-order ${\bf L}_j$, are incremental; since ${\bf L}_{\alpha}$ is not PSD, similar decompositions do not exist and the exponential vanishing moments of $({\bf L}_{\alpha})_j$ are both absorbed into the analysis constraint, while for ${\bf L}$ only one vanishing moment is needed.} 
\end{rmk}
Following Rmk.\ \ref{constr}, for certain ${\Lambda^{\complement}}$, the introduced constraint reduces the subspace ${\bf L}_{\alpha}^{\dagger}={\bf P}_{\alpha}^{-1}{\bf L}_{\alpha,C}^{\dagger}$ of perturbed linear complex exponential polynomials by two orders, i.e. to perturbed complex exponentials, which confines the solution subspace of the analysis model $N({\bf \Psi}_{\Lambda}{\bf L}_{\alpha})$ to perturbed (piecewise) complex exponential functions. As a result, their order effectively coincides with the order of the functions in $N({\bf \Psi}_{\Lambda}{\bf L}_{\alpha})={\bf L}_{\alpha}^{-1}{\bf \Psi}^T_{\Lambda^{\complement}}{\bf c}$ in the invertible case. 
\\
\\
Accordingly, for $\Lambda^{\complement}$ as above, apart from a difference in admissible parameter values, the analysis subspaces of the rank-deficient and full-rank operators are comparable in order, while their corresponding synthesis subspaces differ in order, i.e. the latter exceeds the former by two exponential degrees. 
In direct comparison, the basis functions ${\bf L}_{\alpha}^{-1}$ of the full-rank operator ${\bf L}_{\alpha}$ are (perturbed) complex exponentials facilitating equivalence between the analysis and synthesis domains, while the basis functions ${\bf L}_{\alpha}^{\dagger}$ of the rank-deficient ${\bf L}_{\alpha}$ are (perturbed) complex exponential polynomials in the synthesis case, which exceed the former by two degrees and reduce to complex exponentials via the exponential vanishing moment constraints in the analysis case; hence, within their respective domain for $\alpha$, ${\bf L}_{\alpha}^{-1}$ and ${\bf L}_{\alpha}^{\dagger}$ are comparable in the analysis case as they assume the same order up to the knots\footnote{The knots can be inferred from ${\bf L}_{\alpha} ({\bf L}_{\alpha}^{-1})_j={\bf e}_j$ and, assuming a sufficiently sparse graph support, ${\bf L}_{\alpha}^{2}({\bf L}_{\alpha}^{\dagger})_j=({\bf L}_{\alpha})_j$.} and parameterization, while their degrees differ in the synthesis case.\\
Furthermore, due to the existence of a nullspace for the rank-deficient case, we further note that the corresponding solution subspaces in both analysis and synthesis models, i.e. ${\bf L}_{\alpha}^{\dagger}{\bf \Psi}_{\Lambda^{\complement}}^T{\bf W}_{\alpha}$ and ${\bf L}_{\alpha}^{\dagger}{\bf \Psi}_{\Lambda^{\complement}}^T$, are orthogonal to $N({\bf L}_{\alpha})$. Here, we assume that ${\bf P}_{\alpha}$ is positive definite, invoking a localized perturbation effect; if this is not satisfied, the solution structure is further affected by a potentially dense ${\bf P}_{\alpha}^{-1}$ and, in case of singularity, ${\bf P}_{\alpha}^{\dagger}$ and $N({\bf P}_{\alpha})$.
\\
In the merging of these findings, we substantiate the relation between the cosparse analysis and sparse synthesis models for ${\bf L}_{\alpha}$:
\begin{thm}\label{main2}
The cosparse analysis graph signal model, described by the generalized graph Laplacian ${\bf L}_{\alpha}$, with ${\Lambda^{\complement}}$ as in Rm.\ \ref{constr}, encompasses perturbed (piecewise) complex exponential functions, generated by \\
$(i1)$ ${\bf L}_{\alpha}^{-1}{\bf \Psi}^T_{\Lambda^{\complement}}\tilde{{\bf c}}$ for $\alpha\neq 2\pi k/N,\ k\in\mathbb{N}$, with $\tilde{{\bf c}}\in\mathbb{C}^{|\Lambda^{\complement}|}$, and by\\
$(ii1)$ ${\bf L}_{\alpha}^{\dagger}{\bf \Psi}_{\Lambda^{\complement}}^T N({\bf \Psi}_{\Lambda}{\bf E}_{\alpha}{\bf \Psi}^T_{\Lambda^{\complement}}){\bf c}+z_1e^{i\alpha{\bf t}}+z_2e^{-i\alpha{\bf t}}$ for $\alpha=2\pi k/N,\ k\in\mathbb{N}$ and $\alpha\neq 0, k \pi$, with ${{\bf c}}\in\mathbb{C}^{|\Lambda^{\complement}|-2}$. \\
The sparse synthesis graph signal model \\
$(i2)$ is equivalent to the former, in the nonsingular case, for $\alpha\neq 2\pi k/N,\ k\in\mathbb{N}$, \\
$(ii2)$ otherwise, in the singular case for $\alpha=2\pi k/N,\ k\in\mathbb{N}$ and $\alpha\neq 0, k \pi$, it is described by perturbed piecewise linear complex exponential polynomials, generated by ${\bf L}_{\alpha}^{\dagger}{\bf \Psi}_{\Lambda^{\complement}}^T$. Hence, the analysis model represents a constrained version of it with respect to $N({\bf \Psi}_{\Lambda}{\bf E}_{\alpha}{\bf \Psi}^T_{\Lambda^{\complement}})$, and up to a translation by the nullspace $N({\bf L}_{\alpha})=z_1e^{i\alpha{\bf t}}+z_2e^{-i\alpha{\bf t}}$.
\end{thm}
\noindent In particular, when $\alpha\neq 2\pi k/N$ and ${\bf L}_{\alpha}$ is invertible, both the analysis and synthesis models generate ${{N}\choose{k}}$ subspaces of dimension $k$, for any $k=|\Lambda^{\complement}|$ with $k\in\lbrack 1\enskip N \rbrack$. The unique rank-deficiency of ${\bf L}_{\alpha}$ for $\alpha=2\pi k/N$ ($\alpha\neq 0, k\pi$) further gives rise to the following discrepancy:
\begin{rmk}
For $k=|\Lambda^{\complement}|\geq 3$, the number of subspaces ${{N}\choose{k}}$ of the cosparse analysis model coincides with that of the synthesis model, where, the former is composed of $N({\bf L}_{\alpha})$ of dimension $2$ and 
${\bf L}^{\dagger}_{\alpha}{\bf \Psi}_{\Lambda^{\complement}}^T{\bf W}_{\alpha}$ of dimension $|\Lambda^{\complement}|-2$; when $k<3$ and excluding the special case $\Lambda^{\complement}=\{m,(m+\frac{N}{2}),\ m\in V\}$, the former trivially has a single subspace $N({\bf \Psi}_{\Lambda}{\bf L}_{\alpha})=N({\bf L}_{\alpha})$.  
\end{rmk}
Similarly as for ${\bf L}$ in Rmk.\ \ref{smoth}, one can tailor the smoothness of the analysis subspace ${\bf L}^{\dagger}_{\alpha}{\bf \Psi}_{\Lambda^{\complement}}^T{\bf W}_{\alpha}$ by incorporating columns $({\bf P}_{\alpha})_j$ into the constraint, provided $\Lambda^{\complement}$ covers the support of $(({\bf P}_{\alpha})_j*({\bf L}_{C, \alpha})_k)_{mod N}$ for suitable $j,k\in V$.\\
For a sufficiently sparse graph $G_S$, the perturbed complex exponential linear polynomial $({\bf L}_{\alpha}^{\dagger})_j$ is annihilated by ${\bf L}_{C, \alpha}^{2}$, whose representer polynomial has $4$ exponential vanishing moments, with approximately sparse output ${\bf L}_{C, \alpha}^{2}({\bf L}_{\alpha}^{\dagger})_j={\bf P}_{\alpha}^{-1}({\bf L}_{C, \alpha})_j$, as well as by ${\bf L}_{\alpha}^{2}$, with sparse output ${\bf L}_{\alpha}^{2}({\bf L}_{\alpha}^{\dagger})_j=({\bf L}_{\alpha})_j$.
Furthermore, we can extend insights to higher-order, by noting the generative representation $N({\bf \Psi}_{\Lambda}{\bf L}_{\alpha}^k)=N({\bf L}_{\alpha})+{\bf L}_{\alpha}^{\dagger k}{\bf \Psi}_{\Lambda^{\complement}}^T{\bf W}_{\alpha}{\bf c}$, for $\Lambda^{\complement}$ as in Rm.\ \ref{constr}, such that ${\bf \Psi}_{\Lambda^{\complement}}^T{\bf W}_{\alpha}{\bf c}=({\bf L}_{\alpha})_j$ for some $j$, which encompasses signals of the form ${\bf L}_{\alpha}^{\dagger k-1}$. In an extension of Lemma \ref{summ} $(iii)$ for operator ${\bf L}_{\alpha}^k$, and provided the graph $G_S$ is sufficiently sparse, we discover that linear combinations of atoms $({\bf L}_{\alpha}^{\dagger k-1})_j$ are sparse with respect to ${\bf L}_{\alpha}^k$, with knots in the range of ${\bf L}_{\alpha}$. By additionally leveraging the property that ${\bf L}_{\alpha}^k$ annihilates up to $k-1$ degree complex exponential polynomials of the form ${\bf p}e^{\pm i\alpha{\bf t}}$, subject to a graph-dependent border effect, in conjunction with Lemma \ref{mppa}, we deduce that the range of ${\bf L}_{\alpha}^{\dagger k-1}$ gives rise to perturbed complex exponential polynomials of degree $k-1$, which are orthogonal to $e^{\pm i \alpha {\bf t}}$.
\\
\\
In Fig.\ \ref{fig:a1}, we depict the basis functions of ${\bf L}_{\alpha}^{-1}$ and their constrained linear combination ${\bf L}_{\alpha}^{-1}({\bf L}_{C, \alpha})_{ 30}$ for $\alpha=0.21$, respectively for the unweighted simple cycle and the circulant graph with generating set $S=\{1,2,3\}$ of dimension $N=64$. In Fig.\ \ref{fig:a2}, we depict the basis functions of ${\bf L}_{\alpha}^{\dagger}$ and their constrained (analysis) representation ${\bf L}_{\alpha}^{\dagger}({\bf L}_{C, \alpha})_{30}$ for $\alpha=4\pi /N$, for the same graphs as above. The representer polynomial of ${\bf P}_{\alpha}$ for the unweighted graph $G_S$ with $S=\{1,2,3\}$ is given by $P_{\alpha}(z)=(2+2\cos(\alpha)+2\cos(2\alpha))+(1+2\cos(\alpha))(z+z^{-1})+(z^2+z^{-2})$. Here, Fig.\ \ref{fig:a2} $(c)$ further illustrates the shape of the perturbation given by the columns of ${\bf P}_{\alpha}^{-1}$ at $\alpha=4\pi /N$, which exhibits a small support followed by approximate sparsity. In particular, the matrix ${\bf P}_{\alpha}^{-1}$ represents the difference between the MPPs of the simple cycle in $(a)$ and the extended circulant in $(b)$, and accordingly perturbs the functions underlying the former around its discontinuities, resulting in the constrained analysis representation ${\bf L}_{\alpha}^{\dagger}({\bf L}_{C, \alpha})_{30}={\bf P}_{\alpha}^{-1}({\bf I}_N-\frac{1}{N}{\bf E}_{\alpha})_{30}$, as depicted in $(b)$.\\ 
Further, in Fig.\ \ref{fig:acomp}, we directly compare the functions ${\bf L}_{C, \alpha}^{\dagger}({\bf L}_{C, \alpha})_{30}$ for $\alpha=4\pi/N$ and $({\bf L}_{C, \alpha}^{-1})_{30}$ for $\alpha=0.21$ on the simple cycle, the latter of which has been normalized. Both functions are complex exponentials (i.e. trigonometric functions in the real domain) of comparable order, with a discontinuity at the same location, and represent the analysis subspaces of the singular and non-singular ${\bf L}_{\alpha}$. They are respectively given by ${\bf L}_{C, \alpha}^{\dagger}({\bf L}_{C, \alpha})_{30}={\bf e}_{30}-\frac{2\cos(\alpha \tilde{{\bf t}})}{N}$ and $({\bf L}_{C, \alpha}^{-1})_{30}=-\frac{1}{2\sin(\alpha)}(\cot(\alpha N/2)\cos(\alpha \tilde{{\bf t}})+\sin(\alpha \tilde{{\bf t}}))$ (the latter of which follows from Lemma \ref{inva} using Euler's formula), for $\tilde{{\bf t}}=({\bf t}\ast {\bf e}_{30})_{mod N}$.\\
In order to complete the comparison, we further depict in Fig.\ \ref{fig:a3} the same set of functions and constraints for the standard graph Laplacian case with $\alpha=0$. Here, the representer polynomial of ${\bf P}_{\alpha=0}={\bf P}_{G_S}$ at $\alpha=0$ takes the form $P_{G_S}(z)=6+3(z+z^{-1})+(z^2+z^{-2})$. \\
While generating different types of functions in the synthesis case of the unconstrained ${\bf L}_{\alpha}^{\dagger}$ and ${\bf L}^{\dagger}$, i.e. ranging from linear complex exponential polynomials to quadratic polynomials, their constrained analysis representations, through multiplication of the former respectively by columns $({\bf L}_{C, \alpha})_j$ and $({\bf L}_{C})_j$, give rise to the same two-fold reduction in degree. While in the case of $\alpha=4\pi/N$, we effectively generate (perturbed) complex exponentials as a result of the two exponential vanishing moments in $({\bf L}_{C, \alpha})_j$, in the graph Laplacian case with $\alpha=0$, we obtain a piecewise constant function due to the double vanishing moment property of $({\bf L}_{C})_j$. The perturbations in form of ${\bf P}_{\alpha}^{-1}$ and ${\bf P}_{G_S}^{-1}$ nevertheless remain relatively close in value to the proximity of their inverses for $\alpha=0$ and $\alpha=4\pi/N\approx0.1963$.


 \begin{figure}
 \centering
  \begin{subfigure}[htbp]{0.49\textwidth}
{\includegraphics[width=2.45in]{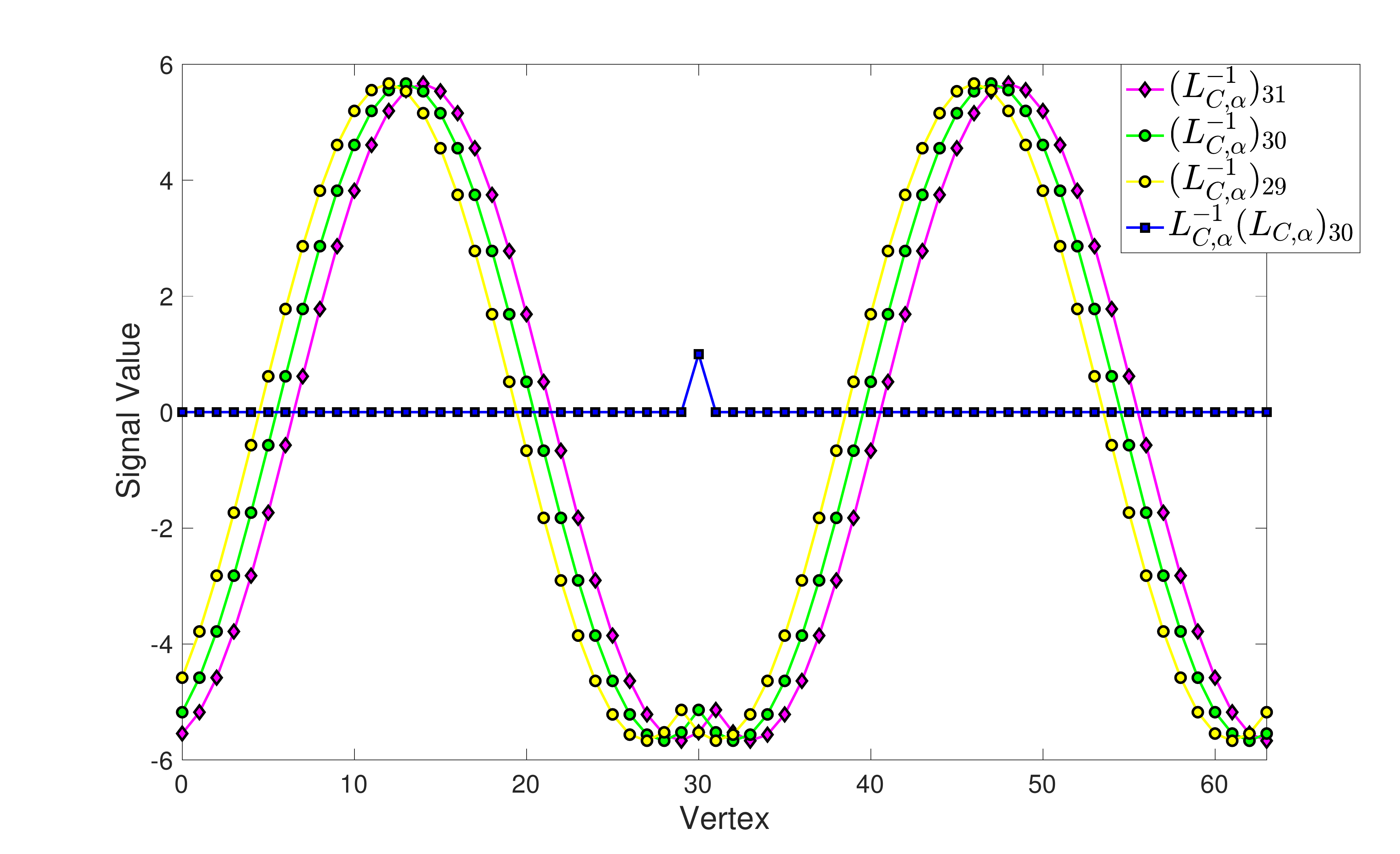}} 
\caption{\small{Functions on $G_S$ with $S=\{1\}$}}
\end{subfigure}
\begin{subfigure}[htbp]{0.49\textwidth}
 { \includegraphics[width=2.45in]{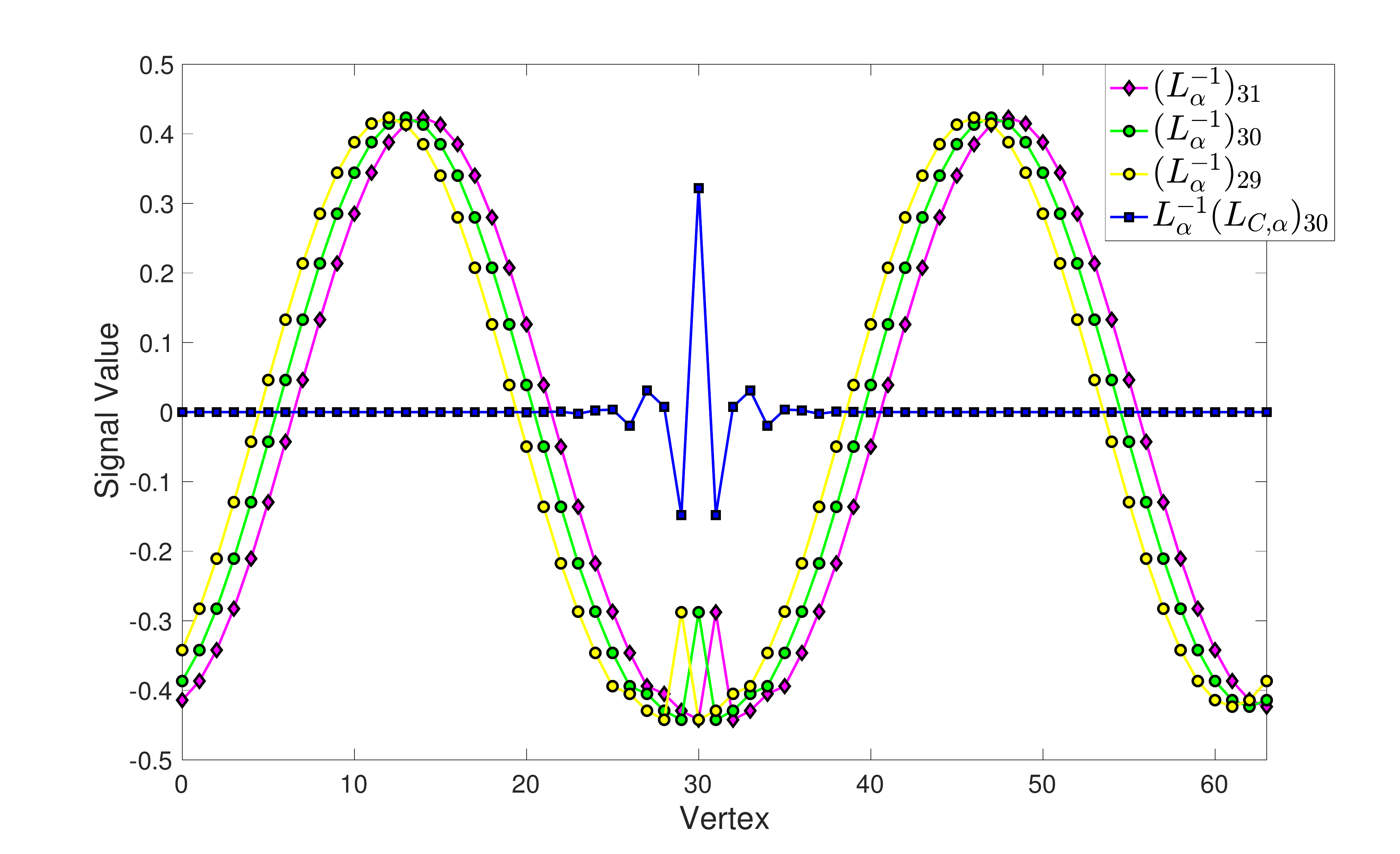}}
 \caption{\small{Functions on $G_S$ with $S=\{1,2,3\}$}}
 \end{subfigure}
	\caption{Comparison of signal models on circulant graphs for ${\bf L}_{\alpha}^{-1}$ at $\alpha=0.21$.}
\label{fig:a1}\end{figure}

 \begin{figure}
 \centering
  \begin{subfigure}[htbp]{0.49\textwidth}
{\includegraphics[width=2.5in]{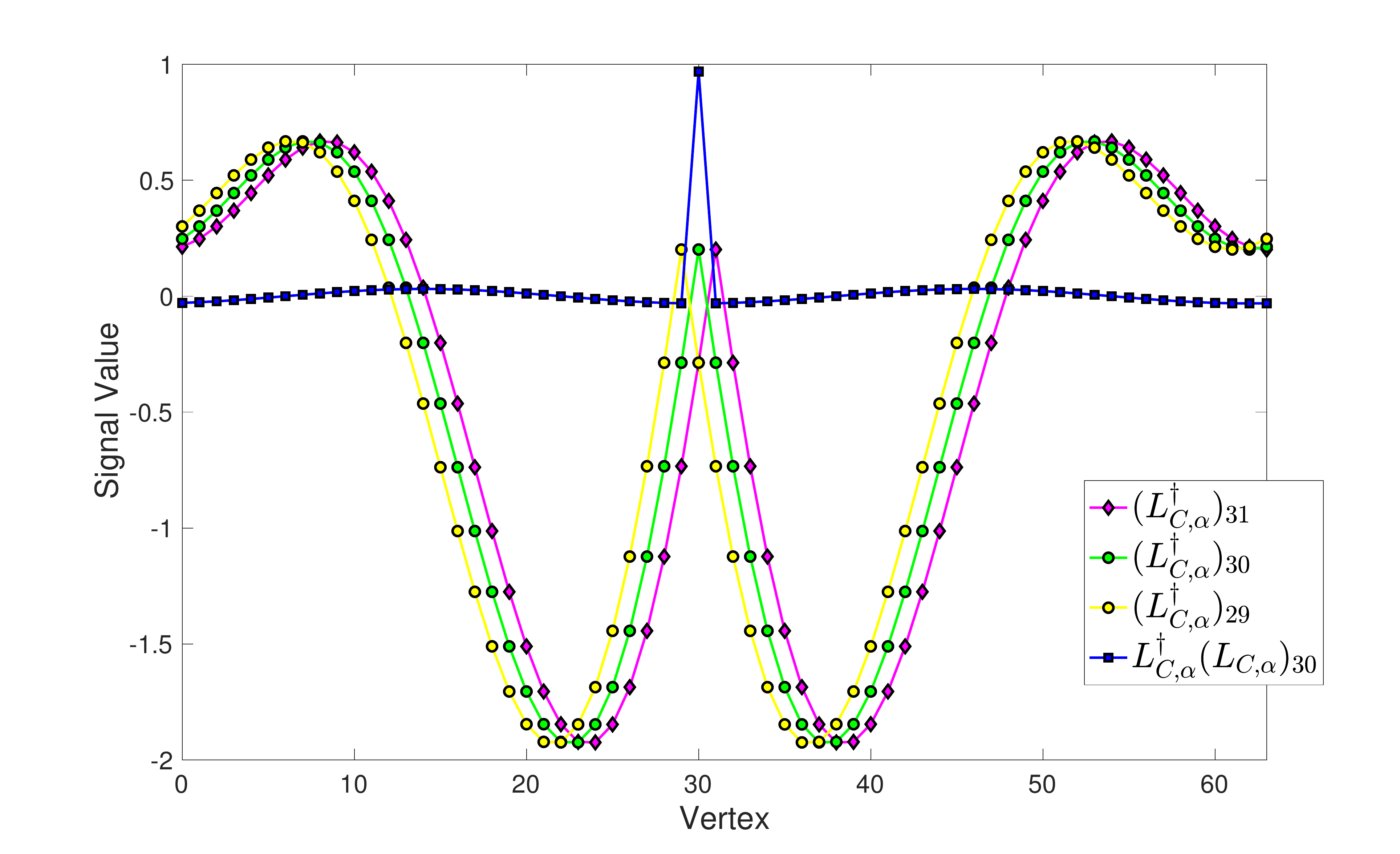}} 
\caption{\small{Functions on $G_S$ with $S=\{1\}$}}
\end{subfigure}
\begin{subfigure}[htbp]{0.49\textwidth}
 { \includegraphics[width=2.5in]{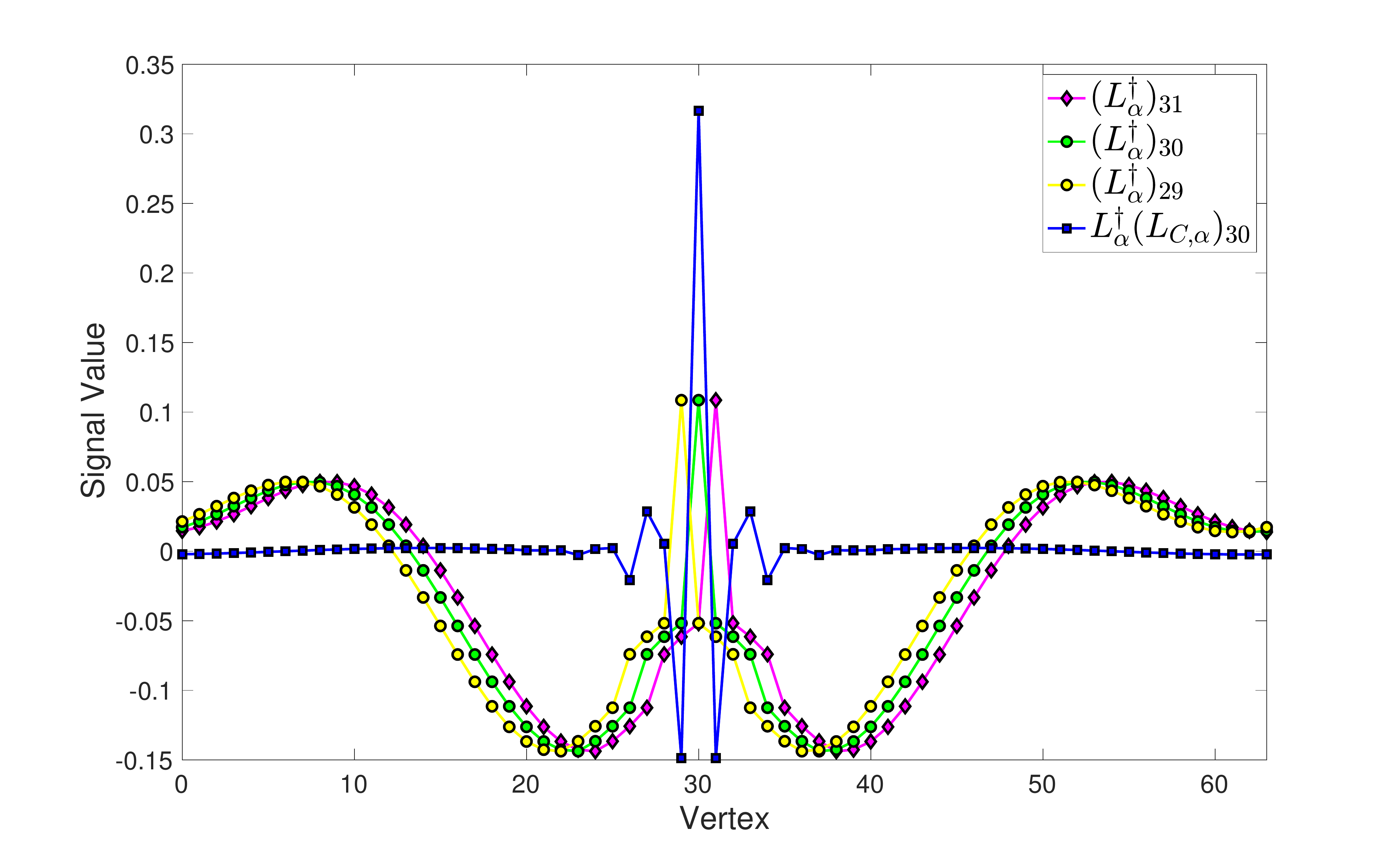}}
 \caption{\small{Functions on $G_S$ with $S=\{1,2,3\}$}}
 \end{subfigure}
 \begin{subfigure}[htbp]{0.49\textwidth}
 { \includegraphics[width=2.9in]{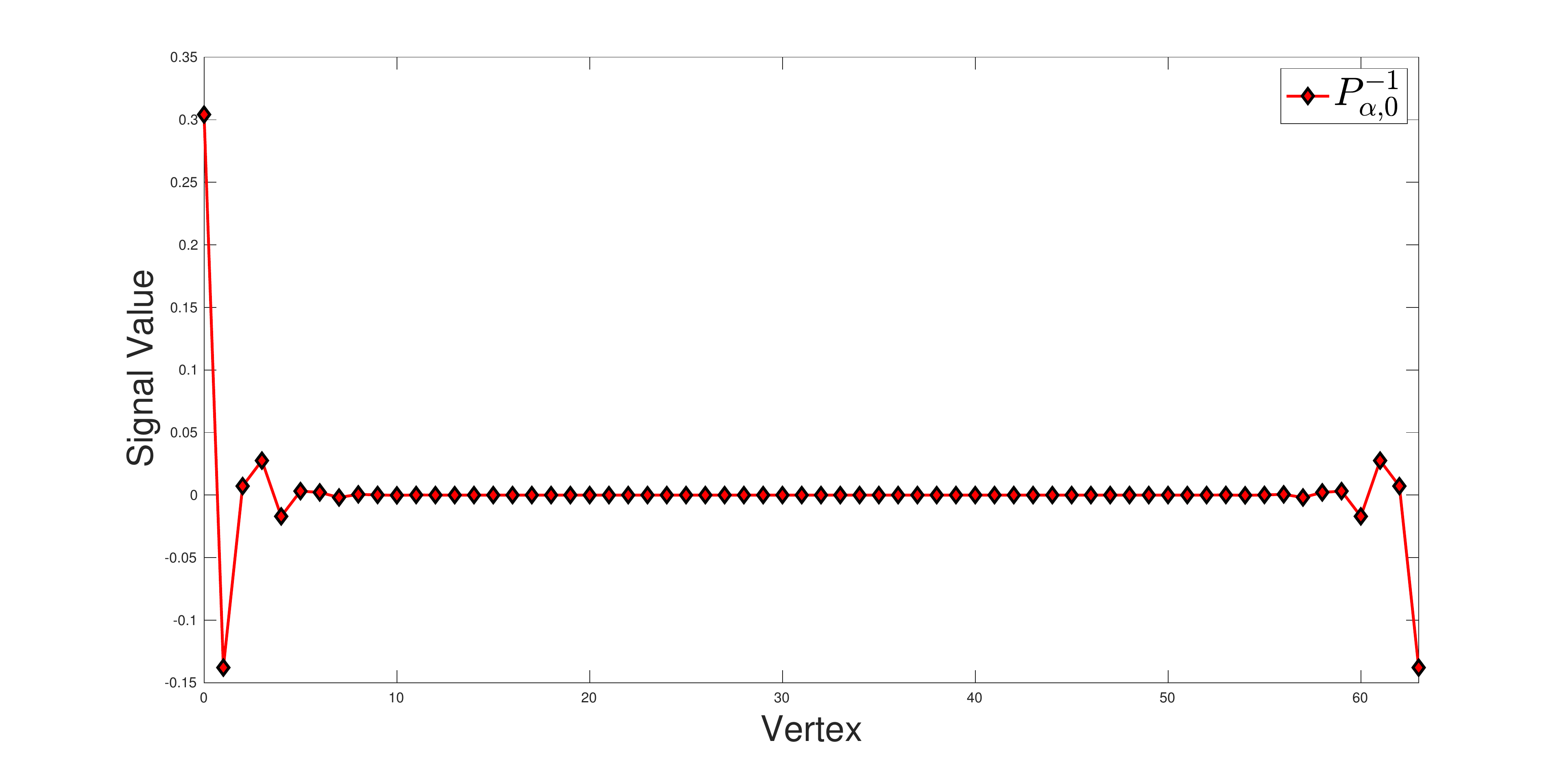}}
 \caption{\small{Perturbations on $G_S$ with $S=\{1,2,3\}$}}
 \end{subfigure}
	\caption{Comparison of signal models on circulant graphs for ${\bf L}_{\alpha}^{\dagger}$ at $\alpha=4\pi/N$.}
\label{fig:a2}\end{figure}

 \begin{figure}
 \centering
 { \includegraphics[width=2.6in]{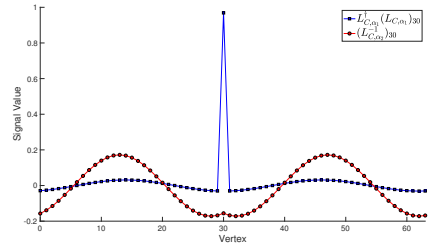}}
 \caption{\small{Comparison of ${\bf L}_{C, \alpha_1}^{\dagger}{\bf L}_{C, \alpha_1}$, $\alpha_1=4\pi/N$, and ${\bf L}_{C, \alpha_2}^{-1}$, $\alpha_2=0.21$, on $G_S$ with $S=\{1\}$}}
\label{fig:acomp} \end{figure}

 \begin{figure}
 \centering
  \begin{subfigure}[htbp]{0.49\textwidth}
{\includegraphics[width=2.5in]{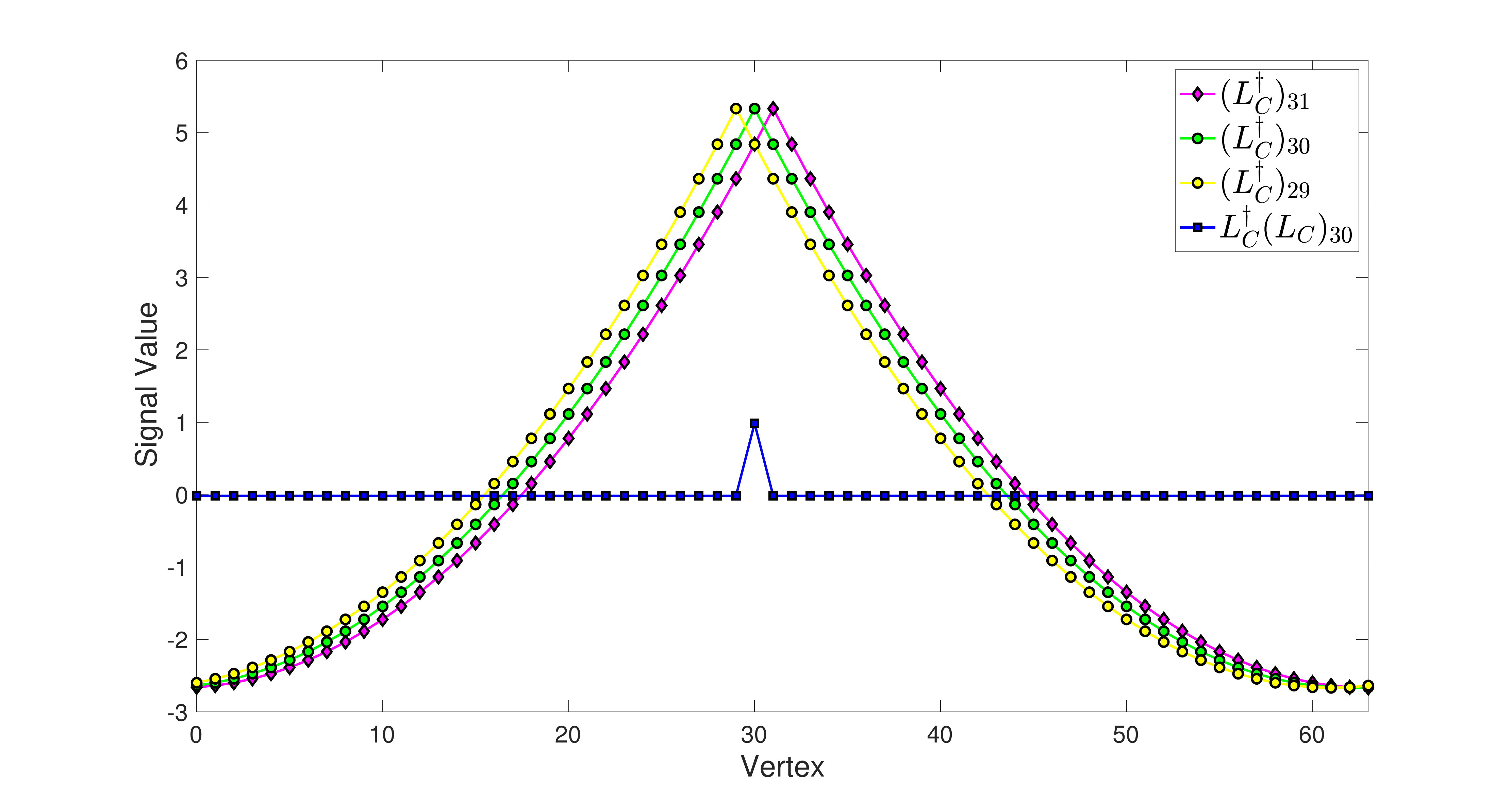}} 
\caption{\small{Functions on $G_S$ with $S=\{1\}$}}
\end{subfigure}
\begin{subfigure}[htbp]{0.49\textwidth}
 { \includegraphics[width=2.5in]{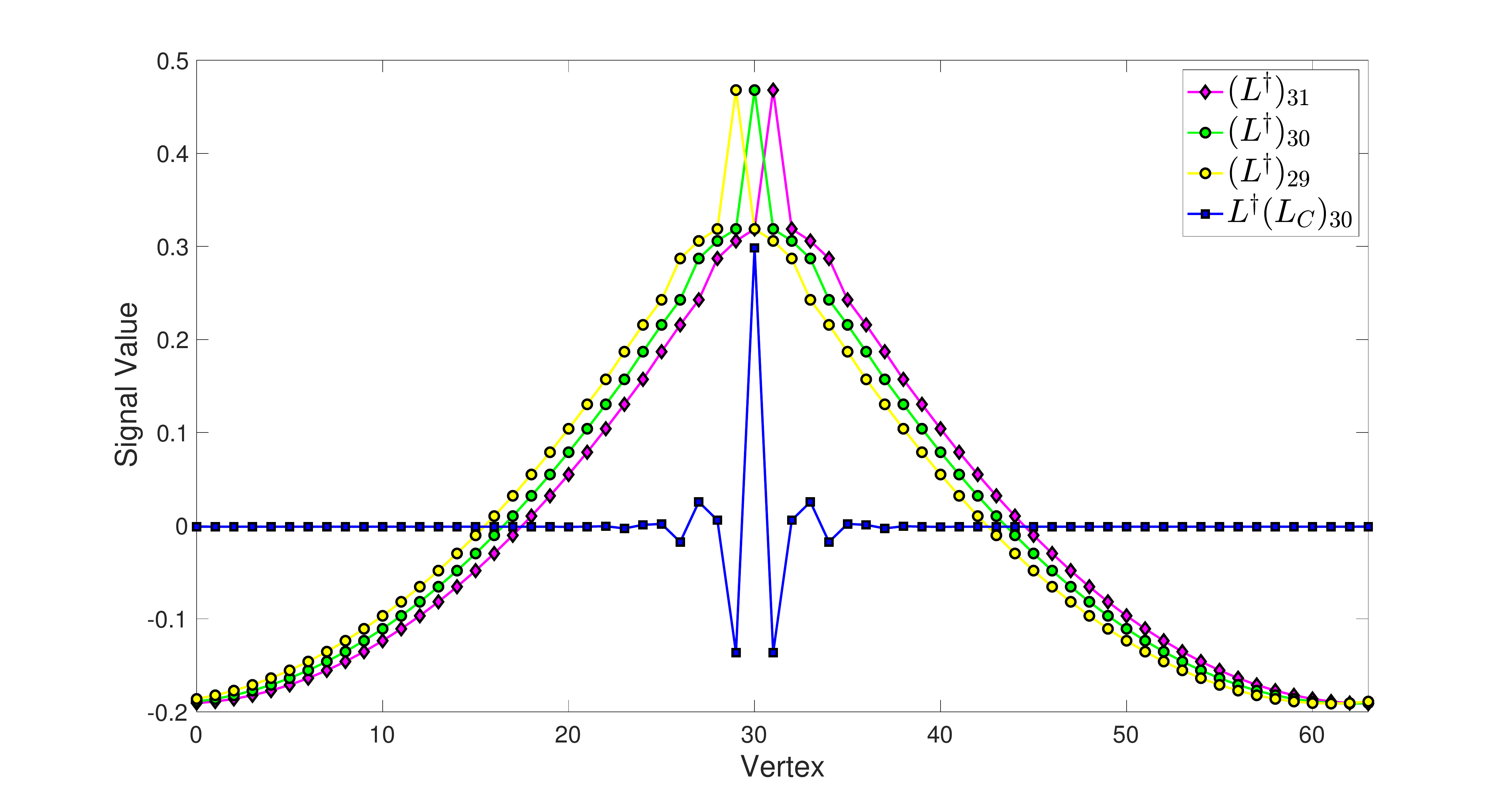}}
 \caption{\small{Functions on $G_S$ with $S=\{1,2,3\}$}}
 \end{subfigure}
 
 \begin{subfigure}[htbp]{0.49\textwidth}
 { \includegraphics[width=2.9in]{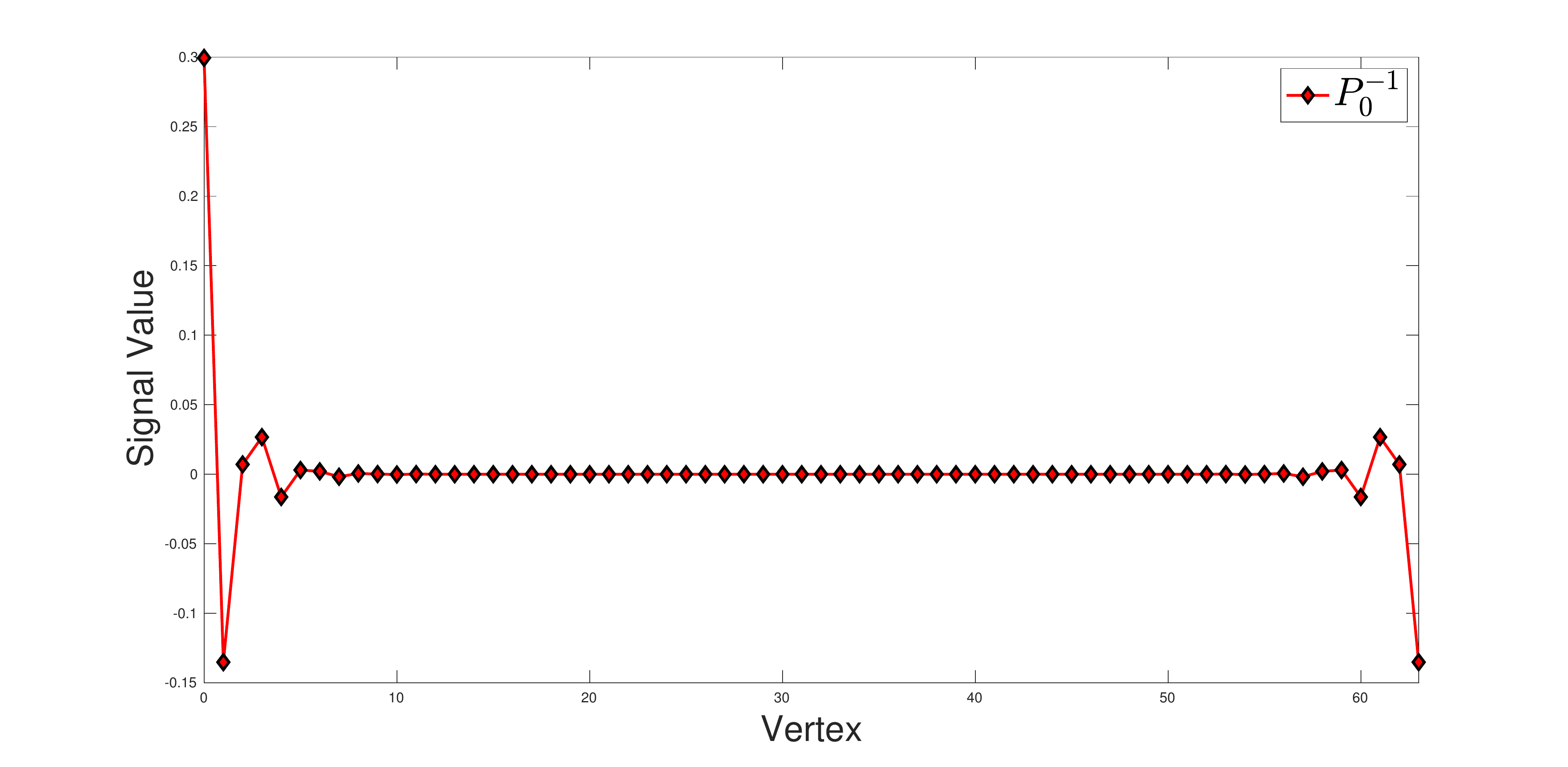}}
 \caption{\small{Perturbations on $G_S$ with $S=\{1,2,3\}$}}
 \end{subfigure}
	\caption{Comparison of signal models on circulant graphs for ${\bf L}^{\dagger}$.}
\label{fig:a3}\end{figure}

\section{Model-based Uniqueness and Recovery Guarantees on Graphs}
Building on the preceding qualitative and quantitative analysis of the subspaces of graph-based UoS models, we seek to leverage the acquired insight for the solution of linear inverse problems subject to graph-structured constraints, specifically, the quantification of specialized guarantees for the uniqueness of a solution. Furthermore, we discuss the derived graph-based UoS models in the context of model-based Compressed Sensing (CS), particularly with regard to the design of structured UoS models with desirable properties for refined sampling and recovery results.
\subsection{Uniqueness}
Let ${\bf M}\in\mathbb{R}^{m\times N}$, $m<N$, denote a suitable (graph-)measurement matrix with linearly independent rows and consider the problem of identifying the unique (co)sparse solution of ${\bf y}={\bf M}{\bf x}$, where ${\bf x}$ belongs to a graph Laplacian based UoS model on a connected graph.\\
For the synthesis model with representation ${\bf x}={\bf D}{\bf c}$, in order to uniquely identify the unknown $k$-sparse signal ${\bf c}$, we require  $2k<spark({\bf M D })\leq m+1$, where the $spark$ measures the minimum number of linearly dependent columns of a matrix \cite{spa}. On the other hand, for the analysis model, an equivalent measure to the spark, $\kappa_{{\bf \Omega}}(l):=\max_{|\Lambda|\geq l}dim(W_{\Lambda})$, is established \cite{cos}, which quantifies the interdependency between rows through the maximum analysis subspace dimension of $W_{\Lambda}$ for a given cosparsity level $l$. Thereby, to determine the $l$-cosparse signal ${\bf x}$ with ${\bf \Omega}_{\Lambda}{\bf x}={\bf 0}_{\Lambda}$, given mutually independent ${\bf M}$ and ${\bf \Omega}$ (i.e.\ their rows do not have non-trivial linear dependencies, according to \cite{cos}), we require $2\kappa_{{\bf \Omega}}(l)\leq m$ as a sufficient condition (see Prop. $3$, \cite{cos}). \\
Due to its uniquely characterizable linear dependencies and following Prop. \ref{prop1}, the interdependency measures for ${\bf L}$ can be  respectively  identified as $spark({\bf L}^{\dagger})=N$ and $\kappa_{{\bf L}}(l)=N-l=k$ for $l<N-1$; when, additionally, the graph is disconnected with $c$ components, we further have $\kappa_{{\bf L}}(l)=k+c-1$. According to \cite{ylu}, a necessary condition for the recovery of cosparse signals belonging to a UoS model is given by $m\geq \tilde{\kappa}_{{\bf \Omega}}(l)$, with $\tilde{\kappa}_{{\bf \Omega}}(l):=\max \{dim (W_{\Lambda_1}+W_{\Lambda_2}):\ |\Lambda_i|\geq l, i=1,2\}$ quantifying the maximum dimension of the sum of any two subspaces in the union. In the graph Laplacian-based UoS model, $N({\bf L})$ forms part of every subspace in the union and its dimension is known, which facilitates the following tighter result for the analysis model:
\begin{cor}
For mutually independent ${\bf M}\in\mathbb{R}^{m\times N}$ and ${\bf L}$ on an undirected graph $G$, the graph-Laplacian-based problem 
\[{\bf M}{\bf x}={\bf y}\enskip\text{with}\enskip ||{\bf L}{\bf x}||_0\leq N-l=k  \] 
 has at most one solution, provided $k>1$, if \\
 $(i)$ $m\geq 2k-1$, when the graph is connected,\\
 $(ii1)$ $m\geq 2k-2+c$, when the graph is disconnected with $c$ components. \\
 $(ii2)$ If ${\bf x}\in\bigcup_{|\Lambda|=N-k} \textit{W}_{\Lambda}$, for $\textit{W}_{\Lambda}:=N({\bf L}_{\Lambda})$, subject to the additional constraint, $|\Lambda_i|< N_i-1$, the latter condition can be relaxed to $m\geq 2k-c$.
 \end{cor}
 \begin{proof}
 Follows from $dim (U+W)=dim(U)+dim(W)-dim(U\cap W)$ for the sum $U+W=\{{\bf u}+{\bf w}:\ {\bf u}\in U,\ {\bf w}\in W\}$ of any two subspaces $U,W$ \cite{hohn}. We have for any $U,W$ in the resulting union of subspaces $\bigcup_{|\Lambda_i|\geq l} W_{\Lambda_i}$ that $\max \{dim(U)+dim(W)\}=2k$ for $c=1$, $\max \{dim(U)+dim(W)\}=2(k+c-1)$ for $c>1$, and $dim(U\cap W)=c$ in this case. If the cosupport constraint of $(ii2)$ is additionally applied for $c>1$, we have $\max \{dim(U)+dim(W)\}=2k$.
 \end{proof}
\noindent 
Notably, the employment of the measure $\tilde{\kappa}_{{\bf \Omega}}(l)$ leads to a minimum sampling requirement slightly below the number of degrees of freedom of signals in the analysis model. In particular, for a signal on a connected graph, the total number of degrees of freedom is $2k$, with $2k-1$ for representation ${\bf L}^{\dagger}{\bf \Psi}_{\Lambda^{\complement}}^T{\bf W}{\bf c}$ with $|\Lambda^{\complement}|=k$, as a result of its linearly dependent coefficients, and $1$ for its nullspace vector coefficient. Accordingly, for a disconnected graph with $c$ components, this can range between $2k$, if the cosupport constraint is applied, and $2k+c-1$.\\
Nevertheless, we note that while the measures $\kappa_{{\bf \Omega}}(l)$ and  $\tilde{\kappa}_{{\bf \Omega}}(l)$ quantify a dependency, the unique structural quality and composition of the underlying subspaces of ${\bf L}$ (or ${\bf S}^T$) are not leveraged in this result. \\
{\bf Recovery}. In \cite{recov}, the problem of recovering the $k$-sparse solution ${\bf x}'$ to $\min ||{\bf x}||_0,\ \text{subject to}\ {\bf y}={\bf S}^T{\bf x}$, on a connected graph $G$, which contains cycles, was considered, where ${\bf S}^T\in\mathbb{R}^{|V|\times|E|}$ functions as the sampling operator (or dictionary) of $G$ and ${\bf y}={\bf S}^T{\bf x}'$. It is established i.a. that the \textit{spark} of ${\bf S}^T$ is associated with the length $g$ of the shortest simple cycle (or girth) of $G$, i.e. the smallest number of linearly dependent columns in ${\bf S}^T$, which gives rise to the uniqueness of a $k$-sparse solution ${\bf x}'$ to the optimization problem for $g>2k$ (see Prop. $2$, \cite{recov}). Ultimately, the structured sparsity of ${\bf y}$ is leveraged in the algorithmic recovery of ${\bf x}'$. When the graph is acyclic, ${\bf S}^T\in\mathbb{R}^{|V|\times|V|-1}$ has linearly independent columns and the problem is trivially solved by applying the left inverse ${{\bf S}^{\dagger}}^T$ on ${\bf y}$ to obtain ${\bf x}'$. \\
In the present scenario, given the constrained analysis-sparse representation ${\bf x}={\bf L}^{\dagger}{\bf c}$, one may consider the problem of identifying the sparse vector ${\bf t}$ in the structurally sparse vector of knots ${\bf c}={\bf S}^T{\bf t}$, based on \cite{recov}, for either an acyclic and non-acyclic graph.

\subsection{A Desirable UoS Model}
In \cite{uos}, sampling theorems for signals which belong to a finite-dimensional UoS model are developed, with the union $U_K$ of $L$ linear subspaces defined as
\[U_K=\bigcup_j^L \tilde{U}_j=\{{\bf x}={\bf \Omega}_j{\bf a}_j,\ {\bf \Omega}_j\in\mathbb{R}^{N\times k_j},{\bf a}\in\mathbb{R}^{k_j}\}\]
where ${\bf \Omega}_j$ are the subspace bases of dimension $k_j\leq K$. While the invertibility of ${\bf M}$ over the UoS guarantees the unique identification of any signal in that union, in practice, one requires additional properties such as robustness to noise, which is related to distance preserving properties of ${\bf M}$, also known as the Restricted Isometry Property (RIP) \cite{uos}. In particular, it is established that for a stable sampling scheme, the number of required measurements depends necessarily logarithmically on the number of subspaces in $U_K$:\\
\\
{\bf Theorem}. (see Thm.\ $3$, \cite{uos2}) \textit{Let $U_K$ be the union of $L$ subspaces of dimension $K$ in $\mathbb{R}^N$. Then, for any $t>0$ and 
\[m\geq \frac{1}{c(\delta_{U_K}/6)}\left(2 ln (L)+2 K ln\left(\frac{36}{\delta_{U_K}}\right)+t\right)\]
there exists a function $c(\delta)>0$ and matrix ${\bf M}\in\mathbb{R}^{m\times N}$, such that $\forall {\bf y}_1, {\bf y}_2\in U_K$, the $U_K$-RIP 
\[(1-\delta_{U_K})||{\bf y}_1-{\bf y}_2||_2^2\leq ||{\bf M}({\bf y}_1-{\bf y}_2)||_2^2\leq (1+\delta_{U_K}) ||{\bf y}_1-{\bf y}_2||_2^2\]
holds with RIP constant $\delta_{U_K}$. If ${\bf M}$ is generated by randomly drawing i.i.d. entries from an appropriately scaled subgaussian distribution, it satisfies the $U_K$-RIP with probability at least $1-e^{-t}$.}\\
\\
\noindent As is evident from the expression, it is generally desirable for a given UoS model to contain a small number $L$ of subspaces of small maximum dimension $K$ in order to achieve robust recovery at a smaller number of measurements.\\
Model-based Compressed Sensing is concerned with the formulation of theoretical guarantees and development of algorithms for the recovery of sparse signals with structural dependencies imposed on (or inherent in) their locations and coefficients \cite{bar}. Specifically, it seeks to exploit the occurrence of structure in a signal model for the improvement of performance guarantees, such as a relaxed RIP constraint with a reduction in the number of samples necessary for recovery. Previous model-based CS frameworks have featured concepts such as block-sparsity, wavelet-trees etc \cite{bar}, which focus on restricting the location, rather that the amplitude, of the sparse coefficients, thereby producing a reduced number of subspaces. \\
\\
In the presented cosparse analysis graph Laplacian UoS model for a connected graph, while the inherent linear dependencies of the coefficients directly affect the dimension of the generated subspaces, the number of subspaces itself remains the same (except for the case $k=1$), unless concrete location constraints, such as the omission of certain vertex sets, are imposed. In the disconnected case, the number of subspaces is naturally affected by the block-diagonal matrix structure and rank-deficiency, which increase the dimension of the nullspace, thereby reducing the number of possible unique subspace combinations and increasing, where applicable, their total dimension. As alluded to in Sect.\ $3.4$, one may further restrict the admissible sparsity level per subgraph, so as to exclude degenerate cases and obtain subspaces exclusively of the same dimension as well as control the dimension itself, thereby effectively further reducing the total number of subspaces.\\
Specifically, in contrast to the models of block-sparsity, where coefficients cluster in blocks and joint sparsity, where the support is shared across blocks, \cite{bar}, the presented cosparse analysis UoS model for disconnected graphs defines independent blocks (subgraphs) which may feature different sparsity patterns and levels per block; in addition, its coefficient structure is constrained through linear dependencies within each block.\\
\\
In particular, we can express the cosparse analysis graph-based UoS model for a graph with $c$ connected components and maximum subspace dimension $K$ as the union $U^{c}_K=\bigcup_s^L \tilde{U}^c_s$ with subspaces ${\bf \Omega}_{s}{\bf a}_s=N({\bf L}){\bf a}_{s,1}+{\bf L}^{\dagger}{\bf Z}_{s}{\bf a}_{s,2}$ for $s=1,...,L$ possible subspace configurations, where \[{\bf Z}_{s}=span(({\bf e}_m-{\bf e}_n),\ m,n\in \Lambda_{t,s}^{\complement}, t=1,...,c)\in\mathbb{R}^{N\times K-c}\]
with $\sum_{t=1}^c |\Lambda_{t,s}^{\complement}|=K$, ${\bf a}_{s,1}\in\mathbb{R}^{c}$ and ${\bf a}_{s,2}\in\mathbb{R}^{K-c}$. Here, ${\bf Z}_s$ simultaneously serves as a basis for structured sparse signals and represents the linear constraints which induce easily characterizable rank deficiencies. We discover that these linear dependencies, arising from the connectivity of the graph, and manifesting themselves in form of linear constraints imposed on the analysis subspaces, signify an instance of a structured sparsity model. This can be specifically defined as the UoS 
\[\tilde{U}^c_{K-c}=\bigcup_{s\in\Sigma} \tilde{Z}_{s},\ \text{s.t.}\ \tilde{Z}_{s}:=\{ {\bf c}: {\bf c}_{\Lambda_{t,s}^{\complement}}\in\mathbb{R}^{K_t}, \sum_{i=0}^{K_t-1} c(i)=0, \ {\bf c}_{\Lambda_{t,s}}={\bf 0}_{\Lambda_{t,s}}, \ t=1,...,c \}\]
with set $\{\Lambda_{t,s}^{\complement}\}_{(t,s)\in \Tau\times \Sigma}$ of admissible supports with $|\Lambda_{t,s}^{\complement}|>1$ and total cardinality $\sum_{t=1}^c |\Lambda_{t,s}^{\complement}|=K> 1$, parameterized by configuration $s\in\Sigma$ for connected component $t\in\Tau=\lbrack 1\ c \rbrack$. Therefore, each subspace $ \tilde{Z}_{s}$ contains all sparse vectors ${\bf c}$ with support in $\bigcup_{t=1}^c \Lambda_{t,s}^{\complement}$, necessarily resulting in $L <{{N}\choose{K}}$ subspaces of reduced dimension $K-c$. Here, one can additionally constrain $|\Lambda_{t,s}^{\complement}|$ individually per connected graph component.\\
While the specific properties of the present graph model were the result of its characteristic rank-deficiency and structure, they motivate the creation of a UoS model with desirable properties by i.a. applying a blockwise restriction of the sparsity level, thereby omitting certain subspace combinations, and/or removing the nullspace in rank-deficient systems to reduce the overall subspace dimension.\\
As such, one may further tailor the characteristics of analysis-driven models toward model-based compressed sensing, i.a. by inducing a small number of subspaces of a certain dimension and composition. In general, in order to reduce the number of subspaces for a graph of arbitrary connectivity, one needs to restrict the parameter set $\Lambda^{\complement}$ such that the sparse signal ${\bf c}$ has its support only on designated subgraphs (or alternatively the set of rows $\Lambda$ of ${\bf L}$ correspond to vertex sets outside those subgraphs), resulting in ${{N_i}\choose{k}}$ subspaces with $k=|\Lambda^{\complement}|$ and $N_i<N$, while the support size $k_i$ may be further restricted per subgraph, resulting in ${{N_1}\choose{k_1}}{{N_2}\choose{k_2}}...{{N_c}\choose{k_c}}<{{N}\choose{k}}$ subspaces. Such restricted node-sets may involve connected subgraphs, cycles or paths, trees of a certain size etc.\\
We leave the development of specialized recovery algorithms and experimental performance assessment for future work.

\section{Conclusion}
In this work, we have provided a comprehensive analysis of union of subspaces models associated with structured difference matrices on graphs; specifically, we have substantiated the discrepancy between the cosparse analysis and sparse synthesis models. We have established that for the graph Laplacian matrix of connected (and disconnected) graphs, the cosparse analysis constitutes a special instance of the sparse synthesis model, subject to a structured sparsity constraint, with differences in the distribution in subspace dimension and number as a result of increased rank-deficiency. In particular, we have shown that the imposed constraint, and associated model-discrepancy, is a result of the Fredholm Alternative. By building a bridge between the MPP and the concept of Green's functions, we have further characterized the exact composition of the underlying functions defining the subspaces for the special case of circulant graphs. Furthermore, the investigation of a generalized, parametric graph Laplacian on circulant graphs with variable singularity, has revealed transitional properties between equivalence and non-equivalence for the two models. At last, we have unified results in the context of uniqueness and recovery of (co)sparse signals in UoS models and discussed the development of ideally constrained instances thereof.\\
In future work, it would be of interest to leverage the developed theory on graph-based UoS models and/or block-wise sparsity for the creation of refined UoS signal models with enhanced, robust sampling and recovery guarantees as well as to further investigate the question of what constitutes a meaningful structured sparsity model. In addition, the investigation of UoS models (their properties and associated guarantees) for rank-deficient rectangular operators on and beyond graph difference matrices constitutes an avenue worth pursuing.

\appendix

\section{}
\label{appa}
\begin{proof}[Proof of Corollary \ref{disc1}]
We proceed similarly as for the proof of Prop.\ \ref{prop1}, by considering the solution to ${\bf L}{\bf u}={\bf \Psi}_{\Lambda^{\complement}}^T{\bf w}$, for suitable ${\bf w}\in\mathbb{R}^{|\Lambda^{\complement}|}$ in the span of basis ${\bf W}\in\mathbb{R}^{|\Lambda^{\complement}|\times k}$ of unknown rank $k$, such that ${\bf \Psi}_{\Lambda^{\complement}}^T{\bf w}\perp N({\bf L})$ is satisfied. Under suitable labelling and wlog, Eq.\ (\ref{eq:nullspace}) then assumes the following blockwise structure
\[{\bf \Psi}_{\Lambda}{\bf L}{\bf L}^{\dagger}{\bf \Psi}_{\Lambda^{\complement}}^T{\bf w}={\bf \Psi}_{\Lambda}\scalebox{0.85}{$\begin{pmatrix} {\bf I}_{N_1}-\frac{1}{N_1}{\bf J}_{N_1}&0&0&...\\0&  {\bf I}_{N_2}-\frac{1}{N_2}{\bf J}_{N_2}&&\\ &..&..&\\0&0&& {\bf I}_{N_t}-\frac{1}{N_t}{\bf J}_{N_t}\end{pmatrix}$}{\bf \Psi}_{\Lambda^{\complement}}^T{\bf w}={\bf 0}_N\]
with \[{\bf L}=\begin{bmatrix} {\bf L}_1 & 0 &\dots&\\0 & {\bf L}_2 &0 &\dots\\ &\dots&&\\0 & \dots & &{\bf L}_t\\
\end{bmatrix}\enskip\text{and}\enskip {\bf L}^{\dagger}=\begin{bmatrix} {\bf L}^{\dagger}_1 & 0 &\dots&\\0 & {\bf L}^{\dagger}_2 &0 &\dots\\& \dots&&\\0 & \dots & &{\bf L}^{\dagger}_t\\
\end{bmatrix},\]
where ${\bf L}_k\in\mathbb{R}^{N_k\times N_k}$ refers to the graph Laplacian on the $k$-th connected component. We define $\Lambda=\Lambda_1\cup\Lambda_2...\cup\Lambda_t$ with vertex sets (blocks) $C_k=\Lambda_k\cup\Lambda_k^{\complement}$ of cardinality $|C_k|=N_k$ per connected component, assuming that each set $\Lambda_k^{\complement}$ is non-empty, and observe that the constraint matrix ${\bf W}$ takes the form 
\begin{equation}\label{eq:newnull}{\bf W}=\begin{bmatrix} {\bf W}_1 & 0 &\dots&\\0 & {\bf W}_2 &0 &\dots\\ \dots&&&\\0 & \dots & &{\bf W}_t\\
\end{bmatrix}\end{equation} with ${\bf W}_k=N({\bf J}_{\Lambda_k,\Lambda_k^{\complement}})$ as in Eq.\ (\ref{eq:const}) for $k=1,...,t$; here, we note that ${\bf W}$ is not blockdiagonal since its blocks ${\bf W}_k$ are rectangular whose individual sizes depend on the corresponding sets $\Lambda_k$. The resulting solution subspace for $N({\bf \Psi}_{\Lambda}{\bf L})$ is defined as the span of 
\[{\bf L}^{\dagger}{\bf \Psi}_{\Lambda^{\complement}}^T{\bf W}=\scalebox{0.85}{$\begin{bmatrix} {\bf L}_1^{\dagger}\tilde{{\bf \Psi}}_{\Lambda_1^{\complement}}^T {\bf W}_1 & 0 &\dots&\\0 & {\bf L}_2^{\dagger}\tilde{{\bf \Psi}}_{\Lambda_2^{\complement}}^T{\bf W}_2 &0 &\dots\\ \dots&&&\\0 & \dots & & {\bf L}_t^{\dagger}\tilde{{\bf \Psi}}_{\Lambda_t^{\complement}}^T{\bf W}_t\\
\end{bmatrix}$},\enskip \text{with}\enskip \tilde{{\bf \Psi}}_{\Lambda_k}\in\mathbb{R}^{|\Lambda_k|\times |C_k|}\]
and $N({\bf L})=\{{\bf 1}_{C_1},...,{\bf 1}_{C_t}\}$, which comprises the indicator vectors of the $t$ connected components and is of rank $t$. 
\\
Here, the former requires $|\Lambda_k^{\complement}|\geq 2$ for ${\bf W}_k$ to be non-empty. We note that ${\bf W}$, and by extension ${\bf L}^{\dagger}{\bf \Psi}_{\Lambda^{\complement}}^T{\bf W}$, has rank at least $|\Lambda^{\complement}|-t$. According to (Thm.\ 6.5.5, \cite{hohn}) and the Rank-Nullity Thm.\ such that $dim (N({\bf \Psi}_{\Lambda}))+ dim(N({\bf L}))\geq dim( N({\bf \Psi}_{\Lambda}{\bf L}))\geq N-|\Lambda|=|\Lambda^{\complement}|$, $N({\bf \Psi}_{\Lambda}{\bf L})$ has dimension at least $|\Lambda^{\complement}|-t+t= |\Lambda^{\complement}|$; if we assume that each set $\Lambda_k^{\complement}$ is non-empty the latter becomes an equality. 
\end{proof}

\begin{proof}[Proof of Lemma \ref{summ}]
We utilize extensions of the classical MPP relations:\\
$(i)$ We have ${\bf L}^k{\bf L}^{\dagger k}{\bf S}^T=\left({\bf I}_N-\frac{1}{N}{\bf J}_N\right)^k{\bf S}^T={\bf S}^T$, from which it follows ${\bf L}^k{\bf x}=\sum_{j\in E_S}w_j{\bf S}_j^T\in\mathbb{R}^{|V|}$. Further, ${\bf L}^{\dagger k}{\bf S}^T={\bf L}^{\dagger {k-1}}{\bf L}^{\dagger}{\bf S}^T={\bf L}^{\dagger {k-1}}{\bf S}^{\dagger}$ gives ${\bf x}=\sum_{j\in E_S}w_j{\bf L}^{\dagger k}{\bf S}_j^T=\sum_{j\in E_S} w_j{\bf L}^{\dagger k-1}{\bf S}_j^{\dagger}$. Here, the MPP-relation ${\bf S}^{\dagger}{\bf S}^{\dagger T}{\bf S}^T={\bf S}^{\dagger}$ trivially follows from SVD-decomposition with ${\bf S}={\bf U}{\bf \Sigma}{\bf V}^H$
such that \[{\bf S}^{\dagger}{\bf S}^{\dagger T}{\bf S}^T={\bf V}{\bf \Sigma}^{\dagger}({\bf \Sigma}^{\dagger})^T{\bf \Sigma}^{T}{\bf U}^H={\bf V}{\bf \Sigma}^{\dagger}{\bf U}^H={\bf S}^{\dagger}.\]
$(ii)$ Similarly, we have ${\bf S}{\bf L}^k({\bf L}^{\dagger k})_j={\bf S}({\bf e}_j-\frac{1}{N}{\bf 1}_N)={\bf S}_j$, such that ${\bf S}{\bf L}^k{\bf x}=\sum_{j\in V_S}w_j{\bf S}_j$.\\
$(iii)$ From ${\bf L}^{k}{\bf L}^{\dagger k-1}={\bf L}{\bf L}^{k-1}{\bf L}^{\dagger k-1}={\bf L}$, it follows ${\bf L}^{k}{\bf x}=\sum_{j\in V_S}w_j{\bf L}_j$. Further, we have ${\bf L}^{\dagger k-1}={\bf L}^{\dagger k-2}{\bf L}^{\dagger}={\bf L}^{\dagger k}{\bf L}$ such that ${\bf x}=\sum_{j\in V_S}w_j({\bf L}^{\dagger k-1})_j=\sum_{j\in V_S}w_j{\bf L}^{\dagger k}{\bf L}_j$. 
\end{proof}

\begin{proof}[Proof of Lemma \ref{diffind}]
Consider the general expression for row $i$ of $S^{\dagger}_C(i,j)$, from Rem.\ \ref{syu}:
\[S^{\dagger}_C(i,j)=\left \{
  \begin{aligned}
&\frac{-N-1}{2N}-\frac{j-i}{N},&& \ 0\leq j< i-1\\
&\frac{1-N}{2N},&& \  j=i-1\\
&\frac{N-1}{2N}-\frac{j-i}{N},&& \ i\leq j\leq N-1\\
  \end{aligned} \right.
\]
Accordingly, for $S^{\dagger}_C(i_1,j)-S^{\dagger}_C(i_2,j)$, with $i_1<i_2$, we obtain:
\[S^{\dagger}_C(i_1,j)-S^{\dagger}_C(i_2,j)=\left \{
  \begin{aligned}
&\frac{i_1-i_2}{N},&& \ 0\leq j< i_1-1\\
&\frac{i_1-i_2}{N},&& \ j=i_1-1\\
&1+\frac{i_1-i_2}{N},&& \ i_1\leq j< i_2-1\\
&1+\frac{i_1-i_2}{N},&& \ j=i_2-1\\
&\frac{i_1-i_2}{N},&& \ i_2\leq j\leq N-1\\
  \end{aligned} \right.
\]
We further discover that the sum of entries for the above difference is zero: $(i_1+N-i_2)\frac{i_1-i_2}{N}+(i_2-i_1)\left(1+\frac{i_1-i_2}{N}\right)=0$. 
\end{proof}

\begin{proof}[Proof of Lemma \ref{lemdecomp}]
The associated representer polynomial $l(z)$ of a circulant graph Laplacian of bandwidth $M<N/2$ with node degree $d= \sum_{i=1}^M2 d_i$ and weights $d_i=A_{j,(i+j)_N}$, can be decomposed as
\[l(z)=\sum_{i=1}^M2 d_i-\sum_{i=1}^M d_i (z^i+z^{-i})=\sum_{i=1}^M d_i (z^i-1)(z^{-i}-1)\]\[=(z-1)(z^{-1}-1)\left(\sum_{i=1}^M d_i (1+z+z^2+...+z^{i-1})(1+z^{-1}+z^{-2}+...+z^{-(i-1)})\right)\]\[=(-z+2-z^{-1})\sum_{i=1}^M d_i l_i(z) l_i(z^{-1})=(-z+2-z^{-1})P_{G_S}(z)\]
which reveals that any banded circulant graph Laplacian can be factored as the product of the simple cycle graph Laplacian (i.e. $2$ vanishing moments with $2$ zeros at $z=1$, as is known) and a circulant matrix ${\bf P}_{G_S}$ with representer polynomial $P_{G_S}(z)$. In particular, we prove that ${\bf P}_{G_S}$ is positive definite, and hence invertible, by first observing that the term  $l_i(z) l_i(z^{-1})$ gives rise to a Gramian matrix which is at least positive semi-definite by default, while under the assumption of nonnegative weights $d_i\geq 0$ and $d_1>0$ due to $s=1\in S$ (ensuring connectivity), we have $P_{G_S}(z)=d_1+\sum_{i=2}^M d_i l_i(z) l_i(z^{-1})$. Therefore ${\bf P}_{G_S}$ is the sum of positive (semi-)definite matrices and the positive definite matrix $d_1{\bf I}_N$, making it positive definite. Explicitly, it is given by
\[P_{G_S}(z)=d_1+\sum_{i=2}^Md_i \left(i+\sum_{k=1}^{i-1}(i-k)(z^k+z^{-k})\right)\]\[=\left(\sum_{i=1}^M i d_i\right)+\left(\sum_{i=2}^M (i-1) d_i\right)(z+z^{-1})+\left(\sum_{i=3}^M (i-2) d_i\right)(z^2+z^{-2})+...+d_M(z^{M-1}+z^{-(M-1)})\]
and it becomes evident that its coefficients directly mirror the structure of ${\bf y}$ in Sect.\ $4.1$. Hence, we always have ${\bf L}={\bf P}_{G_S}{\bf L}_C$, where ${\bf L}_C$ denotes the graph Laplacian of the unweighted simple cycle.
\end{proof}

\begin{proof}[Proof of Lemma \ref{lemcircl}]
We have ${\bf L}^{\dagger}={\bf P}_{G_S}^{-1}{\bf L}^{\dagger}_C$, which can be shown via simple eigendecomposition. Let the eigendecompositions and their pseudoinverses respectively be given by: ${\bf L}={\bf V}{\bf \Sigma}_2{\bf V}^H$, ${\bf L}_C={\bf V}{\bf \Sigma}_1{\bf V}^H$, ${\bf P}_{G_S}={\bf V}{\bf S}{\bf V}^H$ and ${\bf L}^{\dagger}={\bf V}{\bf \Sigma}^{\dagger}_2{\bf V}^H$, ${\bf L}^{\dagger}_C={\bf V}{\bf \Sigma}^{\dagger}_1{\bf V}^H$, ${\bf P}_{G_S}^{-1}={\bf V}{\bf S}^{-1}{\bf V}^H$, with common basis ${\bf V}$ the DFT-matrix, as a result of circularity. Then
\[{\bf L}={\bf P}_{G_S}{\bf L}_C={\bf V}{\bf S}^{}{\bf \Sigma}_1{\bf V}^H={\bf V}{\bf \Sigma}_2{\bf V}^H\] gives ${\bf \Sigma}_2={\bf S}^{}{\bf \Sigma}_1$ and so \[{\bf P}_{G_S}^{-1}{\bf L}^{\dagger}_C={\bf V}{\bf S}^{-1}{\bf \Sigma}_1^{\dagger}{\bf V}^H={\bf V}({\bf S}{\bf \Sigma}_1)^{\dagger}{\bf V}^H={\bf L}^{\dagger}.\]
According to Thm.\ $2$, in \cite{volkov} (see also Thm.\ $2.4$, \cite{demko}), the inverse of a cyclically banded positive matrix is `approximately' banded with entries that decay exponentially (in absolute value) away from the diagonal and corners of the matrix. 
Therefore the expression ${\bf P}_{G_S}^{-1}{\bf L}^{\dagger}_C$ suggests that for large $N$ and sufficiently small bandwidth $M$, relative to $N$, the rows and columns of ${\bf L}^{\dagger}$ are polynomials of the same form as ${\bf L}^{\dagger}_C$ subject to border effects (or perturbations) dependent on the `approximate' bandwidth of ${\bf P}_{G_S}^{-1}$. 

\end{proof}

\begin{proof}[Proof of Lemma \ref{lemlincirc}]
We have \[{\bf S}^{\dagger}={\bf L}^{\dagger}{\bf S}^T={\bf P}_{G_S}^{-1}{\bf L}_C^{\dagger}{\bf S}^T\]
such that $({\bf S}^{\dagger})_j={\bf P}_{G_S}^{-1}{\bf L}_C^{\dagger}({\bf S}^T)_j$. It was shown in Sect.\ $4.1.1.$ (see Eq. (\ref{eq:diff})), that taking unweighted differences between any two columns (rows) of ${\bf L}_C^{\dagger}$, which are piecewise quadratic polynomials, results in piecewise  linear polynomials. Thus each column of ${\bf L}^{\dagger}{\bf S}^T$ is a `perturbed' piecewise linear polynomial, following a degree reduction, but of different form since ${\bf S}^T$ has non-circular columns, and it is only in the simple cycle case, where ${\bf S}_C$ is circulant, that ${\bf S}^{\dagger}$ becomes linear along both dimensions.
\end{proof}

\begin{proof}[Proof of Lemma \ref{gan}]
Consider first the example of the simple cycle with ${\bf L}_C={\bf S}_C^T {\bf S}_C$, where ${\bf S}_C$ has first row $\lbrack 1\ -1\ 0\ ...\ 0\rbrack$. Both ${\bf S}_C^T$ and ${\bf S}_C$ are circulant with respective representer polynomials $s_C^T(z)=1-z^{-1}$ and $s_C(z)=1-z^{1}$, which have one vanishing moment each. Due to the equivalence between polynomial and circulant matrix multiplication, the representer polynomial of ${\bf L}_C$ has two vanishing moments and, by extension, that of ${\bf S }_C{\bf L}_C^k$ has $2k+1$ vanishing moments, i.e. ${\bf S }_C{\bf L}_C^k$ annihilates polynomials of up to order $2k$. \\
In the general circulant graph case, we have ${\bf S }{\bf L}^k={\bf S }{\bf L}_C^k{\bf P}_{G_S}^k$; with ${\bf S}={\bf S}^{\dagger T}{\bf L}$, this can be rewritten as
\[ {\bf S }{\bf L}^k={\bf S }{\bf L}_C^{k}{\bf P}_{G_S}^k={\bf S}^{\dagger T}{\bf L}{\bf L}_C^{k}{\bf P}_{G_S}^k={\bf S}^{\dagger T}{\bf S}^T_C{\bf P}_{G_S}({\bf S}_C{\bf L}_C^{k}){\bf P}_{G_S}^{k}.\]
The representer polynomial of factor ${\bf S}_C{\bf L}_C^{k}$ has $2k+1$ vanishing moments, while the factor ${\bf P}_{G_S}^k$, provided it is sufficiently banded, simply amplifies the discontinuities of the polynomial signal. The remaining factor ${\bf S}^{\dagger T}{\bf S}^T_C{\bf P}_{G_S}={\bf S}{\bf S}_C^{\dagger}=({\bf S}_C^{\dagger T}{\bf S}^T)^T$ takes differences between the piecewise-linear rows of ${\bf S}_C^{\dagger}$, so that its rows constitute piecewise constant signals, as per Lemma \ref{diffind}.\\
Given a piecewise polynomial ${\bf x}$ of order up to $2k$, the output ${\bf y}={\bf S}_C{\bf L}_C^{k}{\bf P}_{G_S}^{k}{\bf x}$ is sparse with non-zeros around the locations of its discontinuities, further amplified by ${\bf P}_{G_S}^{k}$. In addition, due to $N({\bf S}_C)=N({\bf S}_C^T)=z{\bf 1}_N,\ z\in \mathbb{R}$, ${\bf y}$ is orthogonal to ${\bf 1}_N$, i.e. its entries sum to zero. Hence, ${\bf y}$ is orthogonal to the rows of ${\bf S}{\bf S}_C^{\dagger}$, provided the support of ${\bf y}$ coincides with the locations of all-constant pieces of ${\bf S}{\bf S}_C^{\dagger}$, which is satisfied when the graph is banded, i.e. the spread around the discontinuities via ${\bf P}_{G_S}^{k}$ is sufficiently small. Therefore, the final output ${\bf S }{\bf L}^k{\bf x}={\bf S}{\bf S}_C^{\dagger}{\bf S}_C{\bf L}_C^{k}{\bf P}_{G_S}^{k}{\bf x}$ is sparse, proving that the operator ${\bf S}{\bf L}^k$ annihilates polynomials of order up to $2k$.

\end{proof}

\section{}
\label{app2}

\begin{proof}[Proof of Property \ref{prop4}]
Consider ${\bf L}_{\alpha}=d_{\alpha}{\bf I}_N-{\bf A}$ with $d_{\alpha}=\sum_{j=1}^M 2d_j\cos(\alpha j),\ \alpha\in\mathbb{C}$ on a circulant graph. We need to show when this operator is invertible and begin by assuming that its nullspace is non-empty. Let ${\bf z}$ be in the nullspace of ${\bf L}_{\alpha}$ and consider the representation ${\bf z}={\bf V}{\bf r}$, for DFT-matrix ${\bf V}$ such that ${\bf A}={\bf V}{\bf \Gamma}{\bf V}^H$ with eigenvalue matrix ${\bf \Gamma}$, and coefficient vector ${\bf r}\in\mathbb{C}^N$.\\
Then we have 
\[(d_{\alpha}{\bf I}_N-{\bf A}){\bf z}=d_{\alpha}{\bf V}{\bf r}-{\bf V}{\bf \Gamma}{\bf r}={\bf 0}_N.\]
By taking the $l_2$-norm of both sides and squaring the result, we obtain
\[||d_{\alpha}{\bf V}{\bf r}-{\bf V}{\bf \Gamma}{\bf r}||_2^2=0\]
\[\Leftrightarrow \sum_{i=0}^{N-1} |r(i)|^2 (d_{\alpha}^2-2 d_{\alpha}\gamma_i+\gamma_i^2)=\sum_{i=0}^{N-1} |r(i)|^2 (d_{\alpha}-\gamma_i)^2=0,\]
where $\gamma_i$ denotes the $i$-th eigenvalue in ${\bf \Gamma}$, as ordered by the DFT-matrix. Since $d_{\alpha}$ defines the structure of an eigenvalue of ${\bf A}$, with  $d_{\alpha}=\gamma_k$ when $\alpha=2\pi k/N$ for some $k\in\lbrack 0\enskip N-1\rbrack$, it follows from the above that $|r(i)|\geq 0$ for certain $i$, corresponding to the locations of $\gamma_k$ and its possible multiplicities, and $|r(i)|=0$ otherwise. Letting ${\bf z}={\bf V}_k\tilde{{\bf r}}$ denote the span of eigenvectors associated with eigenvalue $\gamma_k$ of multiplicity $m$, with coefficients $\tilde{{\bf r}}\in\mathbb{C}^m$, we obtain 
\[(d_{\alpha}{\bf I}_N-{\bf A}){\bf z}=d_{\alpha}{\bf V}_k\tilde{{\bf r}}-\gamma_k{\bf V}_k\tilde{{\bf r}}={\bf 0}_N.\]
 proving that ${\bf L}_{\alpha}$ has a non-empty nullspace in that case. Otherwise, for $\alpha\neq 2\pi k/N$, we must have ${\bf r}={\bf 0}_N$ and hence ${\bf L}_{\alpha}$ is invertible.
\end{proof}

\begin{proof}[Proof of Lemma \ref{inva}]
For circulant symmetric matrix ${\bf L}_{C, \alpha}$, with $\alpha\neq 2\pi k/N,\ k\in\mathbb{N}$, and ${\bf L}_{C, \alpha} {\bf L}_{C, \alpha}^{-1}={\bf I}_N$, we can express the entries ${L}^{-1}_{C, \alpha}(l):={L}^{-1}_{C, \alpha}(m,n)$ of ${\bf L}^{-1}_{C, \alpha}$ as a function depending only on the index distance $l=|n-m|$ due to symmetry and circularity. This gives the homogeneous linear recurrence relation \[2\cos(\alpha) {L}^{-1}_{C, \alpha}(l-1)-{L}^{-1}_{C, \alpha}(l)-{L}^{-1}_{C, \alpha}(l-2)=0.\] Via the substitution of ${L}^{-1}_{C, \alpha}(l)=r^l$, we obtain $-r^2+2\cos(\alpha) r-1=0$ with roots $r_{1/2}=e^{\pm i\alpha}$, which gives rise to the general homogeneous solution structure ${L}^{-1}_{C, \alpha}(l)=C_1 e^{i\alpha l}+ C_2 e^{-i\alpha l}$ for coefficients $C_1,\ C_2\in\mathbb{C}$.\\
Hence, for the unique solution of this system, we require two constraints. We begin by noting ${\bf L}_{C, \alpha}^{-1}{\bf L}_{C, \alpha}{\bf 1}_N=(2\cos(\alpha)-2){\bf L}_{C, \alpha}^{-1}{\bf 1}_N={\bf 1}_N$, following the circularity and, hence, equal row-sum of ${\bf L}_{C, \alpha}$, such that we obtain ${\bf L}_{C, \alpha}^{-1}{\bf 1}_N=\frac{1}{2\cos(\alpha)-2}{\bf 1}_N$, or $\sum_{l=0}^{N-1} {L}^{-1}_{C, \alpha}(l)=\frac{1}{2\cos(\alpha)-2}$.\\
Therefore, we need to satisfy the constraint $\sum_{l=0}^{N-1} (C_1 e^{i\alpha l} +C_2 e^{-i\alpha l})=\frac{1}{2\cos(\alpha)-2}$ and boundary condition \[2\cos(\alpha){L}^{-1}_{C, \alpha}(0)-2 {L}^{-1}_{C, \alpha}(1)=2\cos(a)(C_1+C_2)-2 (C_1 e^{i\alpha}+C_2 e^{-i \alpha})=1,\] where ${L}^{-1}_{C, \alpha}(-1)={L}^{-1}_{C, \alpha}(1)$, which gives $C_2=C_1+\frac{1}{e^{i\alpha}-e^{-i\alpha}}$.\\
For the former, utilizing the exponential sum formula $\sum_{l=0}^{N-1} e^{i\alpha l}=\frac{1-e^{i\alpha N}}{1-e^{i\alpha}}$, we obtain \[\frac{C_1 (1-e^{i\alpha N})+C_2 (e^{-i\alpha (N-1)}-e^{i\alpha})}{1-e^{i\alpha}}=\frac{1}{2\cos(\alpha)-2}\]

\[\Leftrightarrow C_1 (e^{-i\alpha}-1-e^{i\alpha (N-1)}+e^{i\alpha N})+C_2(e^{-i\alpha N}-1-e^{-i\alpha (N-1)}+e^{i\alpha})=1.\]
Substituting $C_2=C_1+\frac{1}{e^{i\alpha}-e^{-i\alpha}}$ from the boundary condition, and, following simplifications, we obtain 
\[C_1 (e^{-i\alpha}-1-e^{i\alpha (N-1)}+e^{i\alpha N})+\left(C_1+\frac{1}{e^{i\alpha}-e^{-i\alpha}}\right)(e^{-i\alpha N}-1-e^{-i\alpha (N-1)}+e^{i\alpha})=1\]
\[\Leftrightarrow C_1=\frac{1}{(-e^{-i\alpha}+e^{i\alpha})(-1+e^{i\alpha N})}\]
and, equivalently 
\[C_2=\frac{1}{e^{i\alpha}-e^{-i\alpha}}+\frac{e^{i \alpha}}{(-1+e^{2i\alpha})(-1+e^{i\alpha N})}=\frac{1}{(e^{-i\alpha}-e^{i\alpha})(-1+e^{-i\alpha N})}.\]
Further inspection confirms that the coefficients form complex conjugates $C_1=\bar{C_2}$. 
Note that if $\alpha=2\pi k/N$, $k\in\mathbb{N}$, there would be a pole at $(-1+e^{i\alpha N})$, however, this case is not applicable here.
Eventually, we obtain the desired result
\[{L}_{C, \alpha}^{-1}(m,n)=\scalebox{1}{$\frac{1}{(-e^{-i\alpha}+e^{i\alpha})(-1+e^{i\alpha N})} e^{i\alpha |n-m|}+\frac{1}{(e^{-i\alpha}-e^{i\alpha})(-1+e^{-i\alpha N})} e^{-i\alpha |n-m|}$},\] \[0\leq m,n\leq N-1.\]
\end{proof}

\begin{proof}[Proof of Lemma \ref{mppa}]
For circulant symmetric matrix ${\bf L}_{C, \alpha}$, with $\alpha= 2\pi k/N,\ k\in\mathbb{N}$, we now have the system ${\bf L}_{C, \alpha}{\bf L}_{C, \alpha}^{\dagger}={\bf L}_{C, \alpha}^{\dagger}{\bf L}_{C, \alpha}={\bf I}_N-\frac{1}{N}{\bf E}_{\alpha}$, where $E_{\alpha}(m,n)=2\cos(\alpha |n-m|)$.\\
As a result, we obtain the non-homogeneous recurrence relation 
\begin{equation}\label{eq:rec}-{L}_{C, \alpha}^{\dagger}(l)+2\cos(\alpha) {L}_{C, \alpha}^{\dagger}(l-1)-{L}_{C, \alpha}^{\dagger}(l-2)=-\frac{2\cos(\alpha (l-1))}{N}\end{equation} with ${L}_{C, \alpha}^{\dagger}(l):={L}_{C, \alpha}^{\dagger}(m,n)$ and $l=|n-m|$, following the same argument as in the proof of Lemma \ref{inva}. In particular, we know from the previous proof that the homogeneous solution to the system is $({L}_{C, \alpha}^{\dagger})_H(l)=C_1 e^{i\alpha l}+ C_2 e^{-i\alpha l}$ for coefficients $C_1, C_2\in\mathbb{C}$, and, as a result of repeated roots, the particular solution must be of the form $({L}_{C, \alpha}^{\dagger})_P(l)=l(C_3 e^{i\alpha l}+ C_4 e^{-i\alpha l})$ for coefficients $C_3, C_4\in\mathbb{C}$. 
Via substitution of $({L}_{C, \alpha}^{\dagger})_P(l)$ into Eq. (\ref{eq:rec}), 
we obtain

\[({L}_{C, \alpha}^{\dagger})_P(l)=\frac{l}{N}\left(\frac{1}{e^{i\alpha}-e^{-i\alpha}}e^{i\alpha l}+\frac{1}{-e^{i\alpha}+e^{-i\alpha}}e^{-i\alpha l}\right)\]
with complex conjugate coefficients $C_3=\bar{C_4}$.\\ 
At last we impose the boundary condition $2\cos(\alpha) {L}_{C, \alpha}^{\dagger}(0)-2 {L}_{C, \alpha}^{\dagger}(1)=\frac{N-2}{N}$ and constraint
$\sum_{l=0}^{N-1} 2\cos(\alpha l){L}_{C, \alpha}^{\dagger}(l) =0$, following from ${\bf L}^{\dagger}_{C, \alpha} {\bf E}_{\alpha}={\bf 0}$, to describe the homogeneous solution. Here, we use the power sum formula \[\sum_{k=0}^n kx^k=\frac{x-(n+1)x^{n+1}+nx^{n+2}}{(x-1)^2}\] which gives \[\sum_{n=0}^{N-1} ne^{i 2 \alpha n}=\frac{e^{i 2 \alpha}+(N-1) e^{i 2 \alpha (N+1)}-N e^{i 2 \alpha N}}{(e^{i 2 \alpha}-1)^2}\] such that
\begin{equation*}
\begin{split}
\sum_{l=0}^{N-1} 2\cos(\alpha l) ({L}_{C, \alpha}^{\dagger})_P(l) &=\frac{-e^{-i\alpha (2N -1)}e^{2 i\alpha}(N-1)}{N (e^{2i\alpha}-1)^3}\\
&+\frac{e^{-i\alpha (2N -1)}(e^{2 i\alpha(1+2N)}(N-1)+N(e^{4i\alpha}-e^{4 i\alpha N}))}{N (e^{2i\alpha}-1)^3}\end{split}\end{equation*}
and 
\begin{equation*}
\begin{split}\sum_{l=0}^{N-1} 2\cos(\alpha l)({L}_{C, \alpha}^{\dagger})_H(l)&=C_1 \frac{(N+1)-N e^{2i\alpha}-e^{2i\alpha N}}{(1-e^{2 i\alpha})}\\
&+C_2 \frac{N+e^{2i\alpha}(e^{-2i\alpha N}-(N+1))}{(1-e^{2 i\alpha})}.\end{split}\end{equation*}
From $2\cos(\alpha){L}_{C, \alpha}^{\dagger}(0)-2 {L}_{C, \alpha}^{\dagger}(1)+2/N-1=0$, we again have $C_2=C_1+\frac{1}{e^{ i\alpha}-e^{-i\alpha}}$ and substitute this into $\sum_{l=0}^{N-1} 2\cos(\alpha l){L}_{C, \alpha}^{\dagger}(l) =0$, to obtain
\[C_1=-\frac{e^{i\alpha (3-2N)}-e^{i\alpha(3+2N)}+e^{i\alpha}((-1+e^{i\alpha 2}) (e^{i\alpha 2}+e^{i\alpha 2 N}) N+(-1+e^{i\alpha 2})^2 N^2)}{(-1+e^{i\alpha 2})^2 (-(-1+e^{-i\alpha 2 N})(e^{i\alpha 2 N}+e^{i\alpha 2})N+2 N^2 (-1+e^{i\alpha 2}))}\]
\[C_2=\frac{e^{i\alpha (3+4N)}-e^{i\alpha 3}-e^{i\alpha}N(-1+e^{i\alpha 2}) (e^{i\alpha 2}+e^{i\alpha 2 N})+e^{i\alpha (1+2 N)}(-1+e^{i\alpha 2})^2 N^2}{e^{2 i\alpha N}(-1+e^{i\alpha 2})^2 (-(-1+e^{-i\alpha 2 N})(e^{i\alpha 2 N}+e^{i\alpha 2})N+2 N^2 (-1+e^{i\alpha 2}))}\]
\[=-\frac{e^{-i\alpha 3}(e^{2 i\alpha N}-e^{-2 i\alpha N})+e^{-i\alpha}((-1+e^{-i\alpha 2}) (e^{-i\alpha 2}+e^{-i\alpha 2 N}) N+(-1+e^{-i\alpha 2})^2 N^2)}{(-1+e^{-i\alpha 2})^2 (-(-1+e^{i\alpha 2 N})(e^{-i\alpha 2 N}+e^{-i\alpha 2})N+2 N^2 (-1+e^{-i\alpha 2}))}\]
so that $C_2=\bar{C_1}$.
We can simplify expressions further due to $\alpha=2\pi k/N$, which e.g. entails $e^{i\alpha N k}=1$ for any $k$, so that we obtain
\[C_1=-\frac{e^{i\alpha}((e^{i\alpha 2}+1)+(-1+e^{i\alpha 2}) N)}{2 N(-1+e^{i\alpha 2})^2}\]
\[C_2=\frac{e^{i\alpha}(-(1+e^{i\alpha 2})+ (-1+e^{i\alpha 2}) N)}{2 N(-1+e^{i\alpha 2})^2}=-\frac{e^{-i\alpha}((1+e^{-2 i\alpha})+(-1+e^{-2 i\alpha})N)}{2N (-1+e^{-2 i\alpha})^2}\]
\\
Consequently, the pseudoinverse ${\bf L}_{C, \alpha}^{\dagger}$ for $\alpha=2\pi k/N,\enskip k\in\mathbb{N}$ and $\alpha\neq 0, \pi k$, of ${\bf L}_{C, \alpha}$ on the simple cycle has entries
\[{L}_{C, \alpha}^{\dagger}(m,n)=\frac{e^{i\alpha}}{2N}\left(\frac{2|n-m|(-1+e^{2 i\alpha})+(N-1)-e^{2 i\alpha}(N+1)}{(-1+e^{2 i \alpha})^2}\right)e^{i\alpha |n-m|}\]
\[+\frac{e^{-i\alpha}}{2N}\left(\frac{2|n-m|(-1+e^{-2 i\alpha})+(N-1)-e^{-2 i\alpha}(N+1)}{(-1+e^{-2 i \alpha})^2}\right)e^{-i\alpha |n-m|},\enskip 0\leq m,n\leq N-1.\]
Note that for certain $\alpha$ such as $\alpha=0, k \pi$ the formula does not apply due to the pole with $-1+e^{2 i\alpha}=-1+e^{i 2 \pi k}=0$ and an alternative construction is required. \\For the case $\alpha=\pi$ at even $N$, we need to modify the boundary condition and constraint as follows: we have the problem ${\bf L}_{C, \pi}{\bf L}^{\dagger}_{C, \pi}={\bf I}_N-\frac{1}{N}\bar{{\bf J}}_N$, where $\bar{{\bf J}}_N$ is the circulant matrix with first row $\lbrack 1 \ -1 \ 1 \ -1 \ \dots \rbrack$ with $\cos(\alpha)=-1$, and the constraint ${\bf L}^{\dagger}_{C, \pi}\bar{{\bf J}}_N={\bf 0}$, which translates into 
\[-{L}^{\dagger}_{C, \pi}(l+1)-2 {L}^{\dagger}_{C, \pi}(l)-{L}^{\dagger}_{C, \pi}(l-1)=(-1)^{l+1}/N\] with $-2 {L}^{\dagger}_{C, \pi}(0)-2{L}^{\dagger}_{C, \pi}(1)=1-1/N$ and $\sum_{l=0}^{N-1} {L}^{\dagger}_{C, \pi}(l) (-1)^l=0$.
We similarly infer homogeneous solution $({L}^{\dagger}_{C, \pi})_H(l)=(-1)^l C_1+(-1)^l l C_2$ and particular solution $({L}^{\dagger}_{C, \pi})_P(l)=\frac{(-1)^{l+1} l^2}{2N}$, giving $C_1=\frac{1-N^2}{12N}$ and $C_2=\frac{1}{2}$, so that
\[{L}_{C, \pi}^{\dagger}(m,n)=(-1)^{|n-m|+1}\left(\frac{(n-m)^2}{2N}-\frac{1}{2}|n-m|+\frac{N^2-1}{12 N}\right)=(-1)^{|n-m|+1}{L}_{C}^{\dagger}(m,n),\]\[\enskip 0\leq m,n\leq N-1.\]
\end{proof}

\begin{proof}[Proof of Lemma \ref{newdecomp}]
Consider ${\bf L}_{\alpha}=d_{\alpha}{\bf I}_N-{\bf A}$ on a banded circulant graph of bandwidth $M<N/2$ with $d_{\alpha}=\sum_{j=1}^M 2 d_j \cos(\alpha j)$. In Lemma $3.2$, \cite{splinesw}, it was shown that the representer polynomial $l_{\alpha}(z)$ of ${\bf L}_{\alpha}$ has two vanishing exponential moments. We further extend this to yield the decomposition
\[l_{\alpha}(z)=\sum_{j=1}^M d_j (1-e^{i\alpha j}z^j)(1-e^{-i\alpha j}z^j)(-z^{-j})=(1-e^{i \alpha}z)(1-e^{-i \alpha}z)\sum_{j=1}^M d_j p_j(z) q_j(z)(-z^{-j})\]\[=(-z^{-1}+2\cos(\alpha)-z)\sum_{j=1}^M d_j p_j(z)q_j(z)z^{-(j-1)}
=l_{C,\alpha}(z) P_{\alpha}(z), \]
where $p_j(z)=(1+e^{i\alpha}z+e^{2i\alpha}z^2+...+e^{i\alpha (j-1)}z^{j-1})$, and $q_j(z)=(1+e^{-i\alpha}z+e^{-2i\alpha}z^2+...+e^{-i\alpha (j-1)}z^{j-1})$. Here, $l_{C,\alpha}(z)$ and $P_{\alpha}(z)$ are the representer polynomials of ${\bf L}_{C, \alpha}$ for the unweighted simple cycle and the circulant matrix ${\bf P}_{\alpha}$ of bandwidth $M-1$ respectively, such that we have ${\bf L}_{\alpha}={\bf L}_{C, \alpha}{\bf P}_{\alpha}$. It becomes evident that for $\alpha\neq 0$, ${\bf P}_{\alpha}$ is not necessarily positive definite, and we further proceed to derive its entries. \\
Consider the arising terms \[p_j(z)q_j(z)z^{-(j-1)}=z^{-(j-1)} \sum_{t=0}^{2(j-1)}r_t z^t,\ \text{with} \ r_t=\sum_{m+n=t} e^{i\alpha m} e^{-i\alpha n}.\]
It can be easily checked that $p_j(z)q_j(z)=p_j(z^{-1})q_j(z^{-1})z^{2(j-1)}$ is a palindromic polynomial, which gives $p_j(z)q_j(z)z^{-(j-1)}=r_{j-1}+\sum_{t=1}^{j-1} r_{j-1-t} (z^t+z^{-t})$ with $r_t=r_{2(j-1)-t}$, where for even $t$ \[r_t=1+\sum_{k=0}^{t/2-1} e^{i\alpha(2k-t)}+e^{-i\alpha(2k-t)}=1+\sum_{k=0}^{t/2-1} 2\cos(\alpha (2k-t))\]\[=1+2\cos(2\alpha)+2\cos(4\alpha)+...+2\cos(t\alpha)\]
and for odd $t$
\[r_t=\sum_{k=0}^{(t+1)/2-1}  2\cos(\alpha (2k-t))=2\cos(\alpha)+2\cos(3\alpha)+...+2\cos(t\alpha).\]

\end{proof}

\begin{proof}[Proof of Cor. \ref{mppinva}]
$(i)$ From Property \ref{prop4}, we know that if $\alpha\neq 2\pi k/N,\ k\in\mathbb{N}$, both ${\bf L}_{\alpha}$ and ${\bf L}_{C, \alpha}$ must be invertible. By extension, and given the decomposition ${\bf L}_{\alpha}={\bf L}_{C, \alpha}{\bf P}_{\alpha}$, it trivially follows that ${\bf P}_{\alpha}$ must be invertible and ${\bf L}_{\alpha}^{-1}={\bf L}_{C, \alpha}^{-1}{\bf P}_{\alpha}^{-1}$.\\
$(ii)$ For $\alpha=2 \pi k/N$ and $\alpha\neq 0, k\pi, \ k\in\mathbb{N}$ and suitable $G_S$ such that  ${\bf P}_{\alpha}$ is positive definite, the decomposition ${\bf L}_{\alpha}^{\dagger}={\bf L}_{C, \alpha}^{\dagger}{\bf P}_{\alpha}^{-1}$ holds, following the argument of eigendecomposition (as shown in the proof of Lemma \ref{lemcircl}). \\
Provided that ${\bf P}_{\alpha}$ is positive definite such that ${\bf P}_{\alpha}^{-1}$ becomes approximately banded with entries which decay exponentially in absolute value, for suitable $\alpha$ and $G_S$ (see \cite{volkov}), ${\bf P}_{\alpha}^{-1}$ can be interpreted as invoking a perturbation on ${\bf L}_{C, \alpha}^{\dagger}$, which constitute linear complex exponential polynomial functions.
\end{proof}

\begin{proof}[Proof of Prop. \ref{newnullsp1}]
$(i)$ Since ${\bf L}_{\alpha}$ is invertible for $\alpha\neq 2\pi k/N,\ k\in\mathbb{N}$, it trivially follows that $N({\bf \Psi}_{\Lambda}{\bf L}_{\alpha})={\bf L}_{\alpha}^{-1} N({\bf \Psi}_{\Lambda})$ from ${\bf \Psi}_{\Lambda}{\bf L}_{\alpha}{\bf L}_{\alpha}^{-1} N({\bf \Psi}_{\Lambda})={\bf 0}_{\Lambda}$, with $N({\bf \Psi}_{\Lambda})={\bf \Psi}_{\Lambda^{\complement}}^T$ as a possible basis.\\
$(ii)$ For $\alpha=2\pi k/N,\ k\in\mathbb{N}$, excluding $\alpha=0, k \pi$, we have by design $N({\bf L}_{C, \alpha})=z_1 e^{i\alpha {\bf t}} +z_2 e^{- i\alpha {\bf t}},\ z_1,z_2\in\mathbb{C}$, of dimension $2$. With ${\bf P}_{\alpha}$ positive definite, this entails $N({\bf L}_{\alpha})=N({\bf L}_{C, \alpha})$. In addition, in line with the proof of Prop.\ \ref{prop1} for $\alpha=0$, we identify the solution set of ${\bf L}_{\alpha}{\bf u}=N({\bf \Psi}_{\Lambda}){\bf W}_{\alpha}{\bf c}$ subject to some constraint ${\bf W}_{\alpha}\in\mathbb{R}^{|\Lambda^{\complement}|\times k}$ with subspace dimension $k$, which needs to satisfy the F.A. $N({\bf \Psi}_{\Lambda}){\bf W}_{\alpha}\perp N({\bf L}_{\alpha})$, to be of the form
${\bf L}_{\alpha}^{\dagger}{\bf \Psi}_{\Lambda^{\complement}}^T{\bf w}$, with $N({\bf \Psi}_{\Lambda}):={\bf \Psi}_{\Lambda^{\complement}}^T$ and ${\bf w}={\bf W}_{\alpha}{\bf c}$ for suitable ${\bf c}$,
as follows
\[{\bf \Psi}_{\Lambda}{\bf L}_{\alpha}{\bf L}_{\alpha}^{\dagger}{\bf \Psi}_{\Lambda^{\complement}}^T{\bf w}={\bf \Psi}_{\Lambda}\left({\bf I}_N-\frac{1}{N}{\bf E}_{\alpha}\right){\bf \Psi}_{\Lambda^{\complement}}^T{\bf w}={\bf 0}_{\Lambda},\]
where we have used ${\bf L}_{\alpha}{\bf L}_{\alpha}^{\dagger}=({\bf I}_N-\frac{1}{N}{\bf E}_{\alpha})$, with ${\bf E}_{\alpha}=\sum_{\lambda_j=0}{\bf u}_j {\bf u}^H_j$ and ${\bf u}_1=e^{ i\alpha {\bf t}}$, ${\bf u}_2=e^{-i\alpha {\bf t}}$. The latter is a symmetric circulant matrix with diagonals of the form $\frac{-2}{N}\cos(\alpha k)$, where $\cos(\alpha k)=\cos(\alpha (N-k))$ for the given $\alpha$, and main diagonal $\frac{N-2}{N}$. Hence, we need to determine the nullspace of the matrix ${\bf \Psi}_{\Lambda}({\bf I}_N-\frac{1}{N}{\bf E}_{\alpha}){\bf \Psi}_{\Lambda^{\complement}}^T$.\\
\\
We discover that ${\bf w}$ is in $N({\bf \Psi}_{\Lambda}{\bf E}_{\alpha}{\bf \Psi}_{\Lambda^{\complement}}^T)$, i.e. it must be orthogonal to partially supported complex exponential signals, or, $e^{\pm i\alpha {\bf t}}\perp {\bf \Psi}_{\Lambda^{\complement}}^T{\bf w}$. Here, ${\bf E}_{\alpha}$ has rank $k=2$, as does ${\bf \Psi}_{\Lambda}{\bf E}_{\alpha}{\bf \Psi}_{\Lambda^{\complement}}^T$ for $|\Lambda^{\complement}|>1$, up to special cases; according to the rank-nullity Thm.\ \cite{horn}, this implies that $N({\bf \Psi}_{\Lambda}{\bf E}_{\alpha}{\bf \Psi}_{\Lambda^{\complement}}^T)$ is of dimension $|\Lambda^{\complement}|-2$, generating a basis of size $|\Lambda^{\complement}|\times|\Lambda^{\complement}|-2$. Since $N({\bf L}_{\alpha})$ has dimension $2$, we require $|\Lambda^{\complement}|\geq 3$ in order for $N({\bf \Psi}_{\Lambda}{\bf L}_{\alpha})$ to be full-rank and for its subspace ${\bf L}_{\alpha}^{\dagger}{\bf \Psi}_{\Lambda^{\complement}}^T{\bf W}_{\alpha}$ to be non-empty; specifically we have ${\bf W}_{\alpha}\in\mathbb{C}^{|\Lambda^{\complement}|\times|\Lambda^{\complement}|-2}$ which must have at least three entries $|\Lambda^{\complement}|\geq 3$ to yield a solution. Else, we have $N({\bf \Psi}_{\Lambda}{\bf L}_{\alpha})=N({\bf L}_{\alpha})$.\\
In a special case for $|\Lambda^{\complement}|=2$, it can be shown that the system resulting from $e^{\pm i\alpha {\bf t}}\perp {\bf \Psi}_{\Lambda^{\complement}}^T{\bf w}$
\[c_1 e^{\pm i\alpha t_1}+c_2 e^{\pm i\alpha t_2}=0,\ \text{for unknown}\ c_1,c_2\in \mathbb{C},\ t_1, t_2 \in \Lambda^{\complement}\]
has the only non-trivial solution of the form $t_2=t_1+\frac{N}{2}$ for $c_1=c_2(-1)^{k+1}$. Hence, for $\Lambda^{\complement}=\{m,(m+\frac{N}{2}), \ m\in V\}$, we have that $N({\bf L}_{\alpha})$ has dimension $|\Lambda^{\complement}|+1=3$ and the minimum support size to annihilate ${\bf E}_{\alpha}$ is $2$.

\end{proof}





\bibliographystyle{elsarticle-num} 
\bibliography{Bib2018Acha}

\begin{thebibliography}{10}
\expandafter\ifx\csname url\endcsname\relax
  \def\url#1{\texttt{#1}}\fi
\expandafter\ifx\csname urlprefix\endcsname\relax\def\urlprefix{URL }\fi
\expandafter\ifx\csname href\endcsname\relax
  \def\href#1#2{#2} \def\path#1{#1}\fi

\bibitem{elad}
M.~Elad, P.~Milanfar, R.~Rubinstein,
  \href{http://stacks.iop.org/0266-5611/23/i=3/a=007}{Analysis versus synthesis
  in signal priors}, Inverse Problems 23~(3) (2007) 947.
\newline\urlprefix\url{http://stacks.iop.org/0266-5611/23/i=3/a=007}

\bibitem{cos}
S.~Nam, M.~Davies, M.~Elad, R.~Gribonval, The cosparse analysis model and
  algorithms, Applied and Computational Harmonic Analysis 34~(1) (2013) 30--56.
\newblock \href {https://doi.org/10.1016/j.acha.2012.03.006}
  {\path{doi:10.1016/j.acha.2012.03.006}}.

\bibitem{uos}
T.~Blumensath, M.~Davies, Sampling theorems for signals from the union of
  finite-dimensional linear subspaces, IEEE Transactions on Information Theory
  55~(4) (2009) 1872--1882.
\newblock \href {https://doi.org/10.1109/TIT.2009.2013003}
  {\path{doi:10.1109/TIT.2009.2013003}}.

\bibitem{fred}
A.~G. Ramm, \href{http://www.jstor.org/stable/2695558}{A simple proof of the
  fredholm alternative and a characterization of the fredholm operators}, The
  American Mathematical Monthly 108~(9) (2001) 855--860.
\newline\urlprefix\url{http://www.jstor.org/stable/2695558}

\bibitem{splinesw}
M.~S. Kotzagiannidis, P.~L. Dragotti, Splines and wavelets on circulant graphs,
  Applied and Computational Harmonic Analysis (2017), in press,
  https://doi.org/10.1016/j.acha.2017.10.002, available on arXiv:
  arXiv:1603.04917.

\bibitem{horn}
R.~A. Horn, C.~R. Johnson, Matrix {A}nalysis, 2nd Edition, Cambridge University
  Press, New York, NY, USA, 2012.

\bibitem{shu}
D.~I. Shuman, S.~K. Narang, P.~Frossard, A.~Ortega, P.~Vandergheynst, {The
  Emerging Field of Signal Processing on Graphs: Extending High-Dimensional
  Data Analysis to Networks and Other Irregular Domains}, IEEE Signal Process.
  Mag. 30~(3) (2013) 83--98.

\bibitem{acha2}
M.~S. Kotzagiannidis, P.~L. Dragotti, Sampling and {R}econstruction of {S}parse
  {S}ignals on {C}irculant {G}raphs - {A}n {I}ntroduction to {G}raph-{FRI},
  Appl. Comput. Harmon. Anal. (2017), in press,
  http://dx.doi.org/10.1016/j.acha.2017.10.003, available on arXiv:
  arXiv:1606.08085.

\bibitem{unser}
M.~Unser, J.~Fageot, H.~Gupta, Representer theorems for sparsity-promoting
  l1-regularization, IEEE Transactions on Information Theory 62 (2016)
  5167--5180.

\bibitem{flinth}
A.~Flinth, P.~Weiss,
  \href{https://hal.archives-ouvertes.fr/hal-01572196}{{Exact solutions of
  infinite dimensional total-variation regularized problems}}, working paper or
  preprint (Aug. 2017).
\newline\urlprefix\url{https://hal.archives-ouvertes.fr/hal-01572196}

\bibitem{Coifman}
R.~Coifman, M.~Maggioni, Diffusion wavelets, Applied and Computational Harmonic
  Analysis 21~(1) (2006) 53--94.

\bibitem{spectral}
D.~K. Hammond, P.~Vandergheynst, R.~Gribonval, Wavelets on graphs via spectral
  graph theory, Applied and Computational Harmonic Analysis 30~(2) (2011)
  129--150.
\newblock \href {https://doi.org/http://dx.doi.org/10.1016/j.acha.2010.04.005}
  {\path{doi:http://dx.doi.org/10.1016/j.acha.2010.04.005}}.

\bibitem{ortega3}
S.~K. Narang, A.~Ortega, Compact {S}upport {B}iorthogonal {W}avelet
  {F}ilterbanks for {A}rbitrary {U}ndirected {G}raphs, Signal Processing, IEEE
  Transactions on 61~(19) (2013) 4673--4685.
\newblock \href {https://doi.org/10.1109/TSP.2013.2273197}
  {\path{doi:10.1109/TSP.2013.2273197}}.

\bibitem{dict1}
D.~Thanou, D.~I. Shuman, P.~Frossard,
  \href{http://dx.doi.org/10.1109/GlobalSIP.2013.6736921}{Parametric dictionary
  learning for graph signals}, in: {IEEE} Global Conference on Signal and
  Information Processing, GlobalSIP 2013, Austin, TX, USA, December 3-5, 2013,
  2013, pp. 487--490.
\newblock \href {https://doi.org/10.1109/GlobalSIP.2013.6736921}
  {\path{doi:10.1109/GlobalSIP.2013.6736921}}.
\newline\urlprefix\url{http://dx.doi.org/10.1109/GlobalSIP.2013.6736921}

\bibitem{kov}
S.~Chen, A.~Singh, J.~Kovacevic,
  \href{http://arxiv.org/abs/1803.02944}{Multiresolution representations for
  piecewise-smooth signals on graphs}, CoRR abs/1803.02944.
\newblock \href {http://arxiv.org/abs/1803.02944} {\path{arXiv:1803.02944}}.
\newline\urlprefix\url{http://arxiv.org/abs/1803.02944}

\bibitem{tools}
S.~Chen, R.~Varma, A.~Singh, J.~Kovacevic,
  \href{http://arxiv.org/abs/1512.05406}{Signal representations on graphs:
  Tools and applications}, CoRR abs/1512.05406.
\newblock \href {http://arxiv.org/abs/1512.05406} {\path{arXiv:1512.05406}}.
\newline\urlprefix\url{http://arxiv.org/abs/1512.05406}

\bibitem{smola}
Y.-X. Wang, J.~Sharpnack, A.~J. Smola, R.~J. Tibshirani,
  \href{http://jmlr.org/papers/v17/15-147.html}{Trend filtering on graphs},
  Journal of Machine Learning Research 17~(105) (2016) 1--41.
\newline\urlprefix\url{http://jmlr.org/papers/v17/15-147.html}

\bibitem{Pesenson}
I.~Pesenson, Variational {S}plines and {P}aley--{W}iener {S}paces on
  {C}ombinatorial {G}raphs, Constructive Approximation 29~(1) (2008) 1--21.
\newblock \href {https://doi.org/10.1007/s00365-007-9004-9}
  {\path{doi:10.1007/s00365-007-9004-9}}.

\bibitem{global}
M.~S. Kotzagiannidis, M.~E. Davies, Analysis vs synthesis - an investigation of
  (co)sparse signal models on graphs, in: 2018 {IEEE} Global Conference on
  Signal and Information Processing (GlobalSIP), 2018, to appear.

\bibitem{circcon}
F.~Boesch, R.~Tindell, Circulants and their connectivities, Journal of Graph
  Theory 8~(4)  487--499.
\newblock \href
  {http://arxiv.org/abs/https://onlinelibrary.wiley.com/doi/pdf/10.1002/jgt.3190080406}
  {\path{arXiv:https://onlinelibrary.wiley.com/doi/pdf/10.1002/jgt.3190080406}},
  \href {https://doi.org/10.1002/jgt.3190080406}
  {\path{doi:10.1002/jgt.3190080406}}.

\bibitem{circul}
D.~Geller, I.~Kra, S.~Popescu, S.~Simanca, On circulant matrices, Preprint,
  Stony Brook University.

\bibitem{laplaceeigen}
P.~F. Stadler, Laplacian {E}igenvectors of {G}raphs: {Perron-Frobenius} and
  {Faber-Krahn} Type Theorems, Springer, Berlin; New York, 2007.
\newblock \href {https://doi.org/10.1007/978-3-540-73510-6}
  {\path{doi:10.1007/978-3-540-73510-6}}.

\bibitem{mpp}
J.~C.~A. Barata, M.~S. Hussein, The {M}oore--{P}enrose pseudoinverse: A
  tutorial review of the theory, Brazilian Journal of Physics 42~(1) (2012)
  146--165.
\newblock \href {https://doi.org/10.1007/s13538-011-0052-z}
  {\path{doi:10.1007/s13538-011-0052-z}}.

\bibitem{green1}
F.~Chung, S.-T. Yau,
  \href{http://www.sciencedirect.com/science/article/pii/S0097316500930942}{Discrete
  {G}reen's functions}, Journal of Combinatorial Theory, Series A 91~(1--2)
  (2000) 191 -- 214.
\newblock \href {https://doi.org/http://dx.doi.org/10.1006/jcta.2000.3094}
  {\path{doi:http://dx.doi.org/10.1006/jcta.2000.3094}}.
\newline\urlprefix\url{http://www.sciencedirect.com/science/article/pii/S0097316500930942}

\bibitem{plonka2}
S.~Hoffmann, G.~Plonka, J.~Weickert, Discrete green's functions for harmonic
  and biharmonic inpainting with sparse atoms, in: X.-C. Tai, E.~Bae, T.~F.
  Chan, M.~Lysaker (Eds.), Energy Minimization Methods in Computer Vision and
  Pattern Recognition, Springer International Publishing, Cham, 2015, pp.
  169--182.

\bibitem{splinesn}
M.~Unser, T.~Blu, Self-similarity: Part i-splines and operators, IEEE
  Transactions on Signal Processing 55~(4) (2007) 1352--1363.
\newblock \href {https://doi.org/10.1109/TSP.2006.890843}
  {\path{doi:10.1109/TSP.2006.890843}}.

\bibitem{chu}
\href{https://epubs.siam.org/doi/abs/10.1137/1.9781611970470.ch7}{7. connection
  coefficients}, in: Orthogonal Polynomials and Special Functions, pp. 57--69.
\newblock \href
  {http://arxiv.org/abs/https://epubs.siam.org/doi/pdf/10.1137/1.9781611970470.ch7}
  {\path{arXiv:https://epubs.siam.org/doi/pdf/10.1137/1.9781611970470.ch7}},
  \href {https://doi.org/10.1137/1.9781611970470.ch7}
  {\path{doi:10.1137/1.9781611970470.ch7}}.
\newline\urlprefix\url{https://epubs.siam.org/doi/abs/10.1137/1.9781611970470.ch7}

\bibitem{green2}
R.~Ellis, Discrete {G}reen's functions for products of regular graphs, arXiv
  preprint math/0309080.

\bibitem{volkov}
Y.~S. Volkov, \href{https://doi.org/10.1134/S1995423910030018}{Inverses of
  cyclic band matrices and the convergence of interpolation processes for
  derivatives of periodic interpolation splines}, Numerical Analysis and
  Applications 3~(3) (2010) 199--207.
\newblock \href {https://doi.org/10.1134/S1995423910030018}
  {\path{doi:10.1134/S1995423910030018}}.
\newline\urlprefix\url{https://doi.org/10.1134/S1995423910030018}

\bibitem{plonka}
G.~Plonka, S.~Hoffmann, J.~Weickert,
  \href{http://www.sciencedirect.com/science/article/pii/S0024379516300428}{Pseudo-inverses
  of difference matrices and their application to sparse signal approximation},
  Linear Algebra and its Applications 503 (2016) 26 -- 47.
\newblock \href {https://doi.org/https://doi.org/10.1016/j.laa.2016.03.033}
  {\path{doi:https://doi.org/10.1016/j.laa.2016.03.033}}.
\newline\urlprefix\url{http://www.sciencedirect.com/science/article/pii/S0024379516300428}

\bibitem{spa}
D.~L. Donoho, M.~Elad, \href{http://www.pnas.org/content/100/5/2197}{Optimally
  sparse representation in general (nonorthogonal) dictionaries via l1
  minimization}, Proceedings of the National Academy of Sciences 100~(5) (2003)
  2197--2202.
\newblock \href
  {http://arxiv.org/abs/http://www.pnas.org/content/100/5/2197.full.pdf}
  {\path{arXiv:http://www.pnas.org/content/100/5/2197.full.pdf}}, \href
  {https://doi.org/10.1073/pnas.0437847100}
  {\path{doi:10.1073/pnas.0437847100}}.
\newline\urlprefix\url{http://www.pnas.org/content/100/5/2197}

\bibitem{ylu}
Y.~M. Lu, M.~N. Do, A theory for sampling signals from a union of subspaces,
  IEEE Transactions on Signal Processing 56 (2008) 2334--2345.

\bibitem{hohn}
F.~Hohn, \href{https://books.google.gr/books?id=9XlFu4XQ6nUC}{Elementary Matrix
  Algebra}, Dover Books on Mathematics, Dover Publications, 2013.
\newline\urlprefix\url{https://books.google.gr/books?id=9XlFu4XQ6nUC}

\bibitem{recov}
M.~Zhao, M.~D. Kaba, R.~Vidal, D.~P. Robinson, E.~Mallada,
  \href{http://arxiv.org/abs/1803.09631}{Sparse recovery over graph incidence
  matrices: Polynomial time guarantees and location dependent performance},
  CoRR abs/1803.09631.
\newblock \href {http://arxiv.org/abs/1803.09631} {\path{arXiv:1803.09631}}.
\newline\urlprefix\url{http://arxiv.org/abs/1803.09631}

\bibitem{uos2}
T.~Blumensath, \href{https://doi.org/10.1109/TIT.2011.2146550}{Sampling and
  reconstructing signals from a union of linear subspaces}, {IEEE} Trans.
  Information Theory 57~(7) (2011) 4660--4671.
\newblock \href {https://doi.org/10.1109/TIT.2011.2146550}
  {\path{doi:10.1109/TIT.2011.2146550}}.
\newline\urlprefix\url{https://doi.org/10.1109/TIT.2011.2146550}

\bibitem{bar}
R.~G. Baraniuk, V.~Cevher, M.~F. Duarte, C.~Hegde, Model-based compressive
  sensing, IEEE Transactions on Information Theory 56~(4) (2010) 1982--2001.
\newblock \href {https://doi.org/10.1109/TIT.2010.2040894}
  {\path{doi:10.1109/TIT.2010.2040894}}.

\bibitem{demko}
S.~Demko, W.~F. Moss, P.~W. Smith,
  \href{http://www.jstor.org/stable/2008290}{Decay rates for inverses of band
  matrices}, Mathematics of Computation 43~(168) (1984) 491--499.
\newline\urlprefix\url{http://www.jstor.org/stable/2008290}

\end{thebibliography}

\end{document}